\documentclass[preprint,a4paper,11pt,onecolumn,numbers,sort&compress]{elsarticle}

\usepackage{graphicx}
\usepackage{epsfig}
\usepackage{natbib}

\usepackage{url}
\usepackage{amsmath,amsfonts,bm,mathrsfs,amssymb}
\usepackage{float}
\usepackage{subfig}
\usepackage{multirow}
\usepackage{array}
\usepackage{multicol}
\usepackage{color}
\usepackage{xcolor}
\usepackage[normalem]{ulem}

\DeclareMathOperator{\sgn}{sgn}
\newtheorem{theorem}{Theorem}

\newtheorem{lemma}[theorem]{Lemma}

\newtheorem{definition}[theorem]{Definition}
\newdefinition{remark}{Remark}
\newproof{proof}{Proof}

\journal{Journal of Mathematical Economics}

\usepackage{bm}
\usepackage{pgfplots}
\usepackage{tikz}
\usetikzlibrary{arrows,calc,shapes, snakes, intersections}

\usepgfplotslibrary{fillbetween}
\usetikzlibrary{patterns}

\begin{document}

\begin{frontmatter}

\title{On Optimal Mechanisms in the Two-Item Single-Buyer Unit-Demand Setting}  
\author[a1]{D.~Thirumulanathan\corref{cor1}}
\ead{thirumulanathan@gmail.com}
\author[a2]{Rajesh Sundaresan}
\ead{rajeshs@iisc.ac.in}
\author[a3]{Y Narahari}
\ead{narahari@iisc.ac.in}

\cortext[cor1]{Corresponding author}

\address[a1]{Department of Electrical Communication Engineering, Indian Institute of Science, Bengaluru, India 560012}
\address[a2]{Department of Electrical Communication Engineering, and Robert Bosch Centre for Cyber-Physical Systems, Indian Institute of Science, Bengaluru, India 560012}
\address[a3]{Department of Computer Science and Automation, Indian Institute of Science, Bengaluru, India 560012}

\begin{abstract}
We consider the problem of designing a revenue-optimal mechanism in the two-item, single-buyer, unit-demand setting when the buyer's valuations, $(z_1, z_2)$, are uniformly distributed in an arbitrary rectangle $[c,c+b_1]\times[c,c+b_2]$ in the positive quadrant. We provide a complete and explicit solution for arbitrary nonnegative values of $(c,b_1,b_2)$. We identify five simple structures, each with at most five (possibly stochastic) menu items, and prove that the optimal mechanism has one of the five structures. We also characterize the optimal mechanism as a function of $b_1, b_2$, and $c$. When $c$ is low, the optimal mechanism is a posted price mechanism with an exclusion region; when $c$ is high, it is a posted price mechanism without an exclusion region. Our results are the first to show the existence of optimal mechanisms with no exclusion region, to the best of our knowledge.
\end{abstract}


\end{frontmatter}

\section{Introduction}
This paper studies the design of revenue-optimal mechanism in the two-item, one-buyer, unit-demand setting. The solution to the problem is well known when the buyer's value is one-dimensional (\citet{Mye81}). The problem however becomes much harder when the buyer's value is multi-dimensional. Though many partial results are available in the literature, finding the general solution remains open in the two-item setting, be it with or without the unit-demand constraint.

In this paper, we consider the problem of optimal mechanism design in the two-item one-buyer unit-demand setting, when the valuations of the buyer are uniformly distributed in arbitrary rectangles in the positive quadrant having their left-bottom corners on the line $z_1=z_2$. Observe that this is a setting that occurs often in practice. As one example, consider a setting where two houses in a locality are sold. The seller is aware of a minimum and a maximum value for each house. Further, the buyer has a unit-demand, i.e., he can buy at most one of the houses, but submits his bids for both the houses. We consider that the buyer's valuations are uniform in the rectangle formed by those intervals. We compute the optimal mechanism for all cases when the minimum value for both the houses is the same. Another example is one where two sports team franchises in a sports league are sold to a potential buyer. The buyer needs at most one franchise, but submits his bids for both franchises.

\subsection{Prior work}

Consider the setting where the buyer is not restricted by the unit-demand constraint. \citet{DDT13,DDT14,DDT15} provided a solution when the buyer's valuation vector $z$ arises from a rich class of distribution functions each of which gives rise to a so-called ``well-formed canonical partition" of the support set of the distribution. The authors of these papers formulate this problem as an optimization problem, identify its dual as a problem of optimal transport, and exploit its solution to obtain a primal solution. \citet{GK14} computed the solution for the multi-item setting, but only when the valuations for each item are uniformly distributed in $[0,1]$. \citet{GK15} also provided closed form solutions in the two-item setting, when the distribution satisfies some sufficient conditions, by using a dual approach similar to \cite{DDT13,DDT15,DDT17}. In a companion paper \cite{TRN17a} (see also \cite{TRN16}), we used the same approach of solving the optimal transport problem as in \cite{DDT15} to obtain the solution when $z\sim\mbox{Unif}[c_1,c_1+b_1]\times[c_2,c_2+b_2]$ for arbitrary nonnegative values of $(c_1,c_2,b_1,b_2)$. The exact solution in the unrestricted setting has largely been computed using the dual approach designed in \cite{DDT15}.

The exact solution in the unit-demand setting, on the other hand, has been computed using various other methods. \citet{Pav11} obtained a solution both in the unrestricted setting and in the restricted setting of unit-demand constraint, when $z\sim\mbox{Unif}[c,c+1]^2$. The above paper used a marginal profit function $V$, whose properties are analogous to the virtual valuation function in \cite{Mye81}, to compute the exact solution. We thus call this method the {\em virtual valuation method}. The function however depends on the region of zero allocation, the exclusion region, and is thus not as straightforward to compute as the virtual valuation function in \cite{Mye81} for the single item case. \citet{Lev11} provided a solution for the unit-demand setting when the distribution is uniform in certain polygons aligned with the co-ordinate axes; the approach involves analyzing the utility function of the optimal mechanism at the edges of the polygon. \citet{KM16} identified the dual when the valuation space is convex and the space of allocations is restricted. They also solved examples when the allocations are restricted to satisfy either the unit-demand constraint or the deterministic constraint. Other than this lone example solved in \cite{KM16}, we are not aware of any work that computes the exact solution in the unit-demand setting using the duality approach.

There are interesting characterization results on optimal mechanisms in the unit-demand setting. \citet{WT14}, \cite{TW17} proved that when the distributions are uniform in any rectangle in the positive quadrant, the optimal mechanism is a {\em menu} with at most five items. However, the exact menus and associated allocations were left open. \citet{HH15} did a reverse mechanism design; they constructed a mechanism and identified conditions under which there exists a virtual valuation thereby establishing that the mechanism is optimal.

There has been some interest in finding approximately optimal solutions when the distribution of the buyer's valuations satisfies certain conditions. See \cite{Bhat10, BCKW10, BCKW15, CD11, CD15, CDW12a, CDW12b, CDW13, CDW16, CZ17, CHK07, CMS15, CM16}, \cite{DW11}, \cite{DW12}, \cite{Yao14} for relevant literature on approximate solutions. In this paper however we shall focus on exact solutions.

\subsection{Our contributions}
Our contributions are as follows:
\begin{enumerate}
\item[(i)] We identify the dual to the problem of optimal auction in the restricted unit-demand setting, using a result in \cite{KM16}\footnote{The dual to the problem of optimal auction was derived independently in the PhD thesis of the first author.}. We then argue that the computation of the dual measure in the unit-demand setting using the approach of optimal transport in \cite{DDT15} is intricate. Specifically, we consider three examples, $z\sim\mbox{Unif}[1.26,2.26]^2$, $z\sim\mbox{Unif}[1.5,2.5]^2$, and $z\sim\mbox{Unif}[0,1]\times[0,1.2]$, and show that the optimal dual variable differs significantly with variation in $c$, thus making it hard to discover the correct dual measure.
\item[(ii)] Motivated by the above, we explore the virtual valuation method in \cite{Pav11} and nontrivially extend this method to compute the exact solution when $z\sim\mbox{Unif}[c,c+b_1]\times[c,c+b_2]$, for arbitrary nonnegative values of $(c,b_1,b_2)$. We establish that the structure of the optimal mechanism falls within a class of five simple structures, each having at most five constant allocation regions. We also make some remarks on the general case $[c_1,c_1+b_1]\times[c_2,c_2+b_2]$.
\item[(iii)] To the best of our knowledge, our results appear to be the first to show the existence of optimal mechanisms with no region of exclusion (see Figures \ref{fig:e-ini} and \ref{fig:e'-ini}). The results in \citet{Arm96a} and \citet{PBBK14} assert that the optimal multi-dimensional mechanisms have a nontrivial exclusion region under some sufficient conditions on the distributions and the utility functions. \citet{Arm96a} assumes strict convexity of the support set, and \citet{PBBK14} assume strict concavity of the utility function in the allocations. Neither of these assumptions holds in our setting.
\end{enumerate}

In the literature, we already have qualitative results on the structure of optimal mechanism for distributions satisfying certain conditions. For instance, Pavlov \cite{Pav10} considered distributions with negative power rate, while Wang and Tang \cite{WT14} considered uniform distributions on arbitrary rectangles (which do have negative power rate). Our work considers uniform distributions with support set $[c,c+b_1]\times[c,c+b_2]$, a special case of the settings in \cite{Pav10} and \cite{WT14}. It follows that the optimal mechanisms in our setting can have allocations only of the form $(0,0)$, $(a,1-a)$ in accordance with Pavlov's result, and the menus can have at most five items in accordance with Wang and Tang's result.

Though our work is on a further special case, we are able to obtain finer results. We prove that the optimal mechanisms can only be one among the structures depicted in Figures \ref{fig:a-ini}--\ref{fig:e'-ini}. Our results bring out some unexpected structures such as those in Figures \ref{fig:e-ini} and \ref{fig:e'-ini}. Furthermore, our results are explicit in that we can compute the optimal mechanism for uniform distributions on any rectangle of the form $[c,c+b_1]\times[c,c+b_2]$.

The optimal mechanisms for various values of $(c,b_1,b_2)$ are mentioned in Theorem \ref{thm:consolidate}. The phase diagram in Figure \ref{fig:phase-diagram} represents how the structure of optimal mechanism changes when the values of $(c,b_1,b_2)$ change. We interpret the solutions and highlight their features as follows.

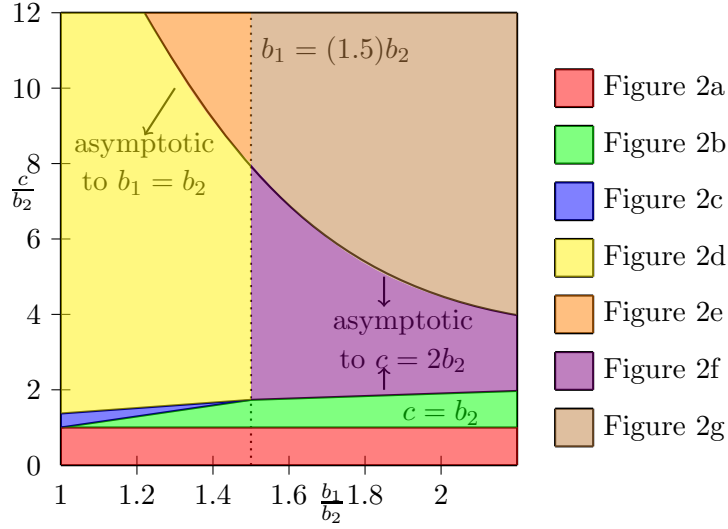
\begin{figure}[t]
\centering
\begin{tikzpicture}[scale=0.5,font=\normalsize,axis/.style={very thick, ->, >=stealth'}]
\draw [axis,thick,-] (0,0)--(0,12);
\node [right] at (6.5,-1) {$\frac{b_1}{b_2}$};
\draw [axis,thick,-] (0,0)--(12,0);
\node [above] at (-1,6.5) {$\frac{c}{b_2}$};
\draw [axis,thick,-] (12,0)--(12,12);
\draw [axis,thick,-] (0,12)--(12,12);
\node [left] at (-0.25,0) {$0$};
\draw [thin,-] (-0.25,2) -- (0.25,2);
\node [left] at (-0.25,2) {$2$};
\draw [thin,-] (-0.25,4) -- (0.25,4);
\node [left] at (-0.25,4) {$4$};
\draw [thin,-] (-0.25,6) -- (0.25,6);
\node [left] at (-0.25,6) {$6$};
\draw [thin,-] (-0.25,8) -- (0.25,8);
\node [left] at (-0.25,8) {$8$};
\draw [thin,-] (-0.25,10) -- (0.25,10);
\node [left] at (-0.25,10) {$10$};
\draw [thin,-] (-0.25,12) -- (0.25,12);
\node [left] at (-0.25,12) {$12$};
\node [below] at (0,-0.25) {$1$};
\draw [thin,-] (2,-0.25) -- (2,0.25);
\node [below] at (2,-0.25) {$1.2$};
\draw [thin,-] (4,-0.25) -- (4,0.25);
\node [below] at (4,-0.25) {$1.4$};
\draw [thin,-] (6,-0.25) -- (6,0.25);
\node [below] at (6,-0.25) {$1.6$};
\draw [thin,-] (8,-0.25) -- (8,0.25);
\node [below] at (8,-0.25) {$1.8$};
\draw [thin,-] (10,-0.25) -- (10,0.25);
\node [below] at (10,-0.25) {$2$};
\draw [thick,-] (0,1) to (12,1);
\draw [thick,-] (0,1) to (5,1.733) to (12,1.97);
\draw [thick,-] (0,1.372) to (5,1.733);
\draw [thick,-] (2.2,12) to[out=-60,in=170] (12,3.98);
\draw [thick,dotted] (5,0) to (5,12);

\node [right] at (5,11) {$b_1=(1.5)b_2$};
\draw [thick,->] (3,10) -- (2.2,8.75);
\node at (2.2,8.5) {asymptotic};
\node at (2.2,7.5) {to $b_1=b_2$};
\node at (9,3.8) {asymptotic};
\node at (9,2.8) {to $c=2b_2$};
\node at (10,1.4) {$c=b_2$};
\draw [thick,->] (8.5,2) -- (8.5,2.6);
\draw [thick,->] (8.5,5) -- (8.5,4.2);

\path[fill=red,opacity=.5] (0,1) to (12,1) to (12,0) to (0,0);
\path[fill=green,opacity=.5] (0,1) to (5,1.733) to (12,1.97) to (12,1) to (0,1);
\path[fill=blue,opacity=.5] (0,1) to (0,1.372) to (5,1.733) to (0,1);
\path[fill=yellow,opacity=.5] (0,1.372) to (0,12) to (2.2,12) to [out=-60,in=135] (5,8) to (5,1.733) to (0,1.372);
\path[fill=violet,opacity=.5] (5,8) to [out=-50,in=170] (12,3.98) to (12,1.97) to (5,1.733) to (5,8);
\path[fill=orange,opacity=.5] (2.2,12) to [out=-60,in=135] (5,8) to (5,12) to (2.2,12);
\path[fill=brown,opacity=.5] (5,12) to (5,8) to [out=-50,in=170] (12,3.98) to (12,12) to (5,12);

\draw [thick,-] (13,9.5) to (13,10.5);
\draw [thick,-] (13,10.5) to (14,10.5);
\draw [thick,-] (14,9.5) to (14,10.5);
\draw [thick,-] (13,9.5) to (14,9.5);
\path[fill=red,opacity=.5] (13,9.5) to (14,9.5) to (14,10.5) to (13,10.5) to (13,9.5);
\node [right] at (14,10) {Figure \ref{fig:a-ini}};

\draw [thick,-] (13,8) to (13,9);
\draw [thick,-] (13,9) to (14,9);
\draw [thick,-] (14,8) to (14,9);
\draw [thick,-] (13,8) to (14,8);
\path[fill=green,opacity=.5] (13,8) to (14,8) to (14,9) to (13,9) to (13,8);
\node [right] at (14,8.5) {Figure \ref{fig:b-ini}};

\draw [thick,-] (13,6.5) to (13,7.5);
\draw [thick,-] (13,7.5) to (14,7.5);
\draw [thick,-] (14,6.5) to (14,7.5);
\draw [thick,-] (13,6.5) to (14,6.5);
\path[fill=blue,opacity=.5] (13,6.5) to (14,6.5) to (14,7.5) to (13,7.5) to (13,7.5);
\node [right] at (14,7) {Figure \ref{fig:c-ini}};

\draw [thick,-] (13,5) to (13,6);
\draw [thick,-] (13,6) to (14,6);
\draw [thick,-] (14,5) to (14,6);
\draw [thick,-] (13,5) to (14,5);
\path[fill=yellow,opacity=.5] (13,5) to (14,5) to (14,6) to (13,6) to (13,5);
\node [right] at (14,5.5) {Figure \ref{fig:d-ini}};

\draw [thick,-] (13,3.5) to (13,4.5);
\draw [thick,-] (13,4.5) to (14,4.5);
\draw [thick,-] (14,3.5) to (14,4.5);
\draw [thick,-] (13,3.5) to (14,3.5);
\path[fill=orange,opacity=.5] (13,3.5) to (14,3.5) to (14,4.5) to (13,4.5) to (13,3.5);
\node [right] at (14,4) {Figure \ref{fig:e-ini}};

\draw [thick,-] (13,2) to (13,3);
\draw [thick,-] (13,3) to (14,3);
\draw [thick,-] (14,2) to (14,3);
\draw [thick,-] (13,2) to (14,2);
\path[fill=violet,opacity=.5] (13,2) to (14,2) to (14,3) to (13,3) to (13,2);
\node [right] at (14,2.5) {Figure \ref{fig:d'-ini}};

\draw [thick,-] (13,0.5) to (13,1.5);
\draw [thick,-] (13,1.5) to (14,1.5);
\draw [thick,-] (14,0.5) to (14,1.5);
\draw [thick,-] (13,0.5) to (14,0.5);
\path[fill=brown,opacity=.5] (13,0.5) to (14,0.5) to (14,1.5) to (13,1.5) to (13,0.5);
\node [right] at (14,1) {Figure \ref{fig:e'-ini}};
\end{tikzpicture}
\caption{A phase diagram of the optimal mechanism when $b_2\leq b_1\leq (2 .2)b_2$.}\label{fig:phase-diagram}
\end{figure}

\begin{itemize}
\item[$\bullet$] Beyond the exclusion (no sale) region, the allocation probabilities are the same for all $z$ falling in the same $45^\circ$ line (Theorem \ref{thm:pav-2}). Observe that this is in sharp contrast with the unrestricted setting, where the allocation probabilities are the same either for all $z$ falling in the same vertical line or the same horizontal line (see \cite[Fig.~1--3]{TRN16}). This is because, in the unit-demand case, the buyer demands at most one of the two items, and thus the seller decides the item to be sold based on the difference of valuations on the items\footnote{The item to be sold is decided based on the difference in valuations only for cases where $q_1 + q_2 = 1$ holds everywhere outside the exclusion region. It would be interesting to interpret the results for cases when $q_1 + q_2 < 1$ can occur outside the exclusion region, but this exploration is beyond the scope of this paper.}.
\item[$\bullet$] Consider the case when $c$ is low. The seller then knows that the buyer possibly could have very low valuations, and thus sets a high price $(c+\delta_i)$ to sell item $i$. Observe that this is a posted price mechanism with prices $c+\delta_1$ and $c+\delta_2$ for items $1$ and $2$ respectively (see Figure \ref{fig:a-ini}).
\item[$\bullet$] When $c$ increases, the seller now finds it optimal to set a second price over and above the first price $c+\delta_i$. He offers a lottery for the first price, and offers an individual item for the second and higher price (see Figures \ref{fig:b-ini} and \ref{fig:c-ini}).
\item[$\bullet$] When $c$ increases further, the seller sells item $i$ only when $z_i$ is very high compared to $z_{-i}$. In case the difference is not sufficiently high, then the seller finds it optimal to allocate randomly one or the other item (see Figures \ref{fig:d-ini} and \ref{fig:d'-ini}).
\item[$\bullet$] When $c$ is very high, the revenue gained by exclusion of certain valuations is always dominated by the revenue lost by it, and thus the seller finds no reason to withhold the items for any valuation profile\footnote{We refer the reader to Remark \ref{rem:armstrong} for a more precise explanation.}. So the optimal mechanism turns out to be a posted price mechanism with prices $c+\frac{b_1}{3}+\max(0,\frac{b_1}{6}-\frac{b_2}{4})$ and $c$ for items $1$ and $2$ respectively. In effect, it is a posted price mechanism with no exclusion region (see Figures \ref{fig:e-ini} and \ref{fig:e'-ini}).
\item[$\bullet$] Starting at $c=0$, consider moving the support set rectangle to infinity. Then, the optimal mechanism starts as a posted price mechanism with an exclusion region, and ends up again as a posted price mechanism but without an exclusion region. The other structures in Figures \ref{fig:b-ini}--\ref{fig:d-ini}, and Figure \ref{fig:d'-ini} are optimal for various intermediate values.
\end{itemize}

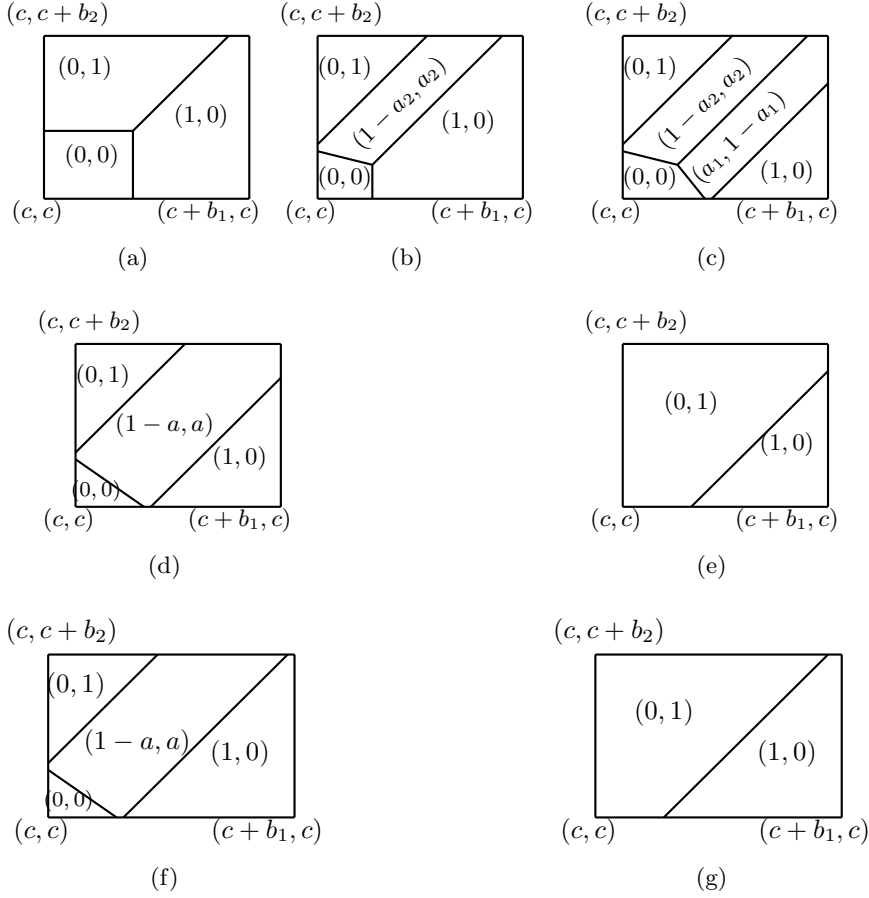
\begin{figure}[t]
\centering
\setlength\tabcolsep{0pt}
\begin{tabular}{cccccc}
\multicolumn{2}{c}{\subfloat[]{\label{fig:a-ini}\begin{tikzpicture}[scale=0.18,font=\footnotesize,axis/.style={very thick, -}]
\node at (-0.5,-1) {$(c,c)$};
\draw [axis,thick,-] (0,0)--(15,0);
\node at (12,-1) {$(c+b_1,c)$};
\draw [axis,thick,-] (0,0)--(0,12);
\node [above] at (1,12) {$(c,c+b_2)$};
\draw [axis,thick,-] (0,12)--(15,12);
\draw [axis,thick,-] (15,0)--(15,12);
\draw [axis,thick,-] (0,5)--(6.5,5);
\draw [axis,thick,-] (6.5,0)--(6.5,5);
\draw [thick,-] (6.5,5)--(13.5,12);
\node [above] at (3.5,1.5) {$(0,0)$};
\node [above] at (11.5,4.5) {$(1,0)$};
\node [above] at (3,8) {$(0,1)$};
\end{tikzpicture}}}&
\multicolumn{2}{c}{\subfloat[]{\label{fig:b-ini}\begin{tikzpicture}[scale=0.18,font=\footnotesize,axis/.style={very thick, -}]
\node at (-0.5,-1) {$(c,c)$};
\draw [axis,thick,-] (0,0)--(15,0);
\node at (12,-1) {$(c+b_1,c)$};
\draw [axis,thick,-] (0,0)--(0,12);
\node [above] at (1,12) {$(c,c+b_2)$};
\draw [axis,thick,-] (0,12)--(15,12);
\draw [axis,thick,-] (15,0)--(15,12);
\draw [axis,thick,-] (0,4)--(8,12);
\draw [axis,thick,-] (0,3.5)--(4,2.5);
\draw [axis,thick,-] (4,0)--(4,2.5);
\draw [thick,-] (4,2.5)--(13.5,12);
\node at (2,1.5) {$(0,0)$};
\node [above] at (11,4) {$(1,0)$};
\node [above] at (2,8) {$(0,1)$};
\node [rotate=45] at (6,7) {$(1-a_2,a_2)$};
\end{tikzpicture}}}&
\multicolumn{2}{c}{\subfloat[]{\label{fig:c-ini}\begin{tikzpicture}[scale=0.18,font=\footnotesize,axis/.style={very thick, -}]
\node at (-0.5,-1) {$(c,c)$};
\draw [axis,thick,-] (0,0)--(15,0);
\node at (12,-1) {$(c+b_1,c)$};
\draw [axis,thick,-] (0,0)--(0,12);
\node [above] at (1,12) {$(c,c+b_2)$};
\draw [axis,thick,-] (0,12)--(15,12);
\draw [axis,thick,-] (15,0)--(15,12);
\draw [axis,thick,-] (0,4)--(8,12);
\draw [axis,thick,-] (0,3.5)--(4,2.5);
\draw [axis,thick,-] (4,2.5)--(6,0);
\draw [axis,thick,-] (6.5,0)--(15,8.5);
\draw [thick,-] (4,2.5)--(13.5,12);
\node at (2,1.5) {$(0,0)$};
\node at (12,2) {$(1,0)$};
\node [above] at (2,8) {$(0,1)$};
\node [rotate=45] at (6,7) {$(1-a_2,a_2)$};
\node [rotate=45] at (8.5,4.5) {$(a_1,1-a_1)$};
\end{tikzpicture}}}\\
&\multicolumn{2}{c}{\subfloat[]{\label{fig:d-ini}\begin{tikzpicture}[scale=0.18,font=\footnotesize,axis/.style={very thick, -}]
\node at (-0.5,-1) {$(c,c)$};
\draw [axis,thick,-] (0,0)--(15,0);
\node at (12,-1) {$(c+b_1,c)$};
\draw [axis,thick,-] (0,0)--(0,12);
\node [above] at (1,12) {$(c,c+b_2)$};
\draw [axis,thick,-] (0,12)--(15,12);
\draw [axis,thick,-] (15,0)--(15,12);
\draw [axis,thick,-] (0,4)--(8,12);
\draw [axis,thick,-] (0,3.5)--(5,0);
\draw [thick,-] (5.5,0)--(15,9.5);
\node at (1.5,1.2) {\scriptsize$(0,0)$};
\node [above] at (12,2) {$(1,0)$};
\node [above] at (2,8) {$(0,1)$};
\node at (6.5,6) {$(1-a,a)$};
\end{tikzpicture}}}&
&\multicolumn{2}{c}{\subfloat[]{\label{fig:e-ini}\begin{tikzpicture}[scale=0.18,font=\footnotesize,axis/.style={very thick, -}]
\node at (-0.5,-1) {$(c,c)$};
\draw [axis,thick,-] (0,0)--(15,0);
\node at (12,-1) {$(c+b_1,c)$};
\draw [axis,thick,-] (0,0)--(0,12);
\node [above] at (1,12) {$(c,c+b_2)$};
\draw [axis,thick,-] (0,12)--(15,12);
\draw [axis,thick,-] (15,0)--(15,12);
\draw [thick,-] (5,0)--(15,10);
\node [above] at (12,3) {$(1,0)$};
\node [above] at (5,6) {$(0,1)$};
\end{tikzpicture}}}\\
&\multicolumn{2}{c}{\subfloat[]{\label{fig:d'-ini}\begin{tikzpicture}[scale=0.18,font=\small,axis/.style={very thick, -}]
\node at (-0.5,-1) {$(c,c)$};
\draw [axis,thick,-] (0,0)--(18,0);
\node at (16,-1) {$(c+b_1,c)$};
\draw [axis,thick,-] (0,0)--(0,12);
\node [above] at (1,12) {$(c,c+b_2)$};
\draw [axis,thick,-] (0,12)--(18,12);
\draw [axis,thick,-] (18,0)--(18,12);
\draw [axis,thick,-] (0,4)--(8,12);
\draw [axis,thick,-] (0,3.5)--(5,0);
\draw [thick,-] (5.5,0)--(17.5,12);
\node at (1.5,1.2) {\scriptsize$(0,0)$};
\node [above] at (14,3) {$(1,0)$};
\node [above] at (2,8) {$(0,1)$};
\node at (6.5,5.5) {$(1-a,a)$};
\end{tikzpicture}}}&
&\multicolumn{2}{c}{\subfloat[]{\label{fig:e'-ini}\begin{tikzpicture}[scale=0.18,font=\small,axis/.style={very thick, -}]
\node at (-0.5,-1) {$(c,c)$};
\draw [axis,thick,-] (0,0)--(18,0);
\node at (16,-1) {$(c+b_1,c)$};
\draw [axis,thick,-] (0,0)--(0,12);
\node [above] at (1,12) {$(c,c+b_2)$};
\draw [axis,thick,-] (0,12)--(18,12);
\draw [axis,thick,-] (18,0)--(18,12);
\draw [thick,-] (5,0)--(17,12);
\node [above] at (14,3) {$(1,0)$};
\node [above] at (5,6) {$(0,1)$};
\end{tikzpicture}}}
\end{tabular}
\caption{An illustration of all possible structures that an optimal mechanism can have.}\label{fig:all-structure-ini}
\end{figure}

\subsection{Our method}

Our method is as follows. We initially formulate the problem at hand (in the unit-demand setting) into an optimization problem, and compute its dual using a result in \cite{KM16}. The dual problem turns out to be an optimal transport problem that transfers mass from the support set $D$ to itself. Mass transfer must occur subject to the constraint that the difference between the mass densities transferred out of and transferred into the set convex-dominates a signed measure that depends only on the distribution of the valuations. The dual problem is similar to that in \cite{DDT15} for the unrestricted setting, but differs in the transportation cost.

The key challenge in solving the dual problem lies in constructing the  ``shuffling measure'' that convex-dominates $0$, and in finding the location in the support set $D$ where the shuffling measure sits. The shuffling measure was always added at fixed locations in the unrestricted setting, and had a fixed structure for the uniform distribution of valuations over any rectangle in the positive quadrant (see \cite{TRN16}). In the unit-demand setting, however, we see that both the locations and the structures of the shuffling measure vary significantly for different values of $c$. There is as yet no clear understanding on how to construct the shuffling measure, and hence on how to compute the optimal solution via the dual method.

Motivated by the above, we explore the virtual valuation method used by \citet{Pav11}. \citet{Pav11} computed the optimal mechanism when the buyer's valuations are given by $z\sim\mbox{Unif}[c,c+1]^2$; the optimal mechanism was obtained only for distributions that are symmetric across the two items. When compared with the case of symmetric distributions, the case of asymmetric distributions poses the following challenge. The optimal mechanism is symmetric along a diagonal in the case of symmetric distributions. For asymmetric distributions, the mechanism must be computed over the larger region of the entire support set. The asymmetry leads to more parameters, more conditions to check for optimality, and a more complex variety of solutions determined, as we will soon see, by a larger number of polynomials. All these make the computation more difficult.

In this paper, we demonstrate how to compute the optimal mechanism for asymmetric distributions, when $z\sim\mbox{Unif}[c,c+b_1]\times[c,c+b_2]$. Specifically, we do the following.

\begin{itemize}
\item[$\bullet$] Taking cue from the result in \cite{WT14} that the optimal mechanism is a menu with at most five items, we first construct some possible menus, parametrized by at most four parameters.
\item[$\bullet$] We find the relation between the parameters using the sufficient conditions on the marginal profit function $V$. We show that the parameters can be computed by simultaneously solving at most two polynomials, each of degree at most $4$.
\item[$\bullet$] We then use continuity of the polynomials to prove that there exists a solution having desired values for all parameters. We then prove that the optimal mechanism has one of the five simple structures for arbitrary nonnegative values of $(c,b_1,b_2)$ (see Theorem \ref{thm:consolidate}).
\item[$\bullet$] We conjecture that the optimal mechanisms have a similar structure even when $z\in[c_1,c_1+b_1]\times[c_2,c_2+b_2]$ for all $(c_1,c_2,b_1,b_2)\geq 0$. We provide preliminary results to justify the conjecture (see Theorem \ref{thm:extension}).
\end{itemize}

Proofs of some case use Mathematica to verify certain algebraic inequalities. This is because (i) the parameters turn out to be solutions that simultaneously satisfy two polynomials of degree at most $4$; and (ii) the solutions are complicated functions of $(c,b_1,b_2)$ involving fifth roots and eighth roots of some expressions. Verifying that these expressions satisfy some bounds were automated via the Mathematica software. The results that use Mathematica have been marked with an asterisk in the statement of Theorem \ref{thm:consolidate}. We believe that all of these results can be proved in the strict mathematical sense; but we leave this for the future in the interest of timely dissemination of our conclusions and observations. The skeptical reader could proceed by interpreting the Mathematica-based conclusions as conjectures.

Our work thus provides insights into two well-known approaches to solve representative problems on optimal mechanisms in the multi-item setting, besides solving, in the process, one such problem for asymmetric distributions. Specifically, our work clarifies under what situations the duality approach is likely to work well, and the intrinsic difficulties in using that approach in some other settings. Furthermore, the special cases that we solve provide insights into various possible structures of the optimal mechanisms which, we feel, would act as a guideline to solve the problem of computing good menus in practical settings. We believe that our work is an important step towards understanding the applicability of the two different approaches, and a useful step addition to the growing canvas of canonical problems in multi-dimensional optimal auctions. 

The rest of the paper is organized as follows. In Section 2, we first formulate an optimization problem under the unit-demand setting. We next compute its dual using a result in \cite{KM16}, and solve it for three representative examples of $(c,b_1,b_2)$. The main purpose behind these examples is to bring out the variety in structure, and therefore the difficulty in guessing and computing, the dual measure for more general settings. In Section 3, we nontrivially extend the {\em virtual valuation method\/} of \cite{Pav11} to provide a complete and explicit solution for the case of asymmetric distributions. In particular, we prove that the optimal mechanism has one of the five simple structures. In Section 4, we conjecture, with promising preliminary results, that the optimal mechanism when the valuations are uniformly distributed in an arbitrary rectangle $[c_1,c_1+b_1]\times[c_2,c_2+b_2]$ also has similar structures. In Section 5, we conclude the paper and provide some directions for future work.

\section{Exploring The Dual Approach}\label{sec:prelim}
Consider a two-item, single-buyer, unit-demand setting. The buyer's valuation is $z=(z_1,z_2)$ for the two items, sampled according to the joint density $f(z)=f_1(z_1)f_2(z_2)$, where $f_1(z_1)$ and $f_2(z_2)$ are marginal densities. The support set of $f$ is defined as $D:=\{z:f(z)>0\}$. Throughout the paper, we consider $D=[c,c+b_1]\times[c,c+b_2]$, where $(c,b_1,b_2)$ are nonnegative.

Our aim is to design a revenue-optimal mechanism. By the revelation principle \cite[Prop.~9.25]{NRTV07}, it suffices to focus only on direct mechanisms. Further, we focus on mechanisms where the buyer has a quasilinear utility. Specifically, we assume an allocation function $q:D\rightarrow\{(q_1,q_2):0\leq q_1,q_2,q_1+q_2\leq 1\}$ and a payment function $t:D \rightarrow \mathbb{R}_+$ that represent, respectively, the probabilities of allocation of the items to the buyer and the amount of transfer from the buyer to the seller. If the buyer's true valuation is $z$, and he reports $\hat{z}$, his realized (quasilinear) utility is $\hat{u}(z, \hat{z}) := z \cdot q(\hat{z}) - t(\hat{z})$, which is the expected value of the lottery he receives minus the payment.

A mechanism $(q,t)$ satisfies {\em incentive compatibility} (IC) when truth telling is a weakly dominant strategy for the buyer, i.e., $\hat{u}(z,z) \geq \hat{u}(z,\hat{z})$ for every $z,\hat{z}\in D$. In this case the buyer's realized utility is $u(z) := \hat{u}(z, z) = z\cdot q(z)-t(z)$. An incentive compatible mechanism satisfies {\em individual rationality} (IR) if the buyer is not worse off by participating in the mechanism, i.e., $u(z) \geq 0$ for every $z \in D$, with zero being the buyer's utility if he chooses not to participate.

The following result is well known:
\begin{theorem}\cite{Roc87}\/.\label{thm:Rochet}
A mechanism $(q,t)$, with $u(z) = z\cdot q(z)-t(z)$, is incentive compatible if and only if $u$ is continuous, convex and $\nabla u(z)=q(z)$ for a.e. $z\in D$.
\end{theorem}

An optimal mechanism is one that maximizes the expected revenue to the seller subject to incentive compatibility and individual rationality (\citet[p.~67]{VKBook}). By virtue of Theorem \ref{thm:Rochet}, an optimal mechanism solves the problem
\begin{alignat*}{2}
  &\max_u \int_D(z\cdot\nabla u(z)-u(z))f(z)\,dz \\
  \mbox{ subject to }\hspace*{.5in}
  &(a)\,u\mbox{ continuous, convex},\nonumber\\
  &(b)\,\nabla u(z)\in[0,1]^2,\nabla u(z)\cdot\mathbf{1}\in[0,1],\mbox{ a.e. }z \in D,\nonumber\\
  &(c)\,u(z)\geq 0,\,\,\forall z\in D.\nonumber
\end{alignat*}
Using the arguments in \cite[Sec.~2.1]{DDT17}, we simplify the aforementioned problem as
\begin{align}\label{eqn:optim}
  &\max_u\int_D(z\cdot\nabla u(z)-(u(z)-u(c,c)))f(z)\,dz \\
  \mbox{ subject to }\hspace*{.5in}
  &(a)\,u\mbox{ continuous, convex},\nonumber\\
  &(b)\,\nabla u(z)\in[0,1]^2,\nabla u(z)\cdot\mathbf{1}\in[0,1],\mbox{ a.e. }z \in D.\nonumber
\end{align}

We now further simplify the objective function of the problem. Using integration by parts, the objective function can be written as $\int_D u(z)\mu(z)\,dz+\int_{\partial D}u(z)\mu_s(z)\,dz+u(c,c)\mu_p(c,c)$, where the functions $\mu$, $\mu_s$, and $\mu_p$ are defined as
\begin{align}\label{eqn:mu-measure}
 \mu(z)&:=-z\cdot\nabla f(z)-3f(z),\,z\in D,\nonumber\\
 \mu_s(z)&:=(z\cdot n(z))f(z),\,z\in\partial D,\\
 \mu_p(z)&:=\delta_{\{(c,c)\}}(z).\nonumber
\end{align}
The vector $n(z)$ is the normal to the surface $\partial D$ at $z$ if it is defined, and $0$ otherwise (at corners). We regard $\mu$ as the density of a signed measure on the support set $D$ that is absolutely continuous with respect to (w.r.t.) the two-dimensional Lebesgue measure, and $\mu_s$ as the density of a signed measure on $\partial D$ that is absolutely continuous w.r.t. the surface Lebesgue measure. We shall refer to both Lebesgue measures as $dz$. We regard $\mu_p$ as a point measure of unit mass at the specified point. The notation $\delta$ denotes the Dirac-delta function. So $\mu_p(z)=1$ if $z=(c,c)$, and $0$ otherwise. By taking $u(z) = 1 ~ \forall z \in D$, we observe that
\begin{align}\label{eqn:mu-D}
  &\int_D \mu(z)\,dz+\int_{\partial D}\mu_s(z)\,dz+\mu_p(c,c)\nonumber\\&\hspace*{.5in}=\int_D u(z)\mu(z)\,dz+\int_{\partial D}u(z)\mu_s(z)\,dz+u(c,c)\mu_p(c,c)\nonumber\\&\hspace*{.5in}=\int_D(z\cdot\nabla u(z)-u(z))f(z)\,dz+u(c,c)\nonumber\\&\hspace*{.5in}=\int_D(0-1)f(z)\,dz+u(c,c)=0.
\end{align}
We now define the measure $\bar{\mu}$, supported on set $D$, as
$$
\bar{\mu}(A):=\int_D\mathbf{1}_A(z)\mu(z)\,dz+\int_{\partial D}\mathbf{1}_A(z)\mu_s(z)\,dz+\mu_p(A\cap(c,c))
$$
for all measurable sets $A$. We thus observe that $\bar{\mu}(D)=0$. Observe that $\bar{\mu}$ is a signed Radon measure in $D$, and that the functions $\mu$ and $\mu_s$ are just the Radon-Nikodym derivatives of the respective components of $\bar{\mu}$ w.r.t. the two-dimensional and one-dimensional Lebesgue measures respectively. Based on the discussion in the paragraph after (\ref{eqn:optim}), the objective function of problem (\ref{eqn:optim}) can now be written as $\int_D u\,d\bar{\mu}$.

We now rewrite the constraint (b) in problem (\ref{eqn:optim}) as the following three constraints.
\begin{alignat*}{2}
  &u(z_1,z_2)-u(z_1',z_2)\leq (z_1-z_1')_+,\, \forall z_1,z_1'\in D_1,\,\forall z_2\in D_2,\nonumber\\
  &u(z_1,z_2)-u(z_1,z_2')\leq (z_2-z_2')_+,\, \forall z_1\in D_1,\,\forall z_2,z_2'\in D_2,\nonumber\\
  &u(z_1,z_2)-u(z_1',z_2-z_1+z_1')\leq (z_1-z_1')_+,\, \forall z_1,z_1'\in D_1,\,\forall z_2\in D_2,\nonumber
\end{alignat*}
where $(\cdot)_+=\max(0,\cdot)$. Observe that these three constraints are equivalent to
$$
  u(z)-u(z')\leq\max((z_1-z_1')_+,(z_2-z_2')_+),\,\forall z,z'\in D.
$$

So the optimization problem can now be written as
\begin{alignat}{2}\label{eqn:primal}
  &\max_u \int_D u\,d\bar{\mu} \\
  \mbox{ subject to }\hspace*{.5in}
  &(a)\,u\mbox{ continuous, convex, increasing},\nonumber\\
  &(b)\,u(z)-u(z')\leq\|z-z'\|_\infty,\,\forall z,z'\in D.\nonumber
\end{alignat}
Note that the objective function of the problem satisfies $\int_Dt(z)f(z)\,dz=\int_Du\,d\bar{\mu}$; thus the $\bar{\mu}$-measure can be interpreted as the marginal contribution of the utility $u$ to the revenue of the seller.

We now recall the definition of the convex ordering relation. A function $f$ is increasing if $z\geq z'$ component-wise implies $f(z)\geq f(z')$.
\begin{definition} (See for e.g., \cite{DDT14})
 Let $\alpha$ and $\beta$ be measures defined on a set $D$. We say $\alpha$ {\em convex-dominates} $\beta$ ($\alpha\succeq_{cvx}\beta$) if $\int_Df\,d\alpha\geq\int_Df\,d\beta$ for all continuous, convex and increasing $f$.
\end{definition}

One can understand convex dominance as follows: A {\em risk-seeking} buyer\footnote{This is to be contrasted with second-order stochastic dominance which says that $\alpha$ second-order dominates $\beta$ (denoted as $\alpha\succeq_2\beta$) if a {\em risk-averse} buyer with an {\em increasing and concave} utility function prefers $\alpha$ to $\beta$. Mathematically, convex dominance and second-order stochastic dominance are related inversely under some conditions. More specifically, $\alpha\succeq_{cvx}\beta\Leftrightarrow\alpha\preceq_2\beta$ if (i) $D$ is a bounded rectangle in the positive orthant and (ii) $\int_D\|x\|_1\,d(\alpha-\beta)=0$ \cite[Lem. 8]{DDT14}}., with $u$ as his utility function (increasing and convex), will choose the lottery $\alpha$ over $\beta$ if $\alpha$ convex-dominates $\beta$.
 
The dual problem of (\ref{eqn:primal}) is found to be \cite[Thm.~3.1]{KM16}.
\begin{align}\label{eqn:dual}
  &\min_\gamma\int_{D\times D}\|z-z'\|_\infty\,d\gamma(z,z')\\
  \mbox{ subject to }\hspace*{.25in}
  &(a)\,\gamma\in Radon_+(D\times D),\nonumber\\
  &(b)\,\gamma(\cdot,D)=\gamma_1,\,\gamma(D,\cdot)=\gamma_2,\,\gamma_1-\gamma_2\succeq_{cvx}\bar{\mu}.\nonumber
 \end{align}
By $\gamma\in Radon_+(D\times D)$, we mean that $\gamma$ is an unsigned Radon measure in $D\times D$. The dual is computed by using the following expressions in the statement of \cite[Thm.~3.1]{KM16}: (i) $l_S(z,z')=\|z-z'\|_\infty$, and (ii) $U(D,S)$ is the set of all utility functions that are continuous, convex, and increasing. We derive the weak duality result in \ref{app:new} to provide an understanding of how the dual arises and why $\gamma$ may be interpreted as prices for violating the primal constraint.


The next lemma gives a sufficient condition for strong duality.
\begin{lemma}\cite[Cor.~4.1]{KM16}\label{lem:compslack}
Let $u^*$ and $\gamma^*$ be feasible for the aforementioned primal (\ref{eqn:primal}) and dual (\ref{eqn:dual}) problems, respectively\/. Then the objective functions of (\ref{eqn:primal}) and (\ref{eqn:dual}) with $u=u^*$ and $\gamma=\gamma^*$ are equal if and only if (i) $\int_Du^*\,d(\gamma_1^*-\gamma_2^*)=\int_Du^*\,d\bar{\mu}$, and (ii) $u^*(z)-u^*(z')=\|z-z'\|_\infty$, hold $\gamma^*-$a.e.
\end{lemma}

We now present a few examples to indicate why it is hard to compute the optimal mechanism using this dual approach. We first compute the components of $\bar{\mu}$ (i.e., $\mu,\mu_s,\mu_p$), with $f(z)=\frac{1}{b_1b_2}$ for $z\in D=[c,c+b_1]\times[c,c+b_2]$, from (\ref{eqn:mu-measure}), as
\begin{align}
 \mbox{(area density) } \, & \mu(z)=-3/(b_1b_2),\quad z \in D,\nonumber\\
 \mbox{(line density) } \, & \mu_s(z)=\sum_{i=1}^2(-c\mathbf{1}(z_i=c)+(c+b_i)\mathbf{1}(z_i=c+b_i))/(b_1b_2),\nonumber\\&\hspace*{2.5in} z\in\partial D,\nonumber\\
 \mbox{(point measure)} \, & \mu_p(z)=\delta_{\{(c,c)\}}(z).\label{eqn:bar-mu}
\end{align}

In the examples that we consider, we start by suggesting a certain mechanism, and prove that it is indeed the optimal mechanism by constructing a feasible $u$ and a feasible $\gamma$ that satisfy the complementary slackness constraints of Lemma \ref{lem:compslack}. While $u$ can be constructed easily from the allocations $q$, the construction of the transport variable $\gamma$ needs some work. This involves transporting mass from each point on the top and right boundaries of $D$ along the $45^\circ$ line containing the point. We shuffle the measure across the points on the boundary in case there is an excess or a deficit. The construction of the shuffling measure is the main challenge; it differs significantly across the examples we consider. We now fill in the details.

\subsection{Example 1: $z\sim\mbox{Unif }[1.26,2.26]^2$}
\begin{theorem}\label{thm:eg-1}\cite{Pav11}
Consider the case when $c=1.26$, and $b_1=b_2=1$. Then, the optimal mechanism is as depicted in Figure \ref{fig:illust-1}, with $\delta_1=\delta_2=20/63$ and $a_1=a_2=a=0.6615$.
\end{theorem}

\begin{proof}
\citet{Pav11} proved this via virtual valuations. We shall use the dual method. To prove this theorem, we must find a feasible $u$ and a feasible $\gamma$, and show that they satisfy the conditions of Lemma \ref{lem:compslack}. We define the allocation $q$ as given in Figure \ref{fig:illust-1}. The primal variable $u$ can be derived by fixing $u(c,c)=0$ and by using the allocation variable $q$, since $\nabla u=q$.

We now define functions $\alpha^{(1)},\beta^{(1)}:D\rightarrow\mathbb{R}$ as follows (see Figures \ref{fig:alpha} and \ref{fig:beta}).

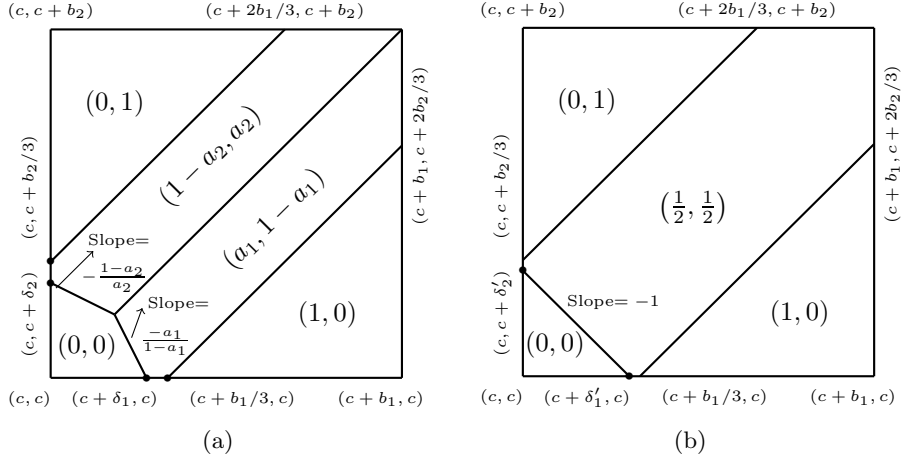
\begin{figure}[H]
\centering
\begin{tabular}{cc}
\subfloat[]{\label{fig:illust-1}\begin{tikzpicture}[scale=0.28,font=\tiny,axis/.style={very thick, -}]
\node at (-1,-1) {$(c,c)$};
\draw [axis,thick,-] (0,0)--(16.5,0);
\node at (15.5,-1) {$(c+b_1,c)$};
\draw [axis,thick,-] (0,0)--(0,16.5);
\node [above] at (0,16.5) {$(c,c+b_2)$};
\draw [axis,thick,-] (0,16.5)--(16.5,16.5);
\draw [axis,thick,-] (16.5,0)--(16.5,16.5);
\draw [axis,thick,-] (5.5,0)--(16.5,11);
\node [below,rotate=90] at (16.5,11) {$(c+b_1,c+2b_2/3)$};
\draw [axis,thick,-] (0,5.5)--(11,16.5);
\node [above] at (11,16.5) {$(c+2b_1/3,c+b_2)$};
\draw [axis,thick,-] (0,4.5)--(3,3);
\draw [axis,thick,-] (4.5,0)--(3,3);
\draw [axis,thick,-] (3,3)--(16.5,16.5);
\node at (9,-1) {$(c+b_1/3,c)$};
\node at (3,-1) {$(c+\delta_1,c)$};
\node [rotate=90] at (-1,9) {$(c,c+b_2/3)$};
\node [rotate=90] at (-1,3) {$(c,c+\delta_2)$};
\foreach \Point in {(4.5,0),(5.5,0),(0,4.5),(0,5.5)}{
 \node at \Point {\textbullet};}
\node at (1.7,1.5) {\small$(0,0)$};
\node at (13,3) {\small$(1,0)$};
\node at (3,13) {\small$(0,1)$};
\node [rotate=45] at (7.5,10.5) {\small$(1-a_2,a_2)$};
\node [rotate=45] at (10.5,7.5) {\small$(a_1,1-a_1)$};
\node at (3.2,6.5) {Slope$=$};
\node at (3,4.75) {$-\frac{1-a_2}{a_2}$};
\draw [thin,->] (0.25,4.25)--(2,6);
\node at (6,3.5) {Slope$=$};
\node at (5.4,1.75) {$\frac{-a_1}{1-a_1}$};
\draw [thin,->] (3.8,1.9)--(4.3,3.3);
\end{tikzpicture}}&
\subfloat[]{\label{fig:illust-2}\begin{tikzpicture}[scale=0.28,font=\tiny,axis/.style={very thick, -}]
\node at (-1,-1) {$(c,c)$};
\draw [axis,thick,-] (0,0)--(16.5,0);
\node at (15.5,-1) {$(c+b_1,c)$};
\draw [axis,thick,-] (0,0)--(0,16.5);
\node [above] at (0,16.5) {$(c,c+b_2)$};
\draw [axis,thick,-] (0,16.5)--(16.5,16.5);
\draw [axis,thick,-] (16.5,0)--(16.5,16.5);
\draw [axis,thick,-] (5.5,0)--(16.5,11);
\node [below,rotate=90] at (16.5,11) {$(c+b_1,c+2b_2/3)$};
\draw [axis,thick,-] (0,5.5)--(11,16.5);
\node [above] at (11,16.5) {$(c+2b_1/3,c+b_2)$};
\draw [axis,thick,-] (0,5)--(5,0);
\node at (9,-1) {$(c+b_1/3,c)$};
\node at (3,-1) {$(c+\delta_1',c)$};
\node [rotate=90] at (-1,9) {$(c,c+b_2/3)$};
\node [rotate=90] at (-1,3) {$(c,c+\delta_2')$};
\foreach \Point in {(5,0),(0,5)}{
 \node at \Point {\textbullet};}
\node at (1.5,1.5) {\small$(0,0)$};
\node at (13,3) {\small$(1,0)$};
\node at (3,13) {\small$(0,1)$};
\node at (8,8) {\small$\left(\frac{1}{2},\frac{1}{2}\right)$};
\node at (4.25,3.5) {Slope$=-1$};
\end{tikzpicture}}
\end{tabular}
\caption{Optimal mechanism when $b_1=b_2=1$ and (a) $c=1.26$, (b) $c=1.5$.}\label{fig:illust}
\end{figure}

\begin{align}
\alpha^{(1)}(c+t,c+t'):&=\begin{cases}3t-1&(t,t')\in[0,2/3]\times\{1\},\\0&\mbox{otherwise}.\end{cases}\label{eqn:alpha}\\
\beta^{(1)}(c+t,c+t'):&=\begin{cases}3t-1&(t,t')\in[2/3,1-\delta_2]\times\{1\},\\3t+3a(1-t-\delta_2)-c-1&(t,t')\in[1-\delta_2,1]\times\{1\},\\0&\mbox{otherwise}.\end{cases}\label{eqn:beta}
\end{align}
The functions $\alpha^{(2)}$ and $\beta^{(2)}$ are defined similarly on the intervals $(\{c+1\}\times[c,c+2/3])$ and $(\{c+1\}\times[c+2/3,c+1])$ respectively. Observe that $\alpha^{(i)}$ and $\beta^{(i)}$ are densities (Radon-Nikodym derivatives) of measures that are absolutely continuous w.r.t. the surface Lebesgue measure. The measures themselves are denoted $\bar{\alpha}^{(i)}$ and $\bar{\beta}^{(i)}$, respectively.

We now construct the dual variable $\gamma$ as follows. First, let (i) $\gamma_1:=\gamma_1^Z+\gamma_1^{D\backslash Z}$, where $Z$ is the exclusion region; (ii) $\gamma_1^Z=\bar{\mu}^Z$, the $\bar{\mu}$ measure restricted to $Z$; and (iii) $\gamma_1^{D\backslash Z}=(\bar{\mu}^{D\backslash Z}+\sum_i(\bar{\alpha}^{(i)}+\bar{\beta}^{(i)}))_+$. So $\gamma_1$ is supported on $Z\cup([1.26,2.26]\times\{2.26\})\cup(\{2.26\}\times[1.26,2.26])$. We define $\gamma_1^s$ as the Radon-Nikodym derivative of $\gamma_1$ w.r.t. the surface Lebesgue measure. It is easy to see that $\gamma_1^s(z)=\mu_s(z)+\sum_i(\alpha^{(i)}(z)+\beta^{(i)}(z))$ when $z\in(Z\cap D)\cup([1.26,2.26]\times\{2.26\})\cup(\{2.26\}\times[1.26,2.26])$, and zero otherwise. We now specify a transition probability kernel $\gamma(\cdot~|~x)$ for all $x$ in the support of $\gamma_1$.

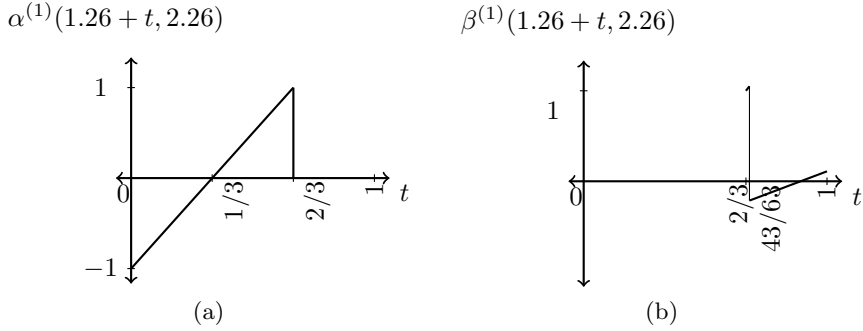
\begin{figure}[t]
\centering
\begin{tabular}{cc}
\subfloat[]{\label{fig:alpha}\begin{tikzpicture}[scale=0.2,font=\small,axis/.style={very thick, -}]
\draw [axis,thick,<->] (-1,7)--(17,7);
\node [right] at (17,6) {$t$};
\draw [axis,thick,<->] (0,0)--(0,15);
\node [above] at (-1,16) {$\alpha^{(1)}(1.26+t,2.26)$};
\draw [thick,-] (0,1)--(10.67,13);
\draw [thick,-] (10.67,13)--(10.67,7);
\node at (-0.5,6) {$0$};
\draw [thin,-] (5.33,6.75) -- (5.33,7.25);
\node [below,rotate=90] at (5.33,5.5) {$1/3$};
\draw [thin,-] (10.67,6.75) -- (10.67,7.25);
\node [below,rotate=90] at (10.67,5.5) {$2/3$};
\draw [thin,-] (16,6.75) -- (16,7.25);
\node [rotate=90] at (16,6.25) {$1$};
\draw [thin,-] (-0.25,1) -- (0.25,1);
\node at (-2,1) {$-1$};
\draw [thin,-] (-0.25,13) -- (0.25,13);
\node at (-2,13) {$1$};
\end{tikzpicture}}&
\subfloat[]{\label{fig:beta}\begin{tikzpicture}[scale=0.2,font=\small,axis/.style={very thick, -}]
\draw [axis,thick,<->] (-1,7)--(17,7);
\node [right] at (17,6) {$t$};
\draw [axis,thick,<->] (0,0)--(0,15);
\node [above] at (-1,16) {$\beta^{(1)}(1.26+t,2.26)$};
\draw [thick,-] (10.67,13)--(10.9,13.3);
\draw [thin,-] (10.9,13.3)--(10.9,5.72);
\draw [thick,-] (10.9,5.72)--(16,7.66);
\node at (-0.5,6) {$0$};
\draw [thin,-] (10.67,6.75) -- (10.67,7.25);
\node [rotate=90] at (10.3,5.5) {$2/3$};
\draw [thin,-] (10.9,6.75) -- (10.9,7.25);
\node [below,rotate=90] at (10.9,4.5) {$43/63$};
\draw [thin,-] (16,6.75) -- (16,7.25);
\node [rotate=90] at (16,6.25) {$1$};
\draw [thin,-] (-0.25,13) -- (0.25,13);
\node at (-2,11.67) {$1$};
\end{tikzpicture}}
\end{tabular}
\caption{(a) The measure $\alpha^{(1)}$. (b) The measure $\beta^{(1)}$.}
\end{figure}

\begin{enumerate}
 \item[(a)] For $x\in Z$, we define $\gamma(y~|~x)=\delta_x(y)$. This is interpreted as no mass being transferred.
 \item[(b)] For $x\in([1.26,2.26]\times\{2.26\})\cup(\{2.26\}\times[1.26,2.26])$, we define $\gamma(y~|~x)=(\mu(y)+\mu_s(y))_-/\gamma_1^s(x)$ if $y\in\{y\in D\backslash Z:y_1-y_2=x_1-x_2\}$, and zero otherwise. This is interpreted as a transfer of $\gamma_1^s(x)$ from the boundary point $x$ to (the $45^\circ$ line segment) $\{y\in D\backslash Z:y_1-y_2=x_1-x_2\}$, which has $x$ as one end-point.
\end{enumerate}
We then define $\gamma(F)=\int_{(x,y)\in F}\gamma_1(dx)\gamma(dy~|~x)$ for any measurable $F\in D\times D$. It is now easy to check that $\gamma_2^Z=\bar{\mu}^Z$, and $\gamma_2^{D\backslash Z}=(\bar{\mu}^{D\backslash Z}+\sum_i(\bar{\alpha}^{(i)}+\bar{\beta}^{(i)}))_-$. Thus we have $(\gamma_1-\gamma_2)^Z=0$, and $(\gamma_1-\gamma_2)^{D\backslash Z}=\bar{\mu}^{D\backslash Z}+\sum_i(\bar{\alpha}^{(i)}+\bar{\beta}^{(i)})$.

We now verify that $\gamma$ is feasible. Observe that the components of $\bar{\mu}^Z$ are positive only at the left-bottom corner of $D$ (i.e., at $(c,c)$) and negative elsewhere, and that $\bar{\mu}_+(Z)=1=\bar{\mu}_-(Z)$ (the second equality requires some calculations). So we have $\int_Z f\,d\bar{\mu}\leq 0$ for any increasing function $f$, and thus $\bar{\mu}^Z\preceq_{cvx}0=(\gamma_1-\gamma_2)^Z$. We next prove that $(\gamma_1-\gamma_2)^{D\backslash Z}\succeq_{cvx}\bar{\mu}^{D\backslash Z}$. Since $(\gamma_1-\gamma_2-\bar{\mu})^{D\backslash Z}=\sum_i(\bar{\alpha}^{(i)}+\bar{\beta}^{(i)})$, it suffices to prove that $\sum_i(\bar{\alpha}^{(i)}+\bar{\beta}^{(i)})\succeq_{cvx}0$. We do this in the next lemma.
\begin{lemma}\label{lem:cvx}
\begin{enumerate}
\item[(i)] The measure $\bar{\alpha}^{(1)}$ is such that $\bar{\alpha}^{(1)}([1.26,1.26+2/3]\times\{2.26\})=0$ and $\int_{1.26}^{1.26+2/3}(t-1.26)\,\bar{\alpha}^{(1)}(dt,2.26)\geq 0$. Hence for any $f$ constant on $[1.26,1.26+2/3]$, we have $\int_{1.26}^{1.26+2/3}f(t)\,d\bar{\alpha}^{(1)}(dt,2.26)=0$. Further, $\bar{\alpha}^{(1)}\succeq_{cvx}0$. A similar result holds for $\bar{\alpha}^{(2)}$.
\item[(ii)] $\bar{\beta}^{(1)}([1.26+2/3,2.26]\times\{2.26\})=0$ and $\int_{1.26+2/3}^{2.26}(t-1.26)\,\bar{\beta}^{(1)}(dt,2.26)=0$. Hence we have $\int_{1.26+2/3}^{2.26}f(t)\,\bar{\beta}^{(1)}(dt,2.26)=0$ for any affine $f$ on $[1.26+2/3,2.26]$. Further, $\bar{\beta}^{(1)}\succeq_{cvx}0$. A similar result holds for $\bar{\beta}^{(2)}$.
\end{enumerate}
\end{lemma}
\begin{proof}
See \ref{app:a}.\qed
\end{proof}
We have thus established that $\gamma_1-\gamma_2\succeq_{cvx}0$. We now verify if $u$ and $\gamma$ satisfy the conditions in Lemma \ref{lem:compslack}.
\begin{multline*}
  \int_Du\,d(\gamma_1-\gamma_2)=\int_Zu\,d(\gamma_1-\gamma_2)+\int_{D\backslash Z}u\,d(\gamma_1-\gamma_2)\\=\int_{D\backslash Z}u\,d\left(\bar{\mu}+\sum_i(\bar{\alpha}^{(i)}+\bar{\beta}^{(i)})\right)=\int_{D\backslash Z}u\,d\bar{\mu}=\int_Du\,d\bar{\mu},
\end{multline*}
where the second equality holds because $(\gamma_1-\gamma_2)^Z=0$; the third equality holds because $u(z)$ is a constant when $z\in([1.26,1.26+2/3]\times\{2.26\})\cup(\{2.26\}\times[1.26,1.26+2/3])$, and $u(z)$ is affine when $z\in([1.26+2/3,2.26]\times\{2.26\})\cup(\{2.26\}\times[1.26+2/3,2.26])$; and the last equality holds because $u(z)=0$ when $z\in Z$. To see why $u(z)-u(z')=\|z-z'\|_\infty$ holds $\gamma$-a.e., it suffices to check this condition for those $(z,z')$ for which $\gamma(\cdot~|~z)$ is nonzero, as in the cases (a) and (b) above. For $z,z'$ in (a), $z=z'$ and hence $u(z)-u(z')=0$; in (b), $(z,z')$ lie on a $45^\circ$ line, and hence $u(z)-u(z')=(z_1-z_1')=(z_2-z_2')=\|z-z\|_\infty$. Thus $u(z)-u(z')=\|z-z'\|_\infty$ holds $\gamma$-a.e.\qed
\end{proof}

The dual measure $\gamma$ was defined so that the measure $\gamma_1-\gamma_2-\bar{\mu}$, called the shuffling measure, convex-dominates $0$. Our key challenge in computing the optimal mechanism lies in constructing the shuffling measure. In the next example, we use a significantly different shuffling measure.

\subsection{Example 2: $z\sim\mbox{Unif }[1.5,2.5]^2$}
\begin{theorem}\cite{Pav11}\label{thm:eg-2}
Consider the case when $c=1.5$, and $b_1=b_2=1$. Then, the optimal mechanism is as depicted in Figure \ref{fig:illust-2}, with $\delta_1'=\delta_2'=\sqrt{5/3}-1$.
\end{theorem}
We use the shuffling measure $\bar{\lambda}+\sum_i(\bar{\alpha}^{(i)}+\bar{\beta}^{(i)})$, defined as follows. We define $\alpha^{(i)}$ and $\beta^{(i)}$, the respective Radon-Nikodym derivatives of the measures $\bar{\alpha}^{(i)}$ and $\bar{\beta}^{(i)}$ w.r.t. the surface Lebesgue measure, as in (\ref{eqn:alpha}) and (\ref{eqn:beta}) respectively, but with $\delta_1=\delta_2=\frac{((3+\sqrt{33})/8)-1}{(27-3\sqrt{33})/32}>\delta_2'$ and $a=(27-3\sqrt{33})/32$. We define $\lambda:D\rightarrow\mathbb{R}$, the Radon-Nikodym derivative of the measure $\bar{\lambda}$ w.r.t. the surface Lebesgue measure, as follows (see Figure \ref{fig:measure}):
\begin{align}
  \lambda(c&+(t-1+\delta_2)/2,c+\delta_2-(t-1+\delta_2)/2)\nonumber\\&=\lambda(c+\delta_2-(t-1+\delta_2)/2,c+(t-1+\delta_2)/2)\label{eqn:shuffle-eg-2}\\&=\begin{cases}3a(t-1+\delta_2)+c&t\in[1-\delta_2,1-\delta_2'],\\3t(a-1/2)+3/2(1-\delta_2')-3a(1-\delta_2)&t\in[1-\delta_2',1].\nonumber\end{cases}
\end{align}
$\lambda$ is defined to be $0$ at every other point in $D$. Observe that the function is defined on the line $z_1+z_2=2c+\delta_2$, and thus is symmetric about the line $z_1=z_2$.

\begin{figure}[t]
\centering
\begin{tikzpicture}[scale=0.2,font=\small,axis/.style={very thick, -}]
\draw [axis,thick,<->] (-1,6)--(16.5,6);
\node [right] at (16,5) {$t$};
\draw [axis,thick,<->] (0,-1)--(0,13);
\node [above] at (1,16.5) {\tiny L: $\lambda(c+(t-1+\delta_2)/2,c+\delta_2-(t-1+\delta_2)/2)$};
\draw [thick,-] (2,9)--(4,10);
\node [above] at (1,14.75) {\tiny R: $\lambda(c+\delta_1+(t-1+\delta_1)/2,c-(t-1+\delta_1)/2)$};
\draw [thin,-] (4,10)--(4,7);
\draw [thick,-] (4,7)--(8,3);
\draw [thick,-] (8,3)--(12,7);
\draw [thin,-] (12,7)--(12,10);
\draw [thick,-] (12,10)--(14,9);
\draw [thin,-] (2,5.75) -- (2,6.25);
\node [rotate=90] at (2,3) {$(1-\delta_2)$};
\draw [thin,-] (4,5.75) -- (4,6.25);
\node [rotate=90] at (4,3) {$(1-\delta_2')$};
\draw [thin,-] (8,5.75) -- (8,6.25);
\node [rotate=90] at (8,5.25) {$1$};
\draw [thin,-] (12,5.75) -- (12,6.25);
\node [rotate=90] at (12,3) {$(1-\delta_1')$};
\draw [thin,-] (14,5.75) -- (14,6.25);
\node [rotate=90] at (14,3) {$(1-\delta_1)$};
\draw [thin,-] (-0.25,9) -- (0.25,9);
\node at (-1,9) {$c$};
\draw [thin,dotted] (8,-1)--(8,14);
\node at (4,11.5) {Left};
\node at (12,11.5) {Right};
\end{tikzpicture}
\caption{The measure $\lambda$. The $y$-axis expressions for the left and the right portions of the graph are indicated using $L$ and $R$. The measure is symmetric because we have $\delta_1=\delta_2$.}\label{fig:measure}
\end{figure}
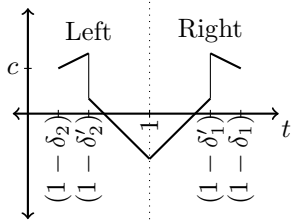

We construct the dual measure using $\bar{\lambda}+\sum_i(\bar{\alpha}^{(i)}+\bar{\beta}^{(i)})$ as the shuffling measure. Observe that the shuffling measure has a significantly different structure compared to Example-1. For a detailed proof of the theorem, we refer the reader to \ref{app:a}.

The results of Theorems \ref{thm:eg-1} and \ref{thm:eg-2} are parts of a more general result shown in \cite{Pav11}. Pavlov's proof uses a virtual valuation method, but our proofs use the dual approach. We now solve another example via the dual approach, going beyond those considered in \cite{Pav11}.

\subsection{Example 3: $z\sim\mbox{Unif }[0,1.2]\times[0,1]$}
\begin{theorem}\label{thm:eg-3}
Consider the case when $c=0$, $b_1=1.2$, and $b_2=1$. Then, the optimal mechanism is as in Figure \ref{fig:illust-3}, with $(\delta_1,\delta_2)$ simultaneously solving
 \begin{align*}
  &-3\delta_1\delta_2-c(\delta_1+\delta_2)+b_1b_2=0.\\
  &-\frac{3}{2}\delta_2^2+2b_2\delta_2-\frac{b_2^2}{2}+(c-2b_2+3\delta_2)\delta_1=0.
 \end{align*}
The values of $(\delta_1,\delta_2)$ can be solved numerically to be
$$
(\delta_1,\delta_2)\approx(0.678837,0.589243).
$$
\end{theorem}
We construct the shuffling measure $\bar{\alpha}+\bar{\alpha}^{(o)}+\bar{\alpha}^{(h)}$ as follows, using its respective Radon-Nikodym derivatives $\alpha$, $\alpha^{(o)}$, $\alpha^{(h)}$ w.r.t. the surface Lebesgue measure. The superscripts $(o)$ and $(h)$ stand for 'oblique' and 'horizontal'.
\begin{align}
\alpha(c+t,c+t'):&=\frac{1}{1.2}\begin{cases}3t-1&(t,t')\in[0,1-\delta_2]\times\{1\},\\3(1-\delta_2)-c-1&(t,t')\in[1-\delta_2,1+\delta^*]\times\{1\},\\0&\mbox{otherwise}.\end{cases}\label{eqn:shuffle-eg-3}\\
\alpha^{(o)}(c+t,c+t'):&=\frac{1}{1.2}\begin{cases}3t-1.2&(t,t')\in\{1.2\}\times[0,1.2-\delta_1],\\2(1.2)-3\delta_1&(t,t')\in\{1.2\}\times[1.2-\delta_1,1],\\0&\mbox{otherwise}.\end{cases}\label{eqn:shuffle-eg-3-o}\\
\alpha^{(h)}(c+t,c+t'):&=\frac{1}{1.2}\begin{cases}3(t-\delta_1+0.2)&(t,t')\in\{1.2\}\times[\delta_1-0.2,\delta_2],\\3(0.2-\delta^*)&(t,t')\in\{1.2\}\times[\delta_2,2/3],\\0&\mbox{otherwise}.\end{cases}\label{eqn:shuffle-eg-3-h}
\end{align}
We construct the dual measure using $\bar{\alpha}+\bar{\alpha}^{(o)}+\bar{\alpha}^{(h)}$ as the shuffling measure. For a detailed proof of Theorem \ref{thm:eg-3}, see \ref{app:a}. In point (d) of that proof, mass from certain points on the right-hand side boundary will be transferred to two line segments -- a $45^\circ$ line (oblique transfer via $\alpha^{(o)}$) and a horizontal line (via $\alpha^{(h)}$). Observe that the shuffling measure has a significantly different structure compared to Examples 1 and 2. 

\begin{minipage}{.46\textwidth}
\begin{figure}[H]
\centering
\begin{tikzpicture}[scale=0.2,font=\tiny,axis/.style={very thick, -}]
\node at (-1,-1) {$(0,0)$};
\draw [axis,thick,-] (0,0)--(16.5,0);
\node [right] at (15,-1) {$(b_1,0)$};
\draw [axis,thick,-] (0,0)--(0,15);
\node [above] at (0,15) {$(0,b_2)$};
\draw [axis,thick,-] (0,15)--(16.5,15);
\draw [axis,thick,-] (16.5,0)--(16.5,15);
\draw [axis,thick,-] (0,8.8)--(9.5,8.8);
\node [rotate=90] at (-1,9) {$(0,\delta_2)$};
\draw [axis,thick,-] (9.5,0)--(9.5,8.8);
\node at (9.5,-1) {$(\delta_1,0)$};
\draw [thick,-] (9.5,8.8)--(15.7,15);
\draw [thick,dotted] (9.5,8.8)--(0.7,0);
\node [right] at (0.2,-1) {$(\delta^*,0)$};
\node at (14.4,16) {$(b_2+\delta^*,b_2)$};
\node [above] at (5,3) {\small$(0,0)$};
\node [above] at (13,6) {\small$(1,0)$};
\node [above] at (6,10) {\small$(0,1)$};
\end{tikzpicture}
\caption{Optimal mechanism when $c=0$, $b_1=1.2$, $b_2=1$.}\label{fig:illust-3}
\end{figure}
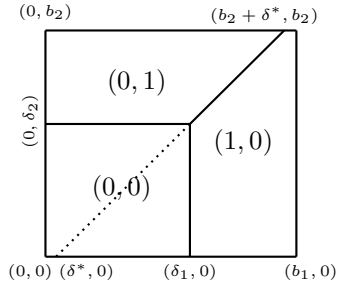
\end{minipage}
\begin{minipage}{.02\textwidth}
\hspace*{.02\textwidth}
\end{minipage}
\begin{minipage}{.46\textwidth}
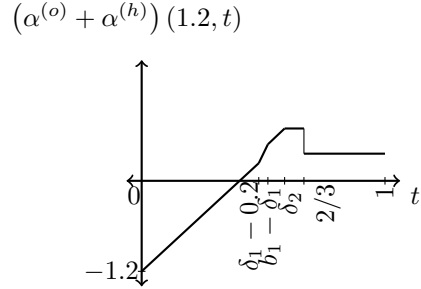
\begin{figure}[H]
\centering
\begin{tikzpicture}[scale=0.2,font=\small,axis/.style={very thick, -}]
\draw [axis,thick,<->] (-1,7)--(17,7);
\node [right] at (17,6) {$t$};
\draw [axis,thick,<->] (0,0)--(0,15);
\node [above] at (-1,16) {$\left(\alpha^{(o)}+\alpha^{(h)}\right)(1.2,t)$};
\draw [thick,-] (0,1)--(7.7,8.17);
\draw [thick,-] (7.7,8.17)--(8.3,9.42);
\draw [thick,-] (8.3,9.42)--(9.4,10.47);
\draw [thick,-] (9.4,10.47)--(10.67,10.47);
\draw [thin,-] (10.67,10.47)--(10.67,8.8);
\draw [thick,-] (10.67,8.8)--(16,8.8);
\node at (-0.5,6) {$0$};
\draw [thin,-] (7.7,6.75) -- (7.7,7.25);
\node [rotate=90] at (7.15,4) {$\delta_1-0.2$};
\draw [thin,-] (8.3,6.75) -- (8.3,7.25);
\node [rotate=90] at (8.4,4) {$b_1-\delta_1$};
\draw [thin,-] (9.4,6.75) -- (9.4,7.25);
\node [rotate=90] at (9.75,5.5) {$\delta_2$};
\draw [thin,-] (10.67,6.75) -- (10.67,7.25);
\node [below,rotate=90] at (10.67,5.5) {$2/3$};
\draw [thin,-] (16,6.75) -- (16,7.25);
\node [rotate=90] at (16,6.25) {$1$};
\draw [thin,-] (-0.25,1) -- (0.25,1);
\node at (-2,1) {$-1.2$};
\end{tikzpicture}
\caption{The measure $\alpha^{(o)}+\alpha^{(h)}$, for $c=0$, $b_1=1.2$, $b_2=1$.}\label{fig:measure2}
\end{figure}
\end{minipage}

\vspace*{10pt}
We have computed the optimal mechanisms for three representative examples using the dual approach. The challenge in each of the examples was to construct the appropriate shuffling measure $\gamma_1-\gamma_2-\bar{\mu}$ that convex-dominates $0$.  We now make some observations on the constructed shuffling measures.

\begin{enumerate}
\item[$\bullet$] The locations of the shuffling measure exhibit significant variations in our examples. For instance, the shuffling measure was non-zero only at the top boundary and the right boundary of $D$ in Theorems \ref{thm:eg-1} and \ref{thm:eg-3}, whereas, it was non-zero additionally on the line $z_1+z_2=2c+\delta_2$ in Theorem \ref{thm:eg-2}.
\item[$\bullet$] The structures of the shuffling measure also exhibit significant variations. The variations can be observed from the structures in Figures \ref{fig:measure} and \ref{fig:measure2}. This is in contrast to the unrestricted setting solved in \cite{TRN16}, where the shuffling measures were added at a fixed location and had a fixed structure.
\item[$\bullet$] In the case of $c=0, b_1=1.2, b_2=1$, the shuffling measure had to be constructed partly for a mass transfer along  the $45^\circ$ line segment, and partly for a transfer along the horizontal line segment (see point (d) in the proof of Theorem \ref{thm:eg-3}, \ref{app:a}). The example thus had two shuffling measures: $\bar{\alpha}^{(o)}$ and $\bar{\alpha}^{(h)}$.
\end{enumerate}

The variability in the examples above makes it difficult for us to arrive at a general algorithmic method to construct shuffling measures, even for the restricted setting of uniform distributions. This motivates us to tackle the general problem using the virtual valuation method in \cite{Pav11}.

\section{Exploring The Virtual Valuation Method}
Recall that we consider the problem of optimal mechanism design in a two-item, one-buyer, unit-demand setting. In this section, we compute the optimal mechanism when the buyer's valuation $z\sim\mbox{Unif}[c,c+b_1]\times[c,c+b_2]$, using the virtual valuation method in \cite{Pav11}. We start with the following general result from \cite{Pav11}.
\begin{theorem}\cite[Prop.~1]{Pav11}\label{thm:pav-1}
 If the distribution $f$ satisfies
 $$
   3f_1(z)f_2(z)+z_1f_1'(z)f_2(z)+z_2f_1(z)f_2'(z)\geq 0\,\forall z\in D,
 $$
 then the allocation function $q$ in the optimal mechanism is such that $q_1+q_2\in\{0,1\}$.
\end{theorem}

Thus, if $f$ satisfies the above sufficient condition, then for every $z\in D\backslash Z$, $q(z)$ satisfies $q_1(z)+q_2(z)=1$. Recall that $Z$ is the exclusion region. Observe that the sufficient condition in Theorem \ref{thm:pav-1} is clearly satisfied for the uniform distribution $\mbox{Unif}[c,c+b_1]\times[c,c+b_2]$. The utility of the buyer in $D\backslash Z$ can be written as $u(z)=(z_1-z_2)q_1(z)+z_2-t(z)$, where we have used $q_2=1-q_1$. Defining $\delta:=z_1-z_2$, we have $\delta\in[-b_2,b_1]$ for the case under consideration. The following theorem from \cite{Pav11} reduces the domains of $q$ and $t$ from two-dimensions to one-dimension.

\begin{theorem}\cite[Prop.~2]{Pav11}\label{thm:pav-2}
 In the optimal mechanism, the allocations and the payments, $(q,t)$, can be rewritten so that they are a constant for every $\{z\in D\backslash Z:z_1-z_2=\delta\}$.
\end{theorem}

The theorem indicates that if $Z$ is fixed, then the domains of $(q,t)$ become one-dimensional in the region $D\backslash Z$; they can be written as $t(\delta)$ and $q_1(\delta)$, where $t:[-b_2,b_1]\rightarrow\mathbb{R}_+$, $q_1:[-b_2,b_1]\rightarrow[0,1]$, and $q_2=1-q_1$. As done in \cite{Pav11}, define $u_1:[-b_2,b_1]\rightarrow\mathbb{R}$, $u_1(\delta):=\delta q_1(\delta)-t(\delta)$, and define 
\[
  g(u_1(\delta),\delta):=\int_{\substack{z:z_1-z_2=\delta,\\u_1(\delta)+z_2>0}}f(z)\,dz.
\]

The function $g(u_1(\delta),\delta)$ resembles the marginal of $f$ along the $z_1-z_2$ axis, but for the fact that the marginal is computed by integrating only up to the point where $u_1(z_1-z_2)+z_2=0$. Call this point $z_2^*(\delta)$, and observe that $u(\delta+z_2^*(\delta),z_2^*(\delta))=u_1(\delta)+z_2^*(\delta)=0$. So $z_2=z_2^*(\delta)$ is the boundary point between the exclusion region $Z$, and the other regions. Further, $\{z:z_1-z_2=\delta,\,z_2<z_2^*(\delta)\}$ belongs to $Z$. So the function $g(u_1(\delta),\delta)$ is actually the marginal of $f$ in $D\backslash Z$, along the $z_1-z_2$ axis.

Consider the problem of maximizing the expected revenue subject to IC and IR constraints. The IC constraint, from \cite[Lem.~2]{Mye81}, can equivalently be written as (i) $q_1$ increasing, and (ii) $u_1(\delta)$ has the representation $u_1(\delta)=u_1(-b_2)+\int_{-b_2}^{\delta}q_1(\tilde{\delta})\,d\tilde{\delta}$ for every $\delta\in[-b_2,b_1]$. The optimal mechanism can thus be computed by solving the following optimization problem.
\begin{align}\label{eqn:optim-myerson}
&\max_{q_1(\cdot),u_1(\cdot)}\int_{-b_2}^{b_1}(\delta q_1(\delta)-u_1(\delta))g(u_1(\delta),\delta)\,d\delta\\
\mbox{subject to}\hspace*{0.5in}
&(a)\,q_1(\delta)\in[0,1]\,\forall\delta\in[-b_2,b_1];\,q_1\mbox{ increasing};\nonumber\\
&(b)\,u_1(\delta)=u_1(-b_2)+\int_{-b_2}^{\delta}q_1(\tilde{\delta})\,d\tilde{\delta}\,\forall\delta\in[-b_2,b_1].\nonumber
\end{align}

The IR constraint is already taken into account because the integral in the objective function of (\ref{eqn:optim-myerson}) is over $D\backslash Z$, i.e., where $u(z)\geq 0$.

Observe that the problem (\ref{eqn:optim-myerson}) is similar to the optimization problem in \cite[Lem.~3]{Mye81}. To solve the problem in a similar way, we now search for an equivalent of the virtual valuation function $\phi$ in our setting.

Applying integration by parts to the objective function of (\ref{eqn:optim-myerson}), we get $\int_{-b_2}^{b_1}\bar{V}(\delta)q_1(\delta)\,d\delta$, where the marginal profit function $\bar{V}:[-b_2,b_1]\rightarrow\mathbb{R}$ is defined as\footnote{We use the term marginal profit function, see \citet{Pav11}, based on the fact that $\bar{V}$ denotes the marginal contribution of allocation $q_1(\delta)$ to the profit of the seller.}
$$
  \bar{V}(\delta):=\delta g(u_1(\delta),\delta)-\int_{\delta}^{b_1}g(u_1(\tilde{\delta}),\tilde{\delta})\,d\tilde{\delta}+\int_{\delta}^{b_1}(\tilde{\delta} q_1(\tilde{\delta})-u_1(\tilde{\delta}))\frac{\partial}{\partial u_1}g(u_1(\tilde{\delta}),\tilde{\delta})\,d\tilde{\delta}.
$$
Notice that in Myerson's setting, we have $g(u_1(\delta),\delta)=f(\delta)$, and thus $\bar{V}(\delta)=\delta f(\delta)-\int_{\delta}^{b_1}f(\tilde{\delta})\,d\tilde{\delta}=\phi(\delta)f(\delta)$. We thus expect $\bar{V}$ to have similar properties of $\phi$. The following result from \cite{Pav06} provides some ``ironing conditions'' on $\bar{V}$, similar to those on $\phi$ in Myerson's setting.

\begin{theorem}\cite[Lem.~3,~Prop.~5]{Pav06}\label{thm:Myerson}
A mechanism is optimal if and only if it satisfies the following conditions:
\begin{enumerate}
 \item $q_1(\delta)$ is strictly increasing on $(\delta',\delta'')$ if and only if (iff) $\bar{V}(\delta)=0$ on this interval.
 \item $q_1(\delta)=0$ for $\delta\in[\delta',\delta'']$ iff (a) $\delta'=-b_2$, (b) $\bar{V}(\delta'')=0$ unless $\delta''=b_1$, (c) $\int_{\delta'}^{\delta''}\bar{V}(\delta)\,d\delta=\underline{k}\leq 0$, and (d) $\int_{\delta'}^x\bar{V}(\delta)\,d\delta\geq\underline{k}$ for all $x\in[\delta',\delta'']$.
 \item $q_1(\delta)=q\in(0,1)$ for $\delta\in[\delta',\delta'']$ iff (a) $\bar{V}(\delta')=0$ unless $\delta'=-b_2$, (b) $\bar{V}(\delta'')=0$ unless $\delta''=b_1$, (c) $\int_{\delta'}^{\delta''}\bar{V}(\delta)\,d\delta=0$, and (d) $\int_{\delta'}^x\bar{V}(\delta)\,d\delta\geq0$ for all $x\in[\delta',\delta'']$.
 \item $q_1(\delta)=1$ for $\delta\in[\delta',\delta'']$ iff (a) $\bar{V}(\delta')=0$ unless $\delta'=-b_2$, (b) $\delta''=b_1$, (c) $\int_{\delta'}^{\delta''}\bar{V}(\delta)\,d\delta=\overline{k}\geq 0$, and (d) $\int_x^{\delta''}\bar{V}(\delta)\,d\delta\leq\overline{k}$ for all $x\in[\delta',\delta'']$.
\end{enumerate}
\end{theorem}

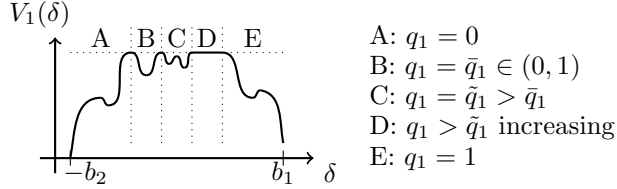
\begin{figure}[H]
\centering
\begin{tikzpicture}[scale=0.2,font=\small,axis/.style={very thick, -}]
\draw [axis,thick,->] (-1,0)--(17,0);
\node [right] at (17,-1) {$\delta$};
\draw [thin,-] (1,-0.5)--(1,0.5);
\node at (2,-1) {$-b_2$};
\draw [thin,-] (15,-0.5)--(15,0.5);
\node at (15,-1) {$b_1$};
\draw [axis,thick,->] (0,-1)--(0,8);
\node [above] at (-1,8) {$V_1(\delta)$};
\draw [axis,thick,-] (1,0) to[out=80,in=180] (3,4);
\draw [axis,thick,-] (3,4) to[out=0,in=180] (3.5,3.5);
\draw [axis,thick,-] (3.5,3.5) to[out=0,in=180] (5,7);
\draw [axis,thick,-] (5,7) to[out=0,in=180] (6,5.5);
\draw [axis,thick,-] (6,5.5) to[out=0,in=180] (7,7);
\draw [axis,thick,-] (7,7) to[out=0,in=180] (7.5,6.25);
\draw [axis,thick,-] (7.5,6.25) to[out=0,in=180] (8,6.75);
\draw [axis,thick,-] (8,6.75) to[out=0,in=180] (8.5,6);
\draw [axis,thick,-] (8.5,6) to[out=0,in=180] (9,7);
\draw [axis,thick,-] (9,7) to[out=0,in=180] (11,7);
\draw [axis,thick,-] (11,7) to[out=0,in=180] (13,4);
\draw [axis,thick,-] (13,4) to[out=0,in=180] (13.5,4.5);
\draw [axis,thick,-] (13.5,4.5) to[out=-10,in=100] (15,1);
\draw [thin,dotted] (1,7)--(15,7);
\node at (3,8) {A};
\draw [thin,dotted] (5,1)--(5,9);
\node at (6,8) {B};
\draw [thin,dotted] (7,1)--(7,9);
\node at (8,8) {C};
\draw [thin,dotted] (9,1)--(9,9);
\node at (10,8) {D};
\draw [thin,dotted] (11,1)--(11,9);
\node at (13,8) {E};
\node [right] at (20,8) {A: $q_1=0$};
\node [right] at (20,6) {B: $q_1=\bar{q}_1\in(0,1)$};
\node [right] at (20,4) {C: $q_1=\tilde{q}_1>\bar{q}_1$};
\node [right] at (20,2) {D: $q_1>\tilde{q}_1$ increasing};
\node [right] at (20,0) {E: $q_1=1$};
\end{tikzpicture}
\caption{Illustration of the conditions in Theorem \ref{thm:Myerson}.}\label{fig:my-illust-1}
\end{figure}

Define $V_1(\delta)=-\int_{-b_2}^{\delta}\bar{V}(\tilde{\delta})\,d\tilde{\delta}$. We now argue that the conditions in Theorem \ref{thm:Myerson} can be interpreted as conditions on $\delta$ where $V_1$ attains its global maximum. The theorem states that the mechanism is optimal if and only if the following conditions hold. Take $\delta'$ and $\delta''$ to be the left and right end points of an interval under consideration.
\begin{itemize}
\item Let $q_1(\delta)=0$ $\forall$ $\delta\in[\delta',\delta'']$. Then (a) $\delta'=-b_2$ and (b) $V_1(\delta)$ is maximized at $\delta''$ (see region A, Figure \ref{fig:my-illust-1}).
\item Let $q_1(\delta)=q\in(0,1)$ when $\delta\in[\delta',\delta'']$. Then $V_1(\delta)$ is maximized at both $\delta'$ and $\delta''$ (see regions B and C, Figure \ref{fig:my-illust-1}).
\item Let $q_1(\delta)$ be strictly increasing when $\delta\in[\delta',\delta'']$. Then $V_1(\delta)=\max_\delta V_1(\delta)$ for all $\delta\in[\delta',\delta'']$ (see region D, Figure \ref{fig:my-illust-1}).
\item Let $q_1(\delta)=1$ $\forall$ $\delta\in[\delta',\delta'']$. Then (a) $\delta''=b_1$ and (b) $V_1(\delta)$ is maximized at $\delta'$ (see region E, Figure \ref{fig:my-illust-1}).
\end{itemize}

Observe that the conditions mentioned above are a consequence of the conditions stated in Theorem \ref{thm:Myerson}. The conditions 2(c)--(d), 3(c)--(d), and 4(c)--(d), are representations that indicate that the global maximum must occur at certain end points of the interval. The value of $q_1$ changes only at those $\delta$ where $V_1$ attains its global maximum. We have a similar result in one-dimension, where the value of $q$ changes only at those $z$ where $-\int_0^{z}\phi(t)f(t)\,dt=z(1-F(z))$ is maximized \cite[p.~338]{NRTV07}.

Theorem \ref{thm:Myerson} and the above interpretation highlight the similarity between the virtual valuation functions $\phi$ and $\bar{V}$. The key difference between $\phi$ and $\bar{V}$ is that the former depends only on $f$, whereas the latter depends on $u_1(\delta)$ in addition, which is known only when the optimal mechanism is known. So the computation of $\bar{V}$ requires the knowledge of the mechanism itself. However, given a mechanism, we can use the theorem to determine if the mechanism is optimal or not.

We now simplify the computation of the marginal profit function. We define virtual valuation function $V:[-b_2,b_1]\rightarrow\mathbb{R}$ as $V(\delta):=\bar{\mu}(\{z:z_1-z_2\geq\delta\}\backslash Z)$ where $\bar{\mu}$ is as defined in Section \ref{sec:prelim}. We then have $\bar{\mu}(D)=0$ (see (\ref{eqn:mu-D})). The following lemma shows that $V$ is equal to the marginal profit function $\bar{V}$.
\begin{lemma}\label{lem:V-V'}
Let the allocation function $q$ be such that there exists a $u:D\rightarrow\mathbb{R}$ with $\nabla u=q$. Then, the functions $V$ and $\bar{V}$ are one and the same.
\end{lemma}
\begin{proof}
 See \ref{app:b}.\qed
\end{proof}

This lemma could be understood as follows.
\begin{itemize}
\item Recall that the expected revenue equals $\int_{-b_2}^{b_1}\bar{V}(\delta)q_1(\delta)\,d\delta$. The expected revenue thus increases by $\bar{V}(\delta)$ for a differential increase in $q_1$ at $\delta$.
\item A differential increase in $q_1$ increases $u$ uniformly for all $\delta'\geq\delta$, since $q=\nabla u$.
\item From (\ref{eqn:primal}), we know that the expected revenue equals $\int_D u\,d\bar{\mu}$. So a uniform increase for all $\delta'\geq\delta$ increases the expected revenue by $\bar{\mu}(\{z:z_1-z_2\geq\delta\}\backslash Z)$.
\item Thus we have $\bar{V}(\delta)=\bar{\mu}(\{z:z_1-z_2\geq\delta\}\backslash Z)$.
\end{itemize}

Observe that the virtual valuation function $V$ can be computed if the exclusion region $Z$ is known. In the rest of the paper, we propose some structures for all possible values of $(c,b_1,b_2)\geq 0$, and then prove that the optimal mechanisms indeed have those structures, using Theorem \ref{thm:Myerson}.

\subsection{Optimal mechanisms for the uniform distribution on a rectangle}
Without loss of generality, we assume $b_1\geq b_2$. The following theorem asserts that the optimal mechanism falls within one of the structures depicted in Figures \ref{fig:a-ini}--\ref{fig:e'-ini}.

\begin{theorem}\label{thm:consolidate}
Consider $z\sim\mbox{Unif }[c,c+b_1]\times[c,c+b_2]$. The optimal mechanism in the unit-demand setting is described as follows.
\begin{multicols}{2}
\begin{enumerate}
\item Case $b_1\in[b_2,3b_2/2]$:
\begin{enumerate}
\item[(a) ] $c\in[0,b_2]$: Figure \ref{fig:a-ini}
\item[(b)$^*$] $c\in[b_2,\alpha_1]$: Figure \ref{fig:b-ini}
\item[(c)$^*$] $c\in[\alpha_1,\alpha_2]$: Figure \ref{fig:c-ini}
\item[(d)$^*$] $c\in[\alpha_2,\frac{27b_1^2b_2^2}{4(b_1^3-b_2^3)}]$: Figure \ref{fig:d-ini}
\item[(e) ] $c\geq\frac{27b_1^2b_2^2}{4(b_1^3-b_2^3)}$: Figure \ref{fig:e-ini}
\end{enumerate}
\item Case $b_1\geq 3b_2/2$:
\begin{enumerate}
\item[(a) ] $c\in[0,b_2]$: Figure \ref{fig:a-ini}
\item[(b)$^*$] $c\in[b_2,\beta]$: Figure \ref{fig:b-ini}
\item[(c) ] $c\in[\beta,\frac{216b_1^2b_2}{108b_1^2-108b_1b_2-5b_2^2}]$: Figure \ref{fig:d'-ini}
\item[(d) ] $c\geq\frac{216b_1^2b_2}{108b_1^2-108b_1b_2-5b_2^2}$: Figure \ref{fig:e'-ini}
\end{enumerate}
\end{enumerate}
\end{multicols}

The values of $\alpha_1$, $\alpha_2$ and $\beta$ are defined as follows.
\begin{itemize}
 \item $c=\alpha_1$ is obtained by solving the following equations simultaneously for $(c,h,\delta^*)$.
\begin{align}
  &3h^2/2+ch+2b_2\delta^*-b_1b_2+b_2^2/2=0.\label{eqn:fig-b-first}\\
 27(c+h+\delta^*)&(b_2+\delta^*)^2-4(4b_2+3\delta^*)(3(h+\delta^*)/2+c)^2=0.\label{eqn:fig-b-second}\\
  2b_1^3/27&-(c+h)h^2/2+b_2(\delta^*)^2-b_2\delta^*(b_1-b_2/2)=0.\label{eqn:fig-b-third}
\end{align}
 \item $c=\alpha_2$ is the solution obtained by solving (\ref{eqn:fig-b-second}) and the following equations simultaneously for $(c,h,\delta^*)$.
\begin{align}
 &(2b_1^3/27+b_2(\delta^*)^2-b_2\delta^*(b_1-b_2/2))(3h/2+c)^2\nonumber\\&\hspace*{1.5in}-(c+h)(2b_2\delta^*+b_2^2/2-b_1b_2)^2/2=0.\label{eqn:fig-c-second}\\
 &2b_1b_2(b_2^2+4b_2\delta^*-2c(\delta^*+h)-3h(2\delta^*+h))\nonumber\\&\hspace*{0.5in}-(b_2^2+4b_2\delta^*-3\delta^*h)(b_2^2+4b_2\delta^*-2c\delta^*-3\delta^*h)=0.\label{eqn:fig-c-third}
\end{align}
\item $c=\beta\geq b_2$ solves
\begin{equation}\label{eqn:menu-2-bound}
  72b_1^2b_2+144b_1b_2^2-90b_2^3+(-36b_1^2+84b_1b_2+399b_2^2)c-(96b_1+208b_2)c^2=0.
\end{equation}
\end{itemize}
\end{theorem}

\begin{remark}
 The starred portions in the theorem statement indicate that we used Mathematica to verify certain inequalities in proving those parts.
\end{remark}
\begin{remark}
 The values of $\alpha_1$ fall in the interval $[b_2,tb_2]$, where $t=3(37+3\sqrt{465})/176\approx 1.733379$. Similarly, the values of $\alpha_2\in[kb_2,tb_2]$ where $k\geq 1$ is the root of $32k^3-54k^2+19=0$ ($k\approx 1.37214$), and the values of $\beta\in[tb_2,2b_2)$. See Figure \ref{fig:phase-diagram}.
\end{remark}

The following is a pictorial representation of the results in Theorem \ref{thm:consolidate}. It depicts the regions in $(c,b_1,b_2)$ space at which each of the mechanisms depicted in Figures 2a--2g turns out to be optimal.

\begin{figure}[H]
\centering
\begin{minipage}{.45\textwidth}
\centering
\begin{tikzpicture}[scale=0.36,font=\small,axis/.style={very thick, ->, >=stealth'}]
\draw [axis,thick,->] (0,-1)--(0,13);
\node [right] at (11,-1) {$\frac{b_1}{b_2}$};
\draw [axis,thick,->] (-0.5,0)--(12,0);
\node [above] at (1,11) {$\frac{c}{b_2}$};
\node at (-1,0.25) {$0$};
\draw [thin,-] (-0.25,2) -- (0.25,2);
\node [left] at (-0.25,2) {$2$};
\draw [thin,-] (-0.25,4) -- (0.25,4);
\node [left] at (-0.25,4) {$4$};
\draw [thin,-] (-0.25,6) -- (0.25,6);
\node [left] at (-0.25,6) {$6$};
\draw [thin,-] (-0.25,8) -- (0.25,8);
\node [left] at (-0.25,8) {$8$};
\draw [thin,-] (-0.25,10) -- (0.25,10);
\node [left] at (-0.25,10) {$10$};
\draw [thin,-] (-0.25,12) -- (0.25,12);
\node [left] at (-0.25,12) {$12$};
\node at (0.5,-1) {$1$};
\draw [thin,-] (2,-0.25) -- (2,0.25);
\node [below] at (2,-0.25) {$1.2$};
\draw [thin,-] (4,-0.25) -- (4,0.25);
\node [below] at (4,-0.25) {$1.4$};
\draw [thin,-] (6,-0.25) -- (6,0.25);
\node [below] at (6,-0.25) {$1.6$};
\draw [thin,-] (8,-0.25) -- (8,0.25);
\node [below] at (8,-0.25) {$1.8$};
\draw [thin,-] (10,-0.25) -- (10,0.25);
\node [below] at (10,-0.25) {$2$};
\draw [thick,-] (0,1) to (12,1);
\draw [thick,-] (0,1) to (5,1.733) to (12,1.97);
\draw [thick,-] (0,1.372) to (5,1.733);
\draw [thick,-] (2.2,12) to[out=-60,in=170] (12,3.98);
\draw [thick,dotted] (5,0) to (5,12);

\node [right] at (5,11) {$b_1=(1.5)b_2$};
\node at (2.2,8.5) {asymptotic};
\node at (2.2,7.5) {to $b_1=b_2$};
\node at (9,3.8) {asymptotic};
\node at (9,2.8) {to $c=2b_2$};
\node at (10,1.4) {$c=b_2$};
\draw [thick,->] (8.5,2) -- (8.5,2.6);
\draw [thick,->] (8.5,5) -- (8.5,4.2);

\path[fill=gray!50,opacity=.5] (0,1) to (12,1) to (12,0) to (0,0);
\end{tikzpicture}
\end{minipage}
\begin{minipage}{0.025\textwidth}
 \hspace*{0.025\textwidth}
\end{minipage}
\begin{minipage}{0.45\textwidth}
\centering
\begin{tikzpicture}[scale=0.25,font=\normalsize,axis/.style={very thick, -}]
\node [rotate=45] at (-1,0) {$c$};
\node [rotate=45] at (0,-1) {$c$};
\draw [axis,thick,-] (0,0)--(15,0);
\node [rotate=45] at (15,-1.2) {$c+b_1$};
\draw [axis,thick,-] (0,0)--(0,12);
\node [rotate=45] at (-1,12) {$c+b_2$};
\draw [axis,thick,-] (0,12)--(15,12);
\draw [axis,thick,-] (15,0)--(15,12);
\draw [axis,thick,-] (0,5)--(6.5,5);
\draw [thin,-] (-0.4,5)--(0.4,5);
\node [rotate=45] at (-2,4) {$c+\delta_2$};
\draw [axis,thick,-] (6.5,0)--(6.5,5);
\draw [thin,-] (6.5,-0.4)--(6.5,0.4);
\node [rotate=45] at (6.5,-1.8) {$c+\delta_1$};
\draw [thick,-] (6.5,5)--(13.5,12);
\draw [thick,dotted] (6.5,5)--(1.5,0);
\draw [thin,-] (1.5,-0.4)--(1.5,0.4);
\node [rotate=45] at (1,-1.8) {$c+\delta^*$};
\draw [thin,-] (13.5,11.6)--(13.4,12.4);
\node at (12.5,13) {$c+b_2+\delta^*$};
\node [above] at (3.5,1.5) {$(0,0)$};
\node [above] at (11.5,4.5) {$(1,0)$};
\node [above] at (3,8) {$(0,1)$};
\end{tikzpicture}
\end{minipage}
\caption{When $(c,b_1,b_2)$ falls in the shaded region in the left, the optimal mechanism is as depicted in the right. Item $1$ is offered for a price of $c+\delta_1$, and item $2$ is offered for a price of $c+\delta_2$.}\label{fig:a-new}
\end{figure}
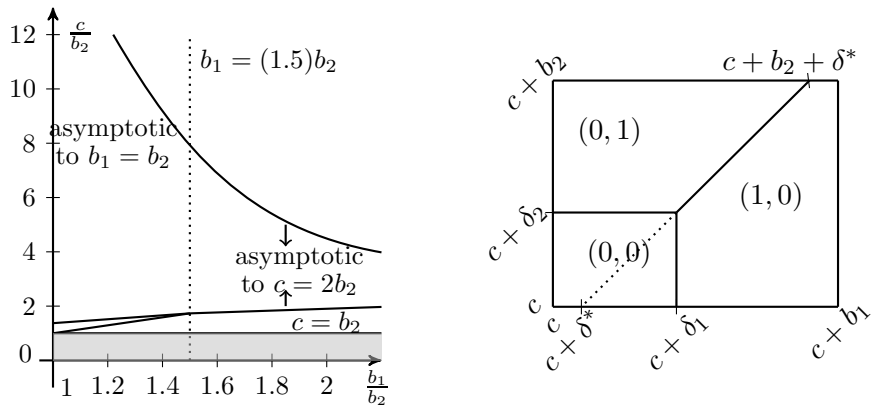

\begin{figure}[H]
\centering
\begin{minipage}{.45\textwidth}
\centering
\begin{tikzpicture}[scale=0.36,font=\small,axis/.style={very thick, ->, >=stealth'}]
\draw [axis,thick,->] (0,-1)--(0,13);
\node [right] at (11,-1) {$\frac{b_1}{b_2}$};
\draw [axis,thick,->] (-0.5,0)--(12,0);
\node [above] at (1,11) {$\frac{c}{b_2}$};
\node at (-1,0.25) {$0$};
\draw [thin,-] (-0.25,2) -- (0.25,2);
\node [left] at (-0.25,2) {$2$};
\draw [thin,-] (-0.25,4) -- (0.25,4);
\node [left] at (-0.25,4) {$4$};
\draw [thin,-] (-0.25,6) -- (0.25,6);
\node [left] at (-0.25,6) {$6$};
\draw [thin,-] (-0.25,8) -- (0.25,8);
\node [left] at (-0.25,8) {$8$};
\draw [thin,-] (-0.25,10) -- (0.25,10);
\node [left] at (-0.25,10) {$10$};
\draw [thin,-] (-0.25,12) -- (0.25,12);
\node [left] at (-0.25,12) {$12$};
\node at (0.5,-1) {$1$};
\draw [thin,-] (2,-0.25) -- (2,0.25);
\node [below] at (2,-0.25) {$1.2$};
\draw [thin,-] (4,-0.25) -- (4,0.25);
\node [below] at (4,-0.25) {$1.4$};
\draw [thin,-] (6,-0.25) -- (6,0.25);
\node [below] at (6,-0.25) {$1.6$};
\draw [thin,-] (8,-0.25) -- (8,0.25);
\node [below] at (8,-0.25) {$1.8$};
\draw [thin,-] (10,-0.25) -- (10,0.25);
\node [below] at (10,-0.25) {$2$};
\draw [thick,-] (0,1) to (12,1);
\draw [thick,-] (0,1) to (5,1.733) to (12,1.97);
\draw [thick,-] (0,1.372) to (5,1.733);
\draw [thick,-] (2.2,12) to[out=-60,in=170] (12,3.98);
\draw [thick,dotted] (5,0) to (5,12);

\node [right] at (5,11) {$b_1=(1.5)b_2$};
\node at (2.2,8.5) {asymptotic};
\node at (2.2,7.5) {to $b_1=b_2$};
\node at (9,3.8) {asymptotic};
\node at (9,2.8) {to $c=2b_2$};
\node at (10,1.4) {$c=b_2$};
\draw [thick,->] (8.5,2) -- (8.5,2.6);
\draw [thick,->] (8.5,5) -- (8.5,4.2);

\path[fill=gray!50,opacity=.5] (0,1) to (5,1.733) to (12,1.97) to (12,1) to (0,1);
\end{tikzpicture}
\end{minipage}
\begin{minipage}{0.025\textwidth}
 \hspace*{0.025\textwidth}
\end{minipage}
\begin{minipage}{0.45\textwidth}
\centering
\begin{tikzpicture}[scale=0.25,font=\normalsize,axis/.style={very thick, -}]
\node [rotate=45] at (-1,0) {$c$};
\node [rotate=45] at (0,-1) {$c$};
\draw [axis,thick,-] (0,0)--(15,0);
\node [rotate=45] at (15,-1.2) {$c+b_1$};
\draw [axis,thick,-] (0,0)--(0,12);
\node [rotate=45] at (-1,12) {$c+b_2$};
\draw [axis,thick,-] (0,12)--(15,12);
\draw [axis,thick,-] (15,0)--(15,12);
\draw [axis,thick,-] (0,4)--(8,12);
\draw [thin,-] (-0.4,4)--(0.4,4);
\node [rotate=45] at (-2,3.5) {$c+b_2/3$};
\draw [axis,thick,-] (0,3.5)--(4,2.5);
\draw [thin,-] (-0.4,3.5)--(0.4,3.5);
\node [rotate=45] at (-1.8,1.8) {$c+\delta_2$};
\draw [axis,thick,-] (4,0)--(4,2.5);
\draw [thin,-] (4,-0.4)--(4,0.4);
\node [rotate=45] at (3.4,-1.8) {$c+\delta_1$};
\draw [thin,<->] (4.5,0)--(4.5,2.5);
\node at (5,1) {$h$};
\draw [thick,-] (4,2.5)--(13.5,12);
\draw [thick,dotted] (4,2.5)--(1.5,0);
\draw [thin,-] (1.5,-0.4)--(1.5,0.4);
\node [rotate=45] at (1,-1.8) {$c+\delta^*$};
\draw [thin,-] (13.5,11.6)--(13.5,12.4);
\node at (14.5,13) {$c+b_2+\delta^*$};
\draw [thin,-] (8,11.6)--(8,12.4);
\node at (7.25,13) {\footnotesize$c+2b_2/3$};
\node at (2,1.5) {$(0,0)$};
\node [above] at (11,4) {$(1,0)$};
\node [above] at (2,8) {$(0,1)$};
\node [rotate=45] at (6,7) {$(1-a_2,a_2)$};
\end{tikzpicture}
\end{minipage}
\caption{When $(c,b_1,b_2)$ falls in the shaded region in the left, the optimal mechanism is as depicted in the right. Item $1$ is offered for a price of $c+\delta_1$, item $2$ is offered for a price of $c+b_2/3$, and a lottery with probabilities $(1-a_2,a_2)$ is offered for a price of $c+a_2\delta_2$.}\label{fig:b-new}
\end{figure}
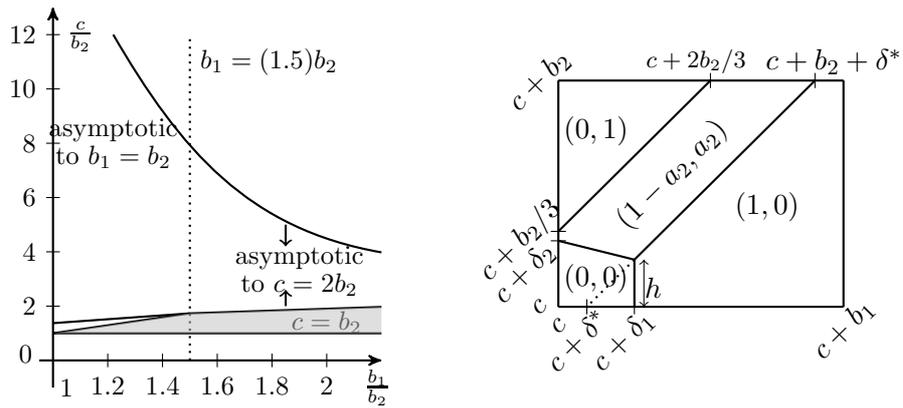

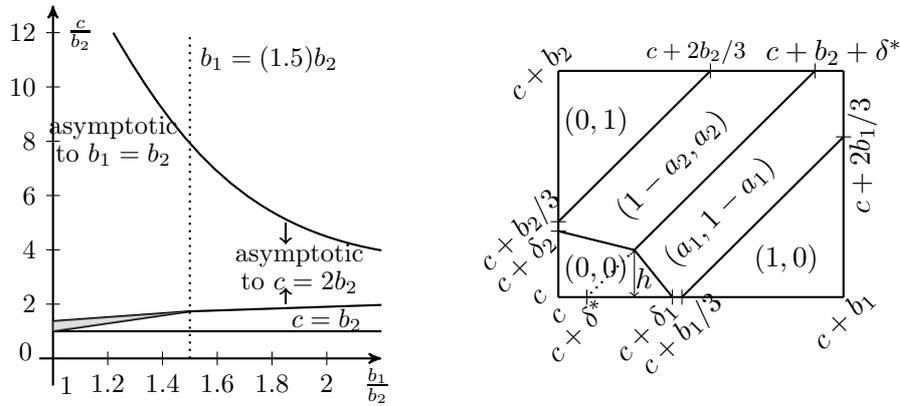
\begin{figure}[H]
\centering
\begin{minipage}{.45\textwidth}
\centering
\begin{tikzpicture}[scale=0.36,font=\small,axis/.style={very thick, ->, >=stealth'}]
\draw [axis,thick,->] (0,-1)--(0,13);
\node [right] at (11,-1) {$\frac{b_1}{b_2}$};
\draw [axis,thick,->] (-0.5,0)--(12,0);
\node [above] at (1,11) {$\frac{c}{b_2}$};
\node at (-1,0.25) {$0$};
\draw [thin,-] (-0.25,2) -- (0.25,2);
\node [left] at (-0.25,2) {$2$};
\draw [thin,-] (-0.25,4) -- (0.25,4);
\node [left] at (-0.25,4) {$4$};
\draw [thin,-] (-0.25,6) -- (0.25,6);
\node [left] at (-0.25,6) {$6$};
\draw [thin,-] (-0.25,8) -- (0.25,8);
\node [left] at (-0.25,8) {$8$};
\draw [thin,-] (-0.25,10) -- (0.25,10);
\node [left] at (-0.25,10) {$10$};
\draw [thin,-] (-0.25,12) -- (0.25,12);
\node [left] at (-0.25,12) {$12$};
\node at (0.5,-1) {$1$};
\draw [thin,-] (2,-0.25) -- (2,0.25);
\node [below] at (2,-0.25) {$1.2$};
\draw [thin,-] (4,-0.25) -- (4,0.25);
\node [below] at (4,-0.25) {$1.4$};
\draw [thin,-] (6,-0.25) -- (6,0.25);
\node [below] at (6,-0.25) {$1.6$};
\draw [thin,-] (8,-0.25) -- (8,0.25);
\node [below] at (8,-0.25) {$1.8$};
\draw [thin,-] (10,-0.25) -- (10,0.25);
\node [below] at (10,-0.25) {$2$};
\draw [thick,-] (0,1) to (12,1);
\draw [thick,-] (0,1) to (5,1.733) to (12,1.97);
\draw [thick,-] (0,1.372) to (5,1.733);
\draw [thick,-] (2.2,12) to[out=-60,in=170] (12,3.98);
\draw [thick,dotted] (5,0) to (5,12);

\node [right] at (5,11) {$b_1=(1.5)b_2$};
\node at (2.2,8.5) {asymptotic};
\node at (2.2,7.5) {to $b_1=b_2$};
\node at (9,3.8) {asymptotic};
\node at (9,2.8) {to $c=2b_2$};
\node at (10,1.4) {$c=b_2$};
\draw [thick,->] (8.5,2) -- (8.5,2.6);
\draw [thick,->] (8.5,5) -- (8.5,4.2);

\path[fill=gray!50,opacity=.5] (0,1) to (0,1.372) to (5,1.733)  to (0,1);
\end{tikzpicture}
\end{minipage}
\begin{minipage}{0.025\textwidth}
 \hspace*{0.025\textwidth}
\end{minipage}
\begin{minipage}{0.45\textwidth}
\centering
\begin{tikzpicture}[scale=0.25,font=\normalsize,axis/.style={very thick, -}]
\node [rotate=45] at (-1,0) {$c$};
\node [rotate=45] at (0,-1) {$c$};
\draw [axis,thick,-] (0,0)--(15,0);
\node [rotate=45] at (15,-1.2) {$c+b_1$};
\draw [axis,thick,-] (0,0)--(0,12);
\node [rotate=45] at (-1,12) {$c+b_2$};
\draw [axis,thick,-] (0,12)--(15,12);
\draw [axis,thick,-] (15,0)--(15,12);
\draw [axis,thick,-] (0,4)--(8,12);
\draw [thin,-] (-0.4,4)--(0.4,4);
\node [rotate=45] at (-2,3.5) {$c+b_2/3$};
\draw [axis,thick,-] (0,3.5)--(4,2.5);
\draw [thin,-] (-0.4,3.5)--(0.4,3.5);
\node [rotate=45] at (-1.8,1.8) {$c+\delta_2$};
\draw [thin,<->] (4,0)--(4,2.5);
\node at (4.5,1) {$h$};
\draw [axis,thick,-] (4,2.5)--(6,0);
\draw [axis,thick,-] (6.5,0)--(15,8.5);
\draw [thin,-] (6,-0.4)--(6,0.4);
\node [rotate=45] at (4.6,-1.8) {$c+\delta_1$};
\draw [thin,-] (6.5,-0.4)--(6.5,0.4);
\node [rotate=45] at (6.5,-2) {$c+b_1/3$};
\draw [thin,-] (14.6,8.5)--(15.4,8.5);
\node [rotate=90] at (16,7.5) {$c+2b_1/3$};
\draw [thick,-] (4,2.5)--(13.5,12);
\draw [thick,dotted] (4,2.5)--(1.5,0);
\draw [thin,-] (1.5,-0.4)--(1.5,0.4);
\node [rotate=45] at (1,-1.8) {$c+\delta^*$};
\draw [thin,-] (13.5,11.6)--(13.5,12.4);
\node at (14.5,13) {$c+b_2+\delta^*$};
\draw [thin,-] (8,11.6)--(8,12.4);
\node at (7.25,13) {\footnotesize$c+2b_2/3$};
\node at (2,1.5) {$(0,0)$};
\node at (12,2) {$(1,0)$};
\node [above] at (2,8) {$(0,1)$};
\node [rotate=45] at (6,7) {$(1-a_2,a_2)$};
\node [rotate=45] at (8.5,4.5) {$(a_1,1-a_1)$};
\end{tikzpicture}
\end{minipage}
\caption{When $(c,b_1,b_2)$ falls in the shaded region in the left, the optimal mechanism is as depicted in the right. Item $1$ is offered for a price of $c+b_1/3$, item $2$ is offered for a price of $c+b_2/3$, a lottery with probabilities $(1-a_2,a_2)$ is offered for a price of $c+a_2\delta_2$, and a lottery with probabilities $(a_1,1-a_1)$ is offered for a price of $c+a_1\delta_1$.}\label{fig:c-new}
\end{figure}

\begin{remark}
 The mechanisms depicted below in Figures \ref{fig:d-new} and \ref{fig:d'-new} differ only in that the line separating the regions with allocations $(1-a,a)$ and $(1,0)$ falls to the right of the line $z_1-z_2=b_1-b_2$ in the former, and to the left of it in the latter. These two structures meet at $b_1=3b_2/2$ when the line of separation exactly falls at $z_1-z_2=b_1-b_2$.
\end{remark}

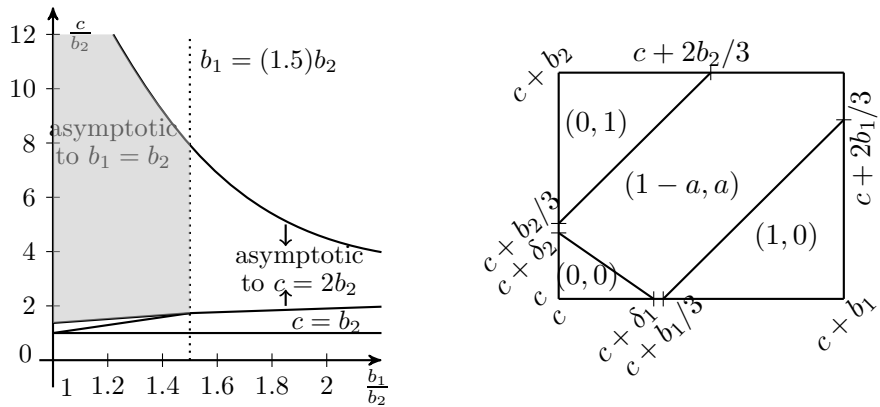
\begin{figure}[H]
\centering
\begin{minipage}{.45\textwidth}
\centering
\begin{tikzpicture}[scale=0.36,font=\small,axis/.style={very thick, ->, >=stealth'}]
\draw [axis,thick,->] (0,-1)--(0,13);
\node [right] at (11,-1) {$\frac{b_1}{b_2}$};
\draw [axis,thick,->] (-0.5,0)--(12,0);
\node [above] at (1,11) {$\frac{c}{b_2}$};
\node at (-1,0.25) {$0$};
\draw [thin,-] (-0.25,2) -- (0.25,2);
\node [left] at (-0.25,2) {$2$};
\draw [thin,-] (-0.25,4) -- (0.25,4);
\node [left] at (-0.25,4) {$4$};
\draw [thin,-] (-0.25,6) -- (0.25,6);
\node [left] at (-0.25,6) {$6$};
\draw [thin,-] (-0.25,8) -- (0.25,8);
\node [left] at (-0.25,8) {$8$};
\draw [thin,-] (-0.25,10) -- (0.25,10);
\node [left] at (-0.25,10) {$10$};
\draw [thin,-] (-0.25,12) -- (0.25,12);
\node [left] at (-0.25,12) {$12$};
\node at (0.5,-1) {$1$};
\draw [thin,-] (2,-0.25) -- (2,0.25);
\node [below] at (2,-0.25) {$1.2$};
\draw [thin,-] (4,-0.25) -- (4,0.25);
\node [below] at (4,-0.25) {$1.4$};
\draw [thin,-] (6,-0.25) -- (6,0.25);
\node [below] at (6,-0.25) {$1.6$};
\draw [thin,-] (8,-0.25) -- (8,0.25);
\node [below] at (8,-0.25) {$1.8$};
\draw [thin,-] (10,-0.25) -- (10,0.25);
\node [below] at (10,-0.25) {$2$};
\draw [thick,-] (0,1) to (12,1);
\draw [thick,-] (0,1) to (5,1.733) to (12,1.97);
\draw [thick,-] (0,1.372) to (5,1.733);
\draw [thick,-] (2.2,12) to[out=-60,in=170] (12,3.98);
\draw [thick,dotted] (5,0) to (5,12);

\node [right] at (5,11) {$b_1=(1.5)b_2$};
\node at (2.2,8.5) {asymptotic};
\node at (2.2,7.5) {to $b_1=b_2$};
\node at (9,3.8) {asymptotic};
\node at (9,2.8) {to $c=2b_2$};
\node at (10,1.4) {$c=b_2$};
\draw [thick,->] (8.5,2) -- (8.5,2.6);
\draw [thick,->] (8.5,5) -- (8.5,4.2);

\path[fill=gray!50,opacity=.5] (0,1.372) to (0,12) to (2.2,12) to[out=-60,in=130] (5,8) to (5,1.733)  to (0,1.372);
\end{tikzpicture}
\end{minipage}
\begin{minipage}{0.025\textwidth}
 \hspace*{0.025\textwidth}
\end{minipage}
\begin{minipage}{0.45\textwidth}
\centering
\begin{tikzpicture}[scale=0.25,font=\normalsize,axis/.style={very thick, -}]
\node [rotate=45] at (-1,0) {$c$};
\node [rotate=45] at (0,-1) {$c$};
\draw [axis,thick,-] (0,0)--(15,0);
\node [rotate=45] at (15,-1.2) {$c+b_1$};
\draw [axis,thick,-] (0,0)--(0,12);
\node [rotate=45] at (-1,12) {$c+b_2$};
\draw [axis,thick,-] (0,12)--(15,12);
\draw [axis,thick,-] (15,0)--(15,12);
\draw [axis,thick,-] (0,4)--(8,12);
\draw [thin,-] (-0.4,4)--(0.4,4);
\node [rotate=45] at (-2,3.5) {$c+b_2/3$};
\draw [axis,thick,-] (0,3.5)--(5,0);
\draw [thin,-] (-0.4,3.5)--(0.4,3.5);
\node [rotate=45] at (-1.8,1.8) {$c+\delta_2$};
\draw [thin,-] (5,-0.4)--(5,0.4);
\node [rotate=45] at (3.5,-1.8) {$c+\delta_1$};
\draw [thick,-] (5.5,0)--(15,9.5);
\draw [thin,-] (5.5,-0.4)--(5.5,0.4);
\node [rotate=45] at (5.5,-2) {$c+b_1/3$};
\draw [thin,-] (14.6,9.5)--(15.4,9.5);
\node [rotate=90] at (16,8) {$c+2b_1/3$};
\draw [thin,-] (8,11.6)--(8,12.4);
\node at (7,13) {$c+2b_2/3$};
\node at (1.5,1.2) {$(0,0)$};
\node [above] at (12,2) {$(1,0)$};
\node [above] at (2,8) {$(0,1)$};
\node at (6.5,6) {$(1-a,a)$};
\end{tikzpicture}
\end{minipage}
\caption{When $(c,b_1,b_2)$ falls in the shaded region in the left, the optimal mechanism is as depicted in the right. Item $1$ is offered for a price of $c+b_1/3$, item $2$ is offered for a price of $c+b_2/3$, and a lottery with probabilities $(1-a,a)$ is offered for a price of $c+a\delta_2$.}\label{fig:d-new}
\end{figure}

\begin{figure}[H]
\centering
\begin{minipage}{.45\textwidth}
\centering
\begin{tikzpicture}[scale=0.36,font=\small,axis/.style={very thick, ->, >=stealth'}]
\draw [axis,thick,->] (0,-1)--(0,13);
\node [right] at (11,-1) {$\frac{b_1}{b_2}$};
\draw [axis,thick,->] (-0.5,0)--(12,0);
\node [above] at (1,11) {$\frac{c}{b_2}$};
\node at (-1,0.25) {$0$};
\draw [thin,-] (-0.25,2) -- (0.25,2);
\node [left] at (-0.25,2) {$2$};
\draw [thin,-] (-0.25,4) -- (0.25,4);
\node [left] at (-0.25,4) {$4$};
\draw [thin,-] (-0.25,6) -- (0.25,6);
\node [left] at (-0.25,6) {$6$};
\draw [thin,-] (-0.25,8) -- (0.25,8);
\node [left] at (-0.25,8) {$8$};
\draw [thin,-] (-0.25,10) -- (0.25,10);
\node [left] at (-0.25,10) {$10$};
\draw [thin,-] (-0.25,12) -- (0.25,12);
\node [left] at (-0.25,12) {$12$};
\node at (0.5,-1) {$1$};
\draw [thin,-] (2,-0.25) -- (2,0.25);
\node [below] at (2,-0.25) {$1.2$};
\draw [thin,-] (4,-0.25) -- (4,0.25);
\node [below] at (4,-0.25) {$1.4$};
\draw [thin,-] (6,-0.25) -- (6,0.25);
\node [below] at (6,-0.25) {$1.6$};
\draw [thin,-] (8,-0.25) -- (8,0.25);
\node [below] at (8,-0.25) {$1.8$};
\draw [thin,-] (10,-0.25) -- (10,0.25);
\node [below] at (10,-0.25) {$2$};
\draw [thick,-] (0,1) to (12,1);
\draw [thick,-] (0,1) to (5,1.733) to (12,1.97);
\draw [thick,-] (0,1.372) to (5,1.733);
\draw [thick,-] (2.2,12) to[out=-60,in=170] (12,3.98);
\draw [thick,dotted] (5,0) to (5,12);

\node [right] at (5,11) {$b_1=(1.5)b_2$};
\node at (2.2,8.5) {asymptotic};
\node at (2.2,7.5) {to $b_1=b_2$};
\node at (9,3.8) {asymptotic};
\node at (9,2.8) {to $c=2b_2$};
\node at (10,1.4) {$c=b_2$};
\draw [thick,->] (8.5,2) -- (8.5,2.6);
\draw [thick,->] (8.5,5) -- (8.5,4.2);

\path[fill=gray!50,opacity=.5] (5,8) to[out=-50,in=170] (12,3.98) to (12,1.97) to (5,1.733) to (5,8);
\end{tikzpicture}
\end{minipage}
\begin{minipage}{0.025\textwidth}
 \hspace*{0.025\textwidth}
\end{minipage}
\begin{minipage}{0.45\textwidth}
\centering
\begin{tikzpicture}[scale=0.225,font=\normalsize,axis/.style={very thick, -}]
\node [rotate=45] at (-1,0) {$c$};
\node [rotate=45] at (0,-1) {$c$};
\draw [axis,thick,-] (0,0)--(18,0);
\node [rotate=45] at (18,-1) {$c+b_1$};
\draw [axis,thick,-] (0,0)--(0,12);
\node [rotate=45] at (-1,12) {$c+b_2$};
\draw [axis,thick,-] (0,12)--(18,12);
\draw [axis,thick,-] (18,0)--(18,12);
\draw [axis,thick,-] (0,4)--(8,12);
\draw [thin,-] (-0.4,4)--(0.4,4);
\node [rotate=45] at (-2,3.8) {$c+b_2/3$};
\draw [axis,thick,-] (0,3.5)--(5,0);
\draw [thin,-] (-0.4,3.5)--(0.4,3.5);
\node [rotate=45] at (-2,2) {$c+\delta_2$};
\draw [thin,-] (5,-0.4)--(5,0.4);
\node [rotate=45] at (3.8,-1.6) {$c+\delta_1$};
\draw [thick,-] (5.5,0)--(17.5,12);
\draw [thin,-] (5.5,-0.4)--(5.5,0.4);
\node [rotate=45] at (5.3,-2.7) {$c+\frac{b_1}{2}-\frac{b_2}{4}$};
\draw [thin,-] (17.5,11.6)--(17.5,12.4);
\node at (16,13.1) {$c+\frac{b_1}{2}+\frac{3b_2}{4}$};
\draw [thin,-] (8,11.6)--(8,12.4);
\node at (8,13) {\footnotesize$c+2b_2/3$};
\node at (1.5,1.2) {$(0,0)$};
\node [above] at (14,3) {$(1,0)$};
\node [above] at (2,8) {$(0,1)$};
\node at (6.5,5.5) {$(1-a,a)$};
\end{tikzpicture}
\end{minipage}
\caption{When $(c,b_1,b_2)$ falls in the shaded region in the left, the optimal mechanism is as depicted in the right. Item $1$ is offered for a price of $c+b_1/2-b_2/4$, item $2$ is offered for a price of $c+b_2/3$, and a lottery with probabilities $(1-a,a)$ is offered for a price of $c+a\delta_2$.}\label{fig:d'-new}
\end{figure}
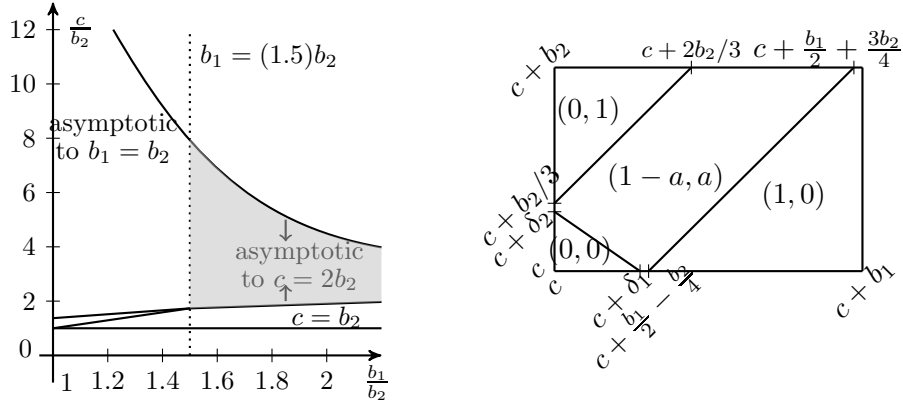

\begin{remark}
 Observe that the mechanisms depicted in Figures \ref{fig:b-new}, \ref{fig:c-new}, \ref{fig:d-new}, and \ref{fig:d'-new} meet at $b_1=3b_2/2$, $c=tb_2$. They meet because at this $(c,b_1,b_2)$, the parameter $h$ (in Figures \ref{fig:b-new} and \ref{fig:c-new}) becomes $0$, and $\delta^*=\delta_1=b_1/2-b_2/4=b_1/3=b_1-b_2$.
\end{remark}
\begin{remark}
 The mechanisms depicted below in Figures \ref{fig:e-new} and \ref{fig:e'-new} differ only in that the line separating the regions with allocations $(0,1)$ and $(1,0)$ falls to the right of the line $z_1-z_2=b_1-b_2$ in the former, and to the left of it in the latter. These two structures meet at $b_1=3b_2/2$ when the line of separation exactly falls at $z_1-z_2=b_1-b_2$.
\end{remark}

\begin{figure}[H]
\centering
\begin{minipage}{.45\textwidth}
\centering
\begin{tikzpicture}[scale=0.36,font=\small,axis/.style={very thick, ->, >=stealth'}]
\draw [axis,thick,->] (0,-1)--(0,13);
\node [right] at (11,-1) {$\frac{b_1}{b_2}$};
\draw [axis,thick,->] (-0.5,0)--(12,0);
\node [above] at (1,11) {$\frac{c}{b_2}$};
\node at (-1,0.25) {$0$};
\draw [thin,-] (-0.25,2) -- (0.25,2);
\node [left] at (-0.25,2) {$2$};
\draw [thin,-] (-0.25,4) -- (0.25,4);
\node [left] at (-0.25,4) {$4$};
\draw [thin,-] (-0.25,6) -- (0.25,6);
\node [left] at (-0.25,6) {$6$};
\draw [thin,-] (-0.25,8) -- (0.25,8);
\node [left] at (-0.25,8) {$8$};
\draw [thin,-] (-0.25,10) -- (0.25,10);
\node [left] at (-0.25,10) {$10$};
\draw [thin,-] (-0.25,12) -- (0.25,12);
\node [left] at (-0.25,12) {$12$};
\node at (0.5,-1) {$1$};
\draw [thin,-] (2,-0.25) -- (2,0.25);
\node [below] at (2,-0.25) {$1.2$};
\draw [thin,-] (4,-0.25) -- (4,0.25);
\node [below] at (4,-0.25) {$1.4$};
\draw [thin,-] (6,-0.25) -- (6,0.25);
\node [below] at (6,-0.25) {$1.6$};
\draw [thin,-] (8,-0.25) -- (8,0.25);
\node [below] at (8,-0.25) {$1.8$};
\draw [thin,-] (10,-0.25) -- (10,0.25);
\node [below] at (10,-0.25) {$2$};
\draw [thick,-] (0,1) to (12,1);
\draw [thick,-] (0,1) to (5,1.733) to (12,1.97);
\draw [thick,-] (0,1.372) to (5,1.733);
\draw [thick,-] (2.2,12) to[out=-60,in=170] (12,3.98);
\draw [thick,dotted] (5,0) to (5,12);

\node [right] at (5,11) {$b_1=(1.5)b_2$};
\node at (2.2,8.5) {asymptotic};
\node at (2.2,7.5) {to $b_1=b_2$};
\node at (9,3.8) {asymptotic};
\node at (9,2.8) {to $c=2b_2$};
\node at (10,1.4) {$c=b_2$};
\draw [thick,->] (8.5,2) -- (8.5,2.6);
\draw [thick,->] (8.5,5) -- (8.5,4.2);

\path[fill=gray!50,opacity=.5] (2.2,12) to[out=-60,in=130] (5,8) to (5,12) to (2.2,12);
\end{tikzpicture}
\end{minipage}
\begin{minipage}{0.025\textwidth}
 \hspace*{0.025\textwidth}
\end{minipage}
\begin{minipage}{0.45\textwidth}
\centering
\begin{tikzpicture}[scale=0.25,font=\normalsize,axis/.style={very thick, -}]
\node [rotate=45] at (-1,0) {$c$};
\node [rotate=45] at (0,-1) {$c$};
\draw [axis,thick,-] (0,0)--(15,0);
\node [rotate=45] at (15,-1.2) {$c+b_1$};
\draw [axis,thick,-] (0,0)--(0,12);
\node [rotate=45] at (-1,12) {$c+b_2$};
\draw [axis,thick,-] (0,12)--(15,12);
\draw [axis,thick,-] (15,0)--(15,12);
\draw [thick,-] (5,0)--(15,10);
\draw [thin,-] (5,-0.4)--(5,0.4);
\node [rotate=45] at (5,-2) {$c+b_1/3$};
\draw [thin,-] (14.6,10)--(15.4,10);
\node [rotate=90] at (16,9) {$c+2b_1/3$};
\node [above] at (12,3) {$(1,0)$};
\node [above] at (5,6) {$(0,1)$};
\end{tikzpicture}
\end{minipage}
\caption{When $(c,b_1,b_2)$ falls in the shaded region in the left, the optimal mechanism is as depicted in the right. Item $1$ is offered for a price of $c+b_1/3$, and item $2$ is offered for a price of $c$.}\label{fig:e-new}
\end{figure}
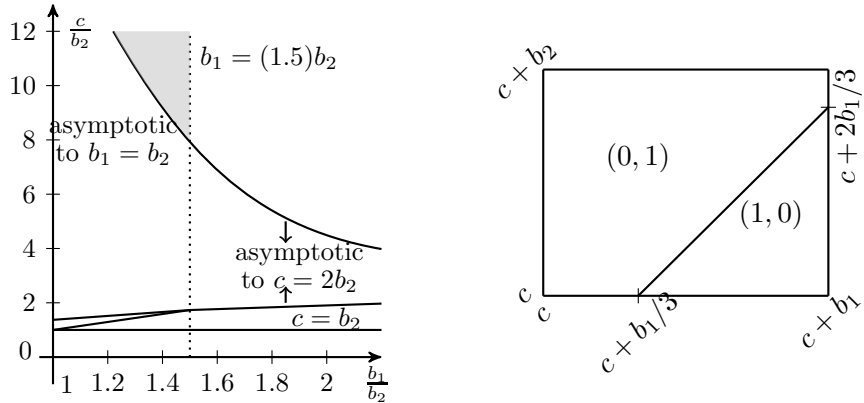

\begin{figure}[H]
\centering
\begin{minipage}{.45\textwidth}
\centering
\begin{tikzpicture}[scale=0.36,font=\small,axis/.style={very thick, ->, >=stealth'}]
\draw [axis,thick,->] (0,-1)--(0,13);
\node [right] at (11,-1) {$\frac{b_1}{b_2}$};
\draw [axis,thick,->] (-0.5,0)--(12,0);
\node [above] at (1,11) {$\frac{c}{b_2}$};
\node at (-1,0.25) {$0$};
\draw [thin,-] (-0.25,2) -- (0.25,2);
\node [left] at (-0.25,2) {$2$};
\draw [thin,-] (-0.25,4) -- (0.25,4);
\node [left] at (-0.25,4) {$4$};
\draw [thin,-] (-0.25,6) -- (0.25,6);
\node [left] at (-0.25,6) {$6$};
\draw [thin,-] (-0.25,8) -- (0.25,8);
\node [left] at (-0.25,8) {$8$};
\draw [thin,-] (-0.25,10) -- (0.25,10);
\node [left] at (-0.25,10) {$10$};
\draw [thin,-] (-0.25,12) -- (0.25,12);
\node [left] at (-0.25,12) {$12$};
\node at (0.5,-1) {$1$};
\draw [thin,-] (2,-0.25) -- (2,0.25);
\node [below] at (2,-0.25) {$1.2$};
\draw [thin,-] (4,-0.25) -- (4,0.25);
\node [below] at (4,-0.25) {$1.4$};
\draw [thin,-] (6,-0.25) -- (6,0.25);
\node [below] at (6,-0.25) {$1.6$};
\draw [thin,-] (8,-0.25) -- (8,0.25);
\node [below] at (8,-0.25) {$1.8$};
\draw [thin,-] (10,-0.25) -- (10,0.25);
\node [below] at (10,-0.25) {$2$};
\draw [thick,-] (0,1) to (12,1);
\draw [thick,-] (0,1) to (5,1.733) to (12,1.97);
\draw [thick,-] (0,1.372) to (5,1.733);
\draw [thick,-] (2.2,12) to[out=-60,in=170] (12,3.98);
\draw [thick,dotted] (5,0) to (5,12);

\node [right] at (5,11) {$b_1=(1.5)b_2$};
\node at (2.2,8.5) {asymptotic};
\node at (2.2,7.5) {to $b_1=b_2$};
\node at (9,3.8) {asymptotic};
\node at (9,2.8) {to $c=2b_2$};
\node at (10,1.4) {$c=b_2$};
\draw [thick,->] (8.5,2) -- (8.5,2.6);
\draw [thick,->] (8.5,5) -- (8.5,4.2);

\path[fill=gray!50,opacity=.5] (5,12) to (5,8) to[out=-50,in=170] (12,3.98) to (12,12) to (5,12);
\end{tikzpicture}
\end{minipage}
\begin{minipage}{0.025\textwidth}
 \hspace*{0.025\textwidth}
\end{minipage}
\begin{minipage}{0.45\textwidth}
\centering
\begin{tikzpicture}[scale=0.225,font=\normalsize,axis/.style={very thick, -}]
\node [rotate=45] at (-1,0) {$c$};
\node [rotate=45] at (0,-1) {$c$};
\draw [axis,thick,-] (0,0)--(18,0);
\node [rotate=45] at (18,-1) {$c+b_1$};
\draw [axis,thick,-] (0,0)--(0,12);
\node [rotate=45] at (-1,12) {$c+b_2$};
\draw [axis,thick,-] (0,12)--(18,12);
\draw [axis,thick,-] (18,0)--(18,12);
\draw [thick,-] (5,0)--(17,12);
\draw [thin,-] (5,-0.4)--(5,0.4);
\node [rotate=45] at (4.8,-2.7) {$c+\frac{b_1}{2}-\frac{b_2}{4}$};
\draw [thin,-] (17,11.6)--(17,12.4);
\node at (16,13.1) {$c+\frac{b_1}{2}+\frac{3b_2}{4}$};
\node [above] at (14,3) {$(1,0)$};
\node [above] at (5,6) {$(0,1)$};
\end{tikzpicture}
\end{minipage}
\caption{When $(c,b_1,b_2)$ falls in the shaded region in the left, the optimal mechanism is as depicted in the right. Item $1$ is offered for a price of $c+b_1/2-b_2/4$, and item $2$ is offered for a price of $c$.}\label{fig:e'-new}
\end{figure}
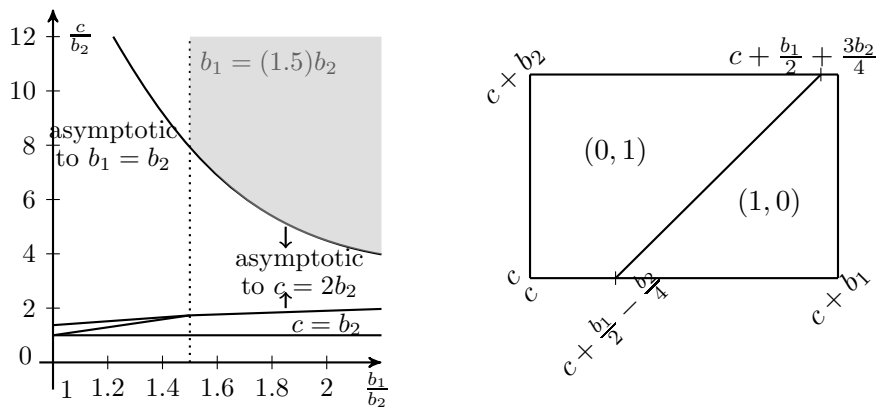

\begin{remark}
  The mechanisms depicted in Figures \ref{fig:d-new}, \ref{fig:d'-new}, \ref{fig:e-new}, and \ref{fig:e'-new} meet at $b_1=3b_2/2$, $c=(243/38)b_2$. They meet because at this $(c,b_1,b_2)$, the parameter $a$ (in Figures \ref{fig:d-new} and \ref{fig:d'-new}) becomes $0$, and $b_1/2-b_2/4=b_1/3=b_1-b_2$.
\end{remark}

\begin{remark}\label{rem:armstrong}
The mechanisms in Figures \ref{fig:e-new} and \ref{fig:e'-new} show an interesting result -- the existence of an optimal multi-dimensional mechanism without an exclusion region. An intuitive explanation for the absence of exclusion region in Figure \ref{fig:e'-new} is as follows. Consider the case where the seller offers each allocation with a small increase in price, say $\epsilon$. The seller then loses a revenue of $c$ from the valuations $\{z:u(z)\leq\epsilon\}$, and gains an extra revenue of $\epsilon$ from the valuations $\{z:u(z)\geq\epsilon\}$. The mechanism will have no exclusion region when the loss dominates the gain. Observe that the expected loss in revenue is \begin{eqnarray*}
c\cdot Pr(\{z:u(z)\leq\epsilon\}) & = & \frac{c}{b_1b_2}(\epsilon(b_1/2-b_2/4+\epsilon))+\frac{(b_1/2-b_2/4)}{b_1b_2}\frac{\epsilon^2}{2}\\
& \approx & \frac{c}{b_1b_2}(\epsilon(b_1/2-b_2/4)),
\end{eqnarray*} 
and that the expected gain in revenue is $$\epsilon\cdot Pr(\{u(z)\geq\epsilon\})=\epsilon\cdot(1-Pr(\{u(z)\leq\epsilon\}))\approx\epsilon.$$ The loss dominates the gain when $c \geq \frac{4b_1b_2}{2b_1-b_2}$. (The actual threshold will depend on more precise calculations than our order estimates.) Observe that both the loss and the gain are of the order of $\epsilon$, which explains the possibility of the loss dominating the gain at very high values of $c$. Figure \ref{fig:e-new} has no exclusion region due to a similar reason.
\end{remark}

\begin{remark}\label{rem:role}
The notations $\delta_1$, $\delta_2$, and $\delta^*$, used in various mechanism depictions, can be understood as follows. (i) The first transition from $q=(0,0)$ on the bottom boundary of $D$ occurs at $\delta=\delta_1$. (ii) Similarly, the first transition on the left boundary of $D$ occurs at $\delta=-\delta_2$. (iii) The final transition of $q$ on the top/right boundary of $D$ (in mechanisms depicted in Figures \ref{fig:a-new}--\ref{fig:c-new}) occurs at $\delta=\delta^*$. 
\end{remark}

For a summarizing phase diagram see Figure \ref{fig:phase-diagram}. To see a portrayal of all possible structures that an optimal mechanism can take, see Figures \ref{fig:a-ini}--\ref{fig:e'-ini}.

We now proceed to prove Theorem \ref{thm:consolidate}. We consider every structure separately, and go through the following steps in order to prove that the optimal mechanism has the specific structure.
\begin{enumerate}
\item[\bf{Step 1:}] We compute the virtual valuation function $V(\delta)$ for every $\delta\in[-b_2,b_1]$.
\item[\bf{Step 2:}] We find the relation between the variables of interest, ($\delta_1$, $\delta_2$, $\delta^*$, $h$, $a_1$, $a_2$), using the equality conditions in Theorem \ref{thm:Myerson}.
\item[\bf{Step 3:}] We prove that the solution that satisfies the relations obtained in Step 2 are indeed meaningful, by evaluating bounds for the variables of interest.
\item[\bf{Step 4:}] We verify that all the inequality conditions of Theorem \ref{thm:Myerson} hold. The bounds evaluated in Step 3 are crucially used in this process of verification.
\end{enumerate}

We now proceed to prove parts 1(a) and 2(a) of Theorem \ref{thm:consolidate}.
\begin{theorem}\label{thm:menu-1}
 Let $c\in[0,b_2]$. Then the optimal mechanism is as depicted in Figure \ref{fig:a-ini} (see also Figure \ref{fig:a-new}). The values of $\delta_1$ and $\delta_2$ are computed by solving the following equations simultaneously.
 \begin{align}
  &-3\delta_1\delta_2-c(\delta_1+\delta_2)+b_1b_2=0.\label{eqn:fig-a-first}\\
  &-\frac{3}{2}\delta_2^2+2b_2\delta_2-\frac{b_2^2}{2}+(c-2b_2+3\delta_2)\delta_1=0.\label{eqn:fig-a-second}
 \end{align}
\end{theorem}
\begin{proof}
{\bf Step 1:} We compute the virtual valuation function for the mechanism depicted in Figure \ref{fig:a-new}. Since $\bar{\mu}(D)=0$, we compute $V$ using the formula
\begin{equation}\label{eqn:V-alternate}
  V(\delta)=-\bar{\mu}(\{z:z_1-z_2<\delta\}\cup Z).
\end{equation}
\begin{equation}\label{eqn:V-fig-a}
V(\delta)=\frac{1}{b_1b_2}\begin{cases}\bar{\mu}(Z)+\frac{3}{2}\delta^2+2b_2\delta+\frac{b_2^2}{2}&\delta\in[-b_2,-\delta_2]\\V(-\delta_2)-(c-2b_2+3\delta_2)(\delta+\delta_2)&\delta\in[-\delta_2,\delta^*]\\V(\delta^*)-(c-2b_2)(\delta-\delta^*)+\frac{3}{2}((\delta_1-\delta)^2-\delta_2^2)&\delta\in[\delta^*,b']\\V(b')-(c-2b_1+3\delta_1)(\delta-b_1+b_2)&\delta\in[b',\delta_1]\\-\frac{3}{2}\delta^2+2b_1\delta-\frac{b_1^2}{2}&\delta\in[\delta_1,b_1]\end{cases}
\end{equation}
where $b_1-b_2$ is denoted as $b'$. For ease of notation, we drop the factor $\frac{1}{b_1b_2}$ in the rest of the paper. 

{\bf Step 2:} The mechanism has three unknowns: $\delta^*$, $\delta_1$, and $\delta_2$. Observe that the line between the points $(c+b_2+\delta^*,c+b_2)$ and $(c+\delta^*,c)$ passes through $(c+\delta_1,c+\delta_2)$. So we have $\delta^*=\delta_1-\delta_2$.

We now proceed to compute $\delta_1$ and $\delta_2$. We do so by equating $\bar{\mu}(Z)=0$ and $V(\delta^*)=0$. The latter follows from Theorem \ref{thm:Myerson} because $q_1=0$ for $\delta\in[-b_2,\delta^*]$. We thus obtain equations (\ref{eqn:fig-a-first}) and (\ref{eqn:fig-a-second}).

{\bf Step 3:} We now show that there exists a meaningful solution $(\delta_1,\delta_2)$ that simultaneously solves (\ref{eqn:fig-a-first}) and (\ref{eqn:fig-a-second}). Specifically, we show that there exists a $(\delta_1,\delta_2)\in[\frac{b_1}{2}-\frac{b_2}{6},\frac{2b_1-c}{3}]\times[\frac{b_2}{3},\frac{2b_2-c}{3}]$ as a simultaneous solution to (\ref{eqn:fig-a-first}) and (\ref{eqn:fig-a-second}). To show this, we do the following.
\begin{itemize}
\item We first define $\delta_1|_{\delta_2=x}$ to be the value of $\delta_1$ that satisfies (\ref{eqn:fig-a-first}) when $\delta_2=x$ and $\delta_2|_{\delta_1=x}$ to be the value of $\delta_2$ that satisfies (\ref{eqn:fig-a-first}) when $\delta_1=x$. We then show that there exists a $(\delta_1,\delta_2)\in[\frac{b_1}{2}-\frac{b_2}{6},\frac{2b_1-c}{3}]\times[\frac{b_2}{3},\frac{2b_2-c}{3}]$ satisfying (\ref{eqn:fig-a-first}). We do this by showing that (a) $\delta_1|_{\delta_2=x}$ is continuous in $x$, (b) $\delta_1|_{\delta_2=\frac{b_2}{3}}\geq\frac{b_1}{2}-\frac{b_2}{6}$, and (c) $\delta_1|_{\delta_2=\frac{2b_2-c}{3}}\leq\frac{2b_1-c}{3}$. We further show that in addition to continuity, $\delta_1|_{\delta_2=x}$ is also monotone; it decreases as $x$ increases.
\item It now suffices to show that the entry and the exit points of the curve $(\delta_1|_{\delta_2=x},x)$ in the rectangle $[\frac{b_1}{2}-\frac{b_2}{6},\frac{2b_1-c}{3}]\times[\frac{b_2}{3},\frac{2b_2-c}{3}]$ changes sign when substituted on the left-hand side of (\ref{eqn:fig-a-second}). The possible entry points are $(\frac{b_1}{2}-\frac{b_2}{6},\delta_2|_{\delta_1=\frac{b_1}{2}-\frac{b_2}{6}})$ and $(\delta_1|_{\delta_2=\frac{2b_2-c}{3}},\frac{2b_2-c}{3})$; we substitute the entry points on left-hand side of (\ref{eqn:fig-a-second}) and show that the expression is nonnegative in both cases. Similarly, the possible exit points are $(\delta_1|_{\delta_2=\frac{b_2}{3}},\frac{b_2}{3})$ and $(\frac{2b_1-c}{3},\delta_2|_{\delta_1=\frac{2b_1-c}{3}})$; we substitute the exit points on left-hand side of (\ref{eqn:fig-a-second}) and show that the expression is nonpositive in both cases.
\end{itemize}
We now fill in the details. We have $\delta_1|_{\delta_2}=\frac{b_1b_2-c\delta_2}{3\delta_2+c}$ and  $\delta_2|_{\delta_1}=\frac{b_1b_2-c\delta_1}{3\delta_1+c}$ from (\ref{eqn:fig-a-first}). It is clear that $\delta_1|_{\delta_2=x}$ is continuous, and also monotonically decreases in $x$. We now verify that $\delta_1|_{\delta_2=\frac{b_2}{3}}\geq\frac{b_1}{2}-\frac{b_2}{6}$; indeed,
$$
  \frac{b_1b_2-cb_2/3}{c+b_2}\geq\frac{b_1b_2-cb_2/3}{2b_2}\geq\frac{b_1b_2-b_2^2/3}{2b_2}=\frac{b_1}{2}-\frac{b_2}{6},
$$
where both the inequalities hold because $c\leq b_2$. We now verify that $\delta_1|_{\delta_2=\frac{2b_2-c}{3}}\leq\frac{2b_1-c}{3}$:
$$
  \frac{b_1b_2-c(2b_2-c)/3}{2b_2}\leq\frac{4b_1b_2/3-2b_2c/3}{2b_2}=\frac{2b_1-c}{3},
$$
where the inequality $c^2\leq b_1b_2$ holds because of $c\leq b_2\leq b_1$.

We now consider the points $(\delta_1|_{\delta_2=\frac{2b_2-c}{3}},\frac{2b_2-c}{3})$ and $(\delta_1|_{\delta_2=\frac{b_2}{3}},\frac{b_2}{3})$. Substituting $\delta_1=\frac{b_1b_2-c\delta_2}{c+3\delta_2}$ in (\ref{eqn:fig-a-second}), we obtain
\begin{equation}\label{eqn:fig-a-delta_2}
  -\frac{9}{2}\delta_2^3+\delta_2^2(6b_2-\frac{9}{2}c)+\delta_2(4b_2 c-c^2-\frac{3}{2}b_2^2+3b_1b_2)-\frac{1}{2}b_2^2c+b_1b_2c-2b_1b_2^2=0.
\end{equation}
When $\delta_2=\frac{2b_2-c}{3}$, the left-hand side of (\ref{eqn:fig-a-delta_2}) equals $\frac{1}{3}b_2(b_2^2-c^2)\geq 0$, and when $\delta_2=\frac{b_2}{3}$, it equals $-b_2(b_1-c/3)(b_2-c)\leq 0$.

We now consider the points $(\frac{2b_1-c}{3},\delta_2|_{\delta_1=\frac{2b_1-c}{3}})$ and $(\frac{b_1}{2}-\frac{b_2}{6},\delta_2|_{\delta_1=\frac{b_1}{2}-\frac{b_2}{6}})$. Substituting $\delta_2=\frac{b_1b_2-c\delta_1}{3\delta_1+c}$ in (\ref{eqn:fig-a-second}), we obtain
 \begin{multline}\label{eqn:fig-a-delta_1}
   -\frac{3}{2}b_1^2b_2^2+2b_1b_2^2c-\frac{1}{2}b_2^2c^2+(6b_1b_2^2+6b_1b_2c-3b_2^2c-4b_2c^2+c^3)\delta_1\\+(9b_1b_2-\frac{9}{2}
 b_2^2-18b_2c+\frac{3}{2}c^2)\delta_1^2-18b_2\delta_1^3=0.
 \end{multline}
When $\delta_1=\frac{2b_1-c}{3}$, the left-hand side of (\ref{eqn:fig-a-delta_1}) equals $\frac{1}{6}(-8b_1^3b_2+3b_1^2b_2^2+4b_1^2c^2+2b_1b_2c^2-c^4)$. We claim that this expression is negative for $b_1\geq b_2$, $c\in[0,b_2]$. Observe that its derivative with respect to $c$ satisfies $4c(b_1(2b_2+b_2)-c^2)\geq 0$ for all $c\in[0,b_2]$, and thus the expression attains its maximum when $c=b_2$. At $c=b_2$, the expression equals $b_2(b_1-b_2)(-8b_1^2-b_1b_2+b_2^2)$ which clearly is nonpositive when $b_1\geq b_2$. We have proved our claim.

Now when $\delta_1=\frac{b_1}{2}-\frac{b_2}{6}$, the left-hand side of (\ref{eqn:fig-a-delta_1}) equals 
 \begin{eqnarray*}
 \frac{1}{24}(b_2-c)(27b_1^2b_2-18b_1b_2^2-b_2^3+(42b_1b_2-9b_1^2-b_2^2)c+4(b_2-3b_1)c^2) \\
 =\frac{1}{24}(b_2-c)(A_0+A_1c+A_2c^2).
 \end{eqnarray*}
Observe that we have a quadratic expression in $c$, with $A_2$ being negative. So to prove that this quadratic expression is nonnegative for $c\in[0,b_2]$, it suffices to prove that it is nonnegative at $c=0$ and $c=b_2$. At $c=0$, the expression equals $27b_1^2b_2-18b_1b_2^2-b_2^3\geq 0$ for $b_1\geq b_2$, and at $c=b_2$, it equals $18b_1^2b_2+12b_1b_2^2+2b_2^3\geq 0$.

We have thus shown that there exists a solution $(\delta_1\,\delta_2)\in[\frac{b_1}{2}-\frac{b_2}{6},\frac{2b_1-c}{3}]\times[\frac{b_2}{3},\frac{2b_2-c}{3}]$ that simultaneously solves (\ref{eqn:fig-a-first}) and (\ref{eqn:fig-a-second}), for every $c\in[0,b_2]$ and $b_1\geq b_2$.

{\bf Step 4:} We now proceed to prove parts (c) and (d) in Theorem \ref{thm:Myerson}(2) and \ref{thm:Myerson}(4). Observe that the proof is complete if we prove that $V(\delta)\leq 0$ when $\delta\in[-b_2,\delta^*]$, and $V(\delta)\geq 0$ when $\delta\in[\delta^*,b_1]$. We now compute $V'(\delta)$ for almost every $\delta\in[-b_2,b_1]$.

\begin{equation}\label{eqn:menu-1-V'}
  V'(\delta)=\begin{cases}3\delta+2b_2&\delta\in(-b_2,-\delta_2)\\-(c-2b_2+3\delta_2)&\delta\in(-\delta_2,\delta^*]\\-(c-2b_2)-3(\delta_1-\delta)&\delta\in[\delta^*,b_1-b_2)\\-(c-2b_1+3\delta_1)&\delta\in(b_1-b_2,\delta_1)\\-3\delta+2b_1&\delta\in(\delta_1,b_1).\end{cases}
\end{equation}

Observe that $V'(\delta)$ is negative when $\delta\in[-b_2,-\frac{2b_2}{3}]$, and positive when $\delta\in[-\frac{2b_2}{3},\delta^*]$ (follows because $\delta_2\leq\frac{2b_2-c}{3})$. We also have $V(-b_2)=V(\delta^*)=0$. So $V(\delta)=V(-b_2)+\int_{-b_2}^{\delta}V'(\tilde{\delta})\,d\tilde{\delta}\leq 0$ for all $\delta\in[-b_2,\delta^*]$, and hence $\int_{-b_2}^{\delta^*}V(\delta)\,d\delta\leq 0$, and $\int_{-b_2}^xV(\delta)\,d\delta\geq\int_{-b_2}^{\delta^*}V(\delta)\,d\delta$ for all $x\in[-b_2,\delta^*]$.

We now claim that $V'(\delta)$ is positive when $\delta\in[\delta^*,\frac{2b_1}{3}]$, and negative when $\delta\in[\frac{2b_1}{3},b_1]$. Observe that $V'(\delta)$ is continuous at $\delta=\delta^*$, and that it increases in the interval $[\delta^*,b_1-b_2]$. So $V'(\delta)\geq 0$ when $\delta\in[\delta^*,b_1-b_2]$. Also, $V'(\delta)\geq 0$ when $\delta\in[b_1-b_2,\delta_1]$ because $\delta_1\leq\frac{2b_1-c}{3}$. That $V'(\delta)$ is positive when $\delta\in[\delta_1,\frac{2b_1}{3}]$, and negative when $\delta\in[\frac{2b_1}{3},b_1]$ is obvious. We have proved our claim.

Since we also have $V(\delta^*)=V(b_1)=0$, it follows that $V(\delta)=V(\delta^*)+\int_{\delta^*}^{\delta}V'(\tilde{\delta})\,d\tilde{\delta}\geq 0$ for all $\delta\in[\delta^*,b_1]$. So we have $\int_{\delta^*}^{b_1}V(\delta)\,d\delta\geq 0$ and $\int_x^{b_1}V(\delta)\,d\delta\leq\int_{\delta^*}^{b_1}V(\delta)\,d\delta$ for all $x\in[\delta^*,b_1]$.\qed
\end{proof}

With the above theorem, we have completely solved the $c\leq b_2$ case. We now analyze the case at which the transition occurs. At $c=b_2$, when we solve (\ref{eqn:fig-a-first}) and (\ref{eqn:fig-a-second}) simultaneously, we obtain $\delta_2=\frac{b_2}{3}=\frac{2b_2-c}{3}$ and $\delta_1=\frac{b_1}{2}-\frac{b_2}{6}$. When $c>b_2$, the left-hand side of (\ref{eqn:fig-a-delta_2}) still continues to change sign at $\delta_2=\frac{b_2}{3}$ and $\delta_2=\frac{2b_2-c}{3}$, but since $\frac{b_2}{3}>\frac{2b_2-c}{3}$, the solution $\delta_2$ now belongs to the interval $[\frac{2b_2-c}{3},\frac{b_2}{3}]$. We thus have (i) $V(-\frac{b_2}{3})=0=V(\delta^*)$, and (ii) $V'(\delta)\geq 0$ when $\delta\in[-\frac{2b_2}{3},-\delta_2]$ and $V'(\delta)\leq 0$ when $\delta\in[-\delta_2,\delta^*]$. These both imply that $V(\delta)\geq 0$ when $\delta\in[-\frac{b_2}{3},\delta^*]$. So the minimum of $\int_{-b_2}^{x}V(\delta)\,d\delta$ can never occur at $x=\delta^*$, causing the condition in part (d) of Theorem \ref{thm:Myerson}(2) to fail.

At $c=b_2$, a transition occurs from the structure depicted in Figure \ref{fig:a-ini} to that in Figure \ref{fig:b-ini}. We now proceed to prove the optimality of the structure in \ref{fig:b-ini}, i.e., parts 1(b) and 2(b) in Theorem \ref{thm:consolidate}.
\begin{theorem}\label{thm:menu-2}
Let $c\in[b_2,\beta]$ if $b_1\geq 3b_2/2$ and let $c\in[b_2,\alpha_1]$ if $b_1\in[b_2,3b_2/2]$ with $\alpha_1$ and $\beta$ as defined in Theorem \ref{thm:consolidate}. Then, the optimal mechanism is as depicted in Figure \ref{fig:b-ini} (see also Figure \ref{fig:b-new}). The values of $h$ and $\delta^*$ are obtained by solving (\ref{eqn:fig-b-first}) and (\ref{eqn:fig-b-second}) simultaneously, and the values of $(\delta_1,\delta_2)$ are given by
$$
  (\delta_1,\delta_2)=\left(h+\delta^*,\frac{b_1b_2-(3h/2+c)(h+\delta^*)}{3/2(h+\delta^*)+c}\right).
$$
The probability of allocation $a_2$ is given by $a_2=\frac{h+\delta^*}{\delta_2+\delta^*}$.
\end{theorem}
\begin{proof}
{\bf Step 1:} We compute the virtual valuation function for the mechanism depicted in Figure \ref{fig:b-new}.
$$
  V(\delta)=\begin{cases}V(-\delta_2)-(c-2b_2+3\delta_2)(\delta+\delta_2)+\frac{3}{2}\frac{\delta_2-h}{\delta_2+\delta^*}(\delta+\delta_2)^2&\delta\in[-\delta_2,\delta^*]\\V(\delta^*)-(c-2b_2)(\delta-\delta^*)+\frac{3}{2}((\delta_1-\delta)^2-h^2)&\delta\in[\delta^*,b']\end{cases}
$$
where $b_1-b_2$ is denoted by $b'$. The expression for $V(\delta)$ when $\delta\in[-b_2,-\delta_2]\cup[b_1-b_2,b_1]$ remains the same as in (\ref{eqn:V-fig-a}).

{\bf Step 2:} The mechanism has five parameters: $h$, $\delta^*$, $\delta_1$, $\delta_2$, and $a_2$. Observe that the $45^\circ$ line segment joining the points $(c+b_2+\delta^*,c+b_2)$ and $(c+\delta^*,c)$ passes through $(c+\delta_1,c+h)$. So we have $\delta_1=h+\delta^*$. Since $q=\nabla u$, a conservative field, we must have the slope of the line separating $(0,0)$ and $(1-a_2,a_2)$ allocation regions satisfying $-\frac{1-a_2}{a_2}=\frac{h-\delta_2}{h+\delta^*}$. This yields $a_2=\frac{h+\delta^*}{\delta_2+\delta^*}$.

We now proceed to compute $h$, $\delta_2$ and $\delta^*$. We do so by equating $\bar{\mu}(Z)=0$, $V(\delta^*)=0$, and $\int_{-\frac{b_2}{3}}^{\delta^*}V(\delta)\,d\delta=0$. The latter two conditions follow from Theorem \ref{thm:Myerson} 3(b) and 3(c) because $q_1(\delta)=1-a_2\in(0,1)$ for $\delta\in[-\frac{b_2}{3},\delta^*]$. We then have the following implications.
\begin{equation}
 \bar{\mu}(Z)=0\Rightarrow-\frac{3}{2}(h+\delta^*)(h+\delta_2)-c(\delta_2+h+\delta^*)+b_1b_2=0.\label{eqn:fig-b-delta2}
\end{equation}
From (\ref{eqn:V-alternate}), we see that $V(\delta^*)$ is the negative of $\bar{\mu}$ measure of the nonconvex pentagon bound by $(c,c)$, $(c,c+b_2)$, $(c+b_2+\delta^*,c+b_2)$, $(c+\delta_1,c+h)$, and $(c+\delta_1,c)$. Thus
\begin{align}
 V(\delta^*)=0&\Rightarrow-\frac{3}{2}h^2-ch-\frac{3}{2}b_2(b_2+2\delta^*)+b_2(b_2+\delta^*)+b_1b_2=0\label{eqn:fig-b-h-initial}\\&\Rightarrow h=\frac{-c+\sqrt{c^2+3b_2(2b_1-b_2-4\delta^*)}}{3}.\label{eqn:fig-b-h}
\end{align}
Next,
\begin{align}
 &\int_{-\frac{b_2}{3}}^{\delta^*}V(\delta)\,d\delta=0\Rightarrow\int_{-\frac{b_2}{3}}^{-\delta_2}V(\delta)\,d\delta+\int_{-\delta_2}^{\delta^*}V(\delta)\,d\delta=0\nonumber\\&\Rightarrow b_2(\delta_2^2-b_2^2/9)+\frac{1}{2}(b_2^3/27-\delta_2^3)+b_2^2/2(b_2/3-\delta_2)\nonumber\\&\hspace*{.1in}-(2b_2\delta_2-3\delta_2^2/2-b_2^2/2)(\delta^*+\delta_2)-(c-2b_2+2\delta_2+h)(\delta^*+\delta_2)^2/2=0\nonumber\\&\Rightarrow\frac{1}{54}(4b_2+3\delta^*)(b_2+3\delta^*)^2-\frac{(c+h+\delta^*)}{2}(\delta^*+\delta_2)^2=0.\label{eqn:fig-b-last}
\end{align}

The values of $h$, $\delta^*$, and $\delta_2$ can be obtained by solving (\ref{eqn:fig-b-delta2}), (\ref{eqn:fig-b-h}), and (\ref{eqn:fig-b-last}) simultaneously. We now proceed to prove that $(h,\delta^*)$ can be obtained by solving (\ref{eqn:fig-b-first}) and (\ref{eqn:fig-b-second}) simultaneously. From (\ref{eqn:fig-b-h-initial}), we get
\begin{equation}\label{eqn:fig-b-first-replica}
  3h^2/2+ch+2b_2\delta^*-b_1b_2+b_2^2/2=0
\end{equation}
which is (\ref{eqn:fig-b-first}). We next find an expression for $\delta_2+\delta^*$. Rearranging (\ref{eqn:fig-b-delta2}), we get
\begin{equation}\label{eqn:fig-b-delta2-clear}
  \delta_2=\frac{b_1b_2-(3h/2+c)(h+\delta^*)}{3/2(h+\delta^*)+c}=\frac{2b_2\delta^*+b_2^2/2-\delta^*(3h/2+c)}{3/2(h+\delta^*)+c}
\end{equation}
where we have used (\ref{eqn:fig-b-first-replica}). Thus
$$
  \delta_2+\delta^*=\frac{(b_2+3\delta^*)(b_2+\delta^*)/2}{3/2(h+\delta^*)+c}.
$$
Plugging this into (\ref{eqn:fig-b-last}), we eliminate $\delta_2$, and obtain
\begin{equation}\label{eqn:fig-b-second-replica}
  27(c+h+\delta^*)(b_2+\delta^*)^2-4(4b_2+3\delta^*)(3(h+\delta^*)/2+c)^2=0
\end{equation}
which is (\ref{eqn:fig-b-second}). It is thus clear that $(h,\delta^*)$ can be obtained by simultaneously solving (\ref{eqn:fig-b-first}) and (\ref{eqn:fig-b-second}).

{\bf Step 3:} We now prove that a meaningful solution that satisfies (\ref{eqn:fig-b-first-replica}) and (\ref{eqn:fig-b-second-replica}) exists, by evaluating the bounds of the variables $h$, $\delta^*$, and $\delta_2$ . In Step 3a, we prove the bounds on $(h,\delta^*)$ when $b_1\geq 3b_2/2$. In Step 3b, we prove the bounds on $(h,\delta^*)$ when $b_1\in[b_2, 3b_2/2]$. In Step 3c, we prove the bounds on $\delta_2$ for all $b_1$.

{\bf Step 3a:} Consider the case when $b_1\geq 3b_2/2$. We consider a pair of $(\delta^*,h)$ values that satisfy (\ref{eqn:fig-b-first-replica}) as the end points, and prove that the expression on the left-hand side of (\ref{eqn:fig-b-second-replica}) changes sign at those end points. Given that $h$ is a decreasing function of $\delta^*$ (see (\ref{eqn:fig-b-h})), this suffices to show the bounds of $(\delta^*,h)$.

We claim that when $c\in[0,\beta]$, there exists a $(\delta^*,h)\in[\frac{c^2+6b_1b_2-7b_2^2}{12b_2},\frac{b_1}{2}-\frac{b_2}{4}]\times[0,\frac{2b_2-c}{3}]$ that solves (\ref{eqn:fig-b-first-replica}) and (\ref{eqn:fig-b-second-replica}) simultaneously. Observe that $h$ is a decreasing function of $\delta^*$ (see (\ref{eqn:fig-b-h})), and that the pairs $(\delta^*,h)=(\frac{b_1}{2}-\frac{b_2}{4},0)$ and $(\delta^*,h)=(\frac{c^2+6b_1b_2-7b_2^2}{12b_2},\frac{2b_2-c}{3})$ satisfy (\ref{eqn:fig-b-first-replica}). The choice $h=\frac{2b_2-c}{3}$ will be motivated later. It suffices now to indicate that it is to satisfy condition 3(d) of Theorem \ref{thm:Myerson}. We now prove that the left-hand side of (\ref{eqn:fig-b-second-replica}) has opposite signs at these pairs of $(\delta^*,h)$. Substituting $(\delta^*,h)=(\frac{c^2+6b_1b_2-7b_2^2}{12b_2},\frac{2b_2-c}{3})$, we obtain
\begin{equation}\label{eqn:fig-b-lower-bound}
  -\frac{(c-b_2)(6b_1b_2^2+b_2^3+6b_1b_2c+9b_2^2c+b_2c^2+c^3)}{4b_2}\leq 0
\end{equation}
for every $c\geq b_2$. Substituting $(\delta^*,h)=(\frac{b_1}{2}-\frac{b_2}{4},0)$, we obtain
$$
  \frac{1}{16}(72b_1^2b_2+144b_1b_2^2-90b_2^3+(-36b_1^2+84b_1b_2+399b_2^2)c-(96b_1+208b_2)c^2)
$$
which is nonnegative for every $c\in[0,\beta]$. So by continuity of (\ref{eqn:fig-b-second-replica}), there exists a $(\delta^*,h)$ in the rectangle $[\frac{c^2+6b_1b_2-7b_2^2}{12b_2},\frac{b_1}{2}-\frac{b_2}{4}]\times[0,\frac{2b_2-c}{3}]$, and by the continuity of (\ref{eqn:fig-b-first-replica}), the pair $(\delta^*,h)$ also satisfies (\ref{eqn:fig-b-first-replica}). We have thus proved our claim.

{\bf Step 3b:} Consider the case when $b_1\in[b_2,3b_2/2]$. We claim that there exists a $(\delta^*,h)\in[\frac{c^2+6b_1b_2-7b_2^2}{12b_2},b_1-b_2]\times[\frac{-c+\sqrt{c^2+3b_2(3b_2-2b_1)}}{3},\frac{2b_2-c}{3}]$ simultaneously solving (\ref{eqn:fig-b-first-replica}) and (\ref{eqn:fig-b-second-replica}). As before, substitution of $(\delta^*,h)=(\frac{c^2+6b_1b_2-7b_2^2}{12b_2},\frac{2b_2-c}{3})$ yields (\ref{eqn:fig-b-lower-bound}). We now substitute the other pair of $(\delta^*,h)$ on the left-hand side of (\ref{eqn:fig-b-second-replica}), and obtain
\begin{multline}\label{eqn:menu-2-upper}
  9b_1^2\left(3b_1-3b_2+2c+\sqrt{9b_2^2-6b_1b_2+c^2}\right)\\-(3b_1+b_2)\left(3b_1-3b_2+c+\sqrt{9b_2^2-6b_1b_2+c^2}\right)^2.
\end{multline}

We now show that this expression is nonnegative for every $b_1\in[b_2,3b_2/2]$, $c\in[b_2,\alpha_1]$. We do so by the following steps: (a) We first differentiate the expression with respect to $c$ and show that the differential is nonpositive; (b) We then evaluate the expression at $c=2(t-1)(b_1-b_2)+b_2$ (recall from Remark 2 that $t=3(37+3\sqrt{465})/176$) and show that it is nonnegative; and (c) We finally show that $\alpha_1\leq 2(t-1)(b_1-b_2)+b_2$.

We now differentiate the expression w.r.t. $c$. Fix $v=\sqrt{9b_2^2-6b_1b_2+c^2}$. When $b_1\in[b_2,3b_2/2]$ and $c\geq b_2$, we have
\begin{align*}
\mbox{(i) }&v=\sqrt{9b_2^2-6b_1b_2+c^2}\geq\sqrt{9b_2^2-6(3b_2/2)b_2+c^2}=c,\\
\mbox{(ii) }&v=\sqrt{9b_2^2-6b_1b_2+c^2}\leq\sqrt{9b_2^2-6(b_2)b_2+c^2}=\sqrt{3b_2^2+c^2}\leq 2c.
\end{align*}
So we have $c\leq v\leq 2c$. Differentiating (\ref{eqn:menu-2-upper}) with respect to $c$, we have
\begin{align*}
  &18b_1^2+\frac{9b_1^2c}{v}-2(3b_1+b_2)(-3b_2+3b_1+c+v)(1+c/v)\\&=\frac{18b_1^2v+9b_1^2c-2(3b_1+b_2)(c+v)^2-(18b_1^2+2(-6b_1b_2-3b_2^2)(c+v))}{v}\\&=\frac{-9b_1^2c+2(c+v)(3b_2(2b_1+b_2)-(3b_1+b_2)(c+v))}{v}\\&=\frac{-9b_1^2c+2(c+v)((2b_1+b_2)(2b_2-c-v)+b_2(2b_1+b_2)-b_1(c+v))}{v}\\&\leq\frac{-9b_1^2c+2(c+v)b_2^2}{v}\leq\frac{-9b_1^2c+6cb_2^2}{v}\leq 0
\end{align*}
where the first inequality follows from $c+v\geq 2c\geq 2b_2$, the second inequality from $c+v\leq 3c$, and the third inequality from $b_2\leq b_1$.

We now proceed to evaluate the expression at $c=2(t-1)(b_1-b_2)+b_2$. Substituting $c=2(t-1)(b_1-b_2)+b_2$ in (\ref{eqn:menu-2-upper}), we now verify if
\begin{multline*}
  \frac{15(117\sqrt{465}-4189)b_1^3+13(13417-225\sqrt{465})b_1^2b_2}{1936}\\-\frac{(70269-981\sqrt{465})b_1b_2^2+9(5021-21\sqrt{465})b_2^3}{1936}\\+\left(\frac{-(201+27\sqrt{465})b_1^2+(134+18\sqrt{465})b_1b_2+(111+9\sqrt{465})b_2^2}{44}\right)\\\sqrt{9b_2^2-6b_1b_2+\left(2\left(\frac{3(37+3\sqrt{465})}{176}-1\right)(b_1-b_2)+b_2\right)^2}\geq 0
\end{multline*}
Writing the above expression as $X+Y\sqrt{Z}$, we note that (i) $X\leq 0$ when $b_1\in b_2[1,1.03873]$, and $X\geq 0$ when $b_1\in b_2[1.03873,1.5]$; (ii) $Y\geq 0$ when $b_1\in b_2[1,1.04088]$, and $Y\leq 0$ when $b_1\in b_2[1.04088,1.5]$. So we now verify if $X^2-Y^2Z\leq 0$ when $b_1\in b_2[1,1.03873]$, and if $X^2-Y^2Z\geq 0$ when $b_1\in b_2[1.04088,1.5]$. That $X+Y\sqrt{Z}\geq 0$ when $b_1\in b_2[1.03873,1.04088]$ is clear since both $X$ and $Y$ are positive in that interval. Evaluating $X^2-Y^2Z$, we have
\begin{multline*}
  \frac{9}{42592}(b_1-b_2)(3b_2-2b_1)((20196\sqrt{465}-447876)b_2^4\\+(108900\sqrt{465}-2234628)b_1b_2^3+(32337\sqrt{465}-952857)b_1^2b_2^2\\+(4841141-276237\sqrt{465})b_1^3b_2+(140940\sqrt{465}-1820460)b_1^4)
\end{multline*}
which is negative when $b_1\in b_2[1,1.03977]$ and positive when $b_1\in b_2[1.03977,1.5]$. We have thus shown that the expression in (\ref{eqn:menu-2-bound}) is nonnegative when $b_1\in[b_2,3b_2/2]$, $b_2\leq c\leq 2(t-1)(b_1-b_2)+b_2$. That $\alpha_1\leq 2(t-1)(b_1-b_2)+b_2$ is shown via Mathematica (see \ref{app:c.1}(4)).

{\bf Step 3c:} For both the cases, we now claim that $\delta_2\in[\frac{2b_2-c}{3},\frac{b_2}{3}]$. To prove the claim, we do the following.
\begin{itemize}
\item We show the upper bound $\delta_2\leq\frac{b_2}{3}$ via Mathematica (see \ref{app:c.1}(2)).
\item We next show the lower bound. Since $\delta_2=(b_1b_2-(3h/2+c)(h+\delta^*))/(3(h+\delta^*)/2+c)$ decreases with $(h+\delta^*)$, we first find the upper bound on $(h+\delta^*)$.
\item We then substitute this obtained upper bound on $(h+\delta^*)$ and simplify, resulting in the lower bound $\delta_2\geq\frac{2b_2-c}{3}$.
\end{itemize}

We now fill in the details. To find the upper bound on $\delta_1=h+\delta^*$, we first show that $\delta_1$, as a function of $\delta^*$, decreases with increase in $\delta^*$. Differentiating the expression for $\delta_1=(h+\delta^*)$ with $h$ as in (\ref{eqn:fig-b-h}), we get $1-\frac{2b_2}{\sqrt{c^2+3b_2(2b_1-b_2-4\delta^*)}}$ which is nonpositive for $\delta^*\geq(c^2+6b_1b_2-7b_2^2)/(12b_2)$. But this is exactly the lower bound that we computed for $\delta^*$. The highest value of $\delta_1$ thus occurs at $(h,\delta^*)=(\frac{2b_2-c}{3},\frac{c^2+6b_1b_2-7b_2^2}{12b_2})$. Using these expressions, we get $\delta_1=(h+\delta^*)\leq\frac{c^2+6b_1b_2+b_2^2-4b_2c}{12b_2}$.

We now substitute the end points of $h+\delta^*$ in (\ref{eqn:fig-b-delta2-clear}), to evaluate the lower bound of $\delta_2$.
\begin{align*}
  \delta_2&=\frac{b_1b_2-(3h/2+c)(h+\delta^*)}{3(h+\delta^*)/2+c}\\&\geq\frac{b_1b_2-(b_2+c/2)(c^2+6b_1b_2+b_2^2-4b_2c)/(12b_2)}{(c^2+6b_1b_2+b_2^2+4b_2c)/(8b_2)}\\&=\frac{2b_2-c}{3}+\frac{4b_2(c^2-b_2^2)}{3(c^2+6b_1b_2+b_2^2+4b_2c)}\\&\geq\frac{2b_2-c}{3}
\end{align*}
where the first inequality occurs from the upper bound $h\leq (2b_2-c)/3$ and the above upper bound on $(h+\delta^*)$, and the second inequality from $c\geq b_2$. We have thus shown the lower bound. We have also shown that the probability of allocation $a_2=\frac{h+\delta^*}{\delta_2+\delta^*}\leq 1$, since $\delta_2\geq\frac{2b_2-c}{3}\geq h$.

{\bf Step 4:} We now proceed to prove parts (c) and (d) of Theorem \ref{thm:Myerson} (2)--(4). The expression for $V'(\delta)$ is the same as in the proof of Theorem \ref{thm:menu-1}, except in $[-\delta_2,\delta^*]$, where it is given by
\begin{equation}\label{eqn:menu-2-V'}
  V'(\delta)=-(c-2b_2+3\delta_2)+3\frac{\delta_2-h}{\delta_2+\delta^*}(\delta+\delta_2),\forall\delta\in(-\delta_2,\delta^*].
\end{equation}
From (\ref{eqn:menu-1-V'}), observe that $V'(\delta)$ is negative when $\delta\in[-b_2,-\frac{2b_2}{3}]$ and positive when $\delta\in[-\frac{2b_2}{3},-\frac{b_2}{3}]$. We also have from (\ref{eqn:V-fig-a}) that $V(-b_2)=V(-\frac{b_2}{3})=0$. So $V(\delta)=V(-b_2)+\int_{-b_2}^{\delta}V'(\tilde{\delta})\,d\tilde{\delta}\leq 0$ for all $\delta\in[-b_2,-\frac{b_2}{3}]$. It follows that $\int_{-b_2}^{-\frac{b_2}{3}}V(\delta)\,d\delta\leq 0$, and that $\int_{-b_2}^xV(\delta)\,d\delta\geq\int_{-b_2}^{-\frac{b_2}{3}}V(\delta)\,d\delta$ for all $x\in[-b_2,-\frac{b_2}{3}]$. Thus condition (2) of Theorem \ref{thm:Myerson} is verified.

We now prove that $\int_{-\frac{b_2}{3}}^{x}V(\delta)\,d\delta\geq 0$ for every $x\in[\frac{b_2}{3},\delta^*]$. Observe that $V'(\delta)$ is positive when $\delta\in[-\frac{b_2}{3},-\delta_2]$, negative when $\delta\in[-\delta_2,l_2]$ for some $l_2\in[-\delta_2,\delta^*]$, and positive when $\delta\in[l_2,\delta^*]$. These statements follow from (i) $\delta_2\geq\frac{2b_2-c}{3}$, (ii) $V'(\delta)$ increasing in the interval $[-\delta_2,\delta^*]$, and (iii) $h\leq\frac{2b_2-c}{3}$, all of which can be obtained from (\ref{eqn:menu-2-V'}). We also have $V(-\frac{b_2}{3})=V(\delta^*)=\int_{-\frac{b_2}{3}}^{\delta^*}V(\delta)\,d\delta=0$, which we used to derive the parameters $h$, $\delta_2$, and $\delta^*$. It follows that $\int_{-\frac{b_2}{3}}^xV(\delta)\,d\delta\geq 0$ for all $x\in[-\frac{b_2}{3},\delta^*]$. Thus condition (3) of Theorem \ref{thm:Myerson} is verified.

The proof that the conditions of Theorem \ref{thm:Myerson} (4) are satisfied trace the same steps as in the proof of Theorem \ref{thm:menu-1}, provided $\delta_1\leq\frac{2b_1-c}{3}$. If $\delta_1>\frac{2b_1-c}{3}$, then $V'(\delta)$ is no more positive in the interval $[b_1-b_2,\delta_1]$. We consider two cases.

Let $b_1\geq 3b_2/2$. Then we claim that $V(\delta)\geq 0$ holds for all $\delta\in[\delta^*,b_1]$, even when $V'(\delta)\leq 0$ for $\delta\in[b_1-b_2,\delta_1]$. Observe that (i) $V(\delta)=\frac{1}{2}(3\delta-b_1)(b_1-\delta)\geq 0$ for all $\delta\in[\max(b_1-b_2,\frac{b_1}{3}),b_1]$, and (ii) $\delta_1\geq b_1-b_2\geq\frac{b_1}{3}$, when $b_1\geq 3b_2/2$. So, $V(\delta)\geq 0$ for all $\delta\in[b_1-b_2,b_1]$. Now $V(\delta)\geq 0$ also holds in the interval $\delta\in[\delta^*,b_1-b_2]$ since $V'(\delta)\geq 0$ in that interval (see the discussion following (\ref{eqn:menu-1-V'})), and since $V(\delta^*)=0$. We have proved our claim.

We now consider the case when $b_1\in[b_2,3b_2/2]$. $V(\delta)$ could possibly be negative at some values of $\delta$. We now evaluate $\int_{\delta^*}^{\frac{b_1}{3}}V(\delta)\,d\delta$:
\begin{align*}
  &\int_{\delta^*}^{\frac{b_1}{3}}V(\delta)\,d\delta\\&=\int_{\delta^*}^{b_1-b_2}V(\delta)\,d\delta+\int_{b_1-b_2}^{\delta_1}V(\delta)\,d\delta+\int_{\delta_1}^{\frac{b_1}{3}}V(\delta)\,d\delta\\&=-\frac{2}{27}b_1^3-b_2(\delta^*)^2+b_2\delta^*(b_1-b_2/2)+b_1b_2h-\frac{b_2^2h}{2}-2b_2h\delta^*-\frac{ch^2}{2}-h^3\\&=-\frac{2}{27}b_1^3+\frac{(c+h)}{2}h^2-b_2(\delta^*)^2+b_2\delta^*(b_1-b_2/2)+hV(\delta^*)
\end{align*}
where $V(\delta^*)$ is obtained from (\ref{eqn:fig-b-h-initial}). The last expression is the same as (\ref{eqn:fig-b-third}), since $V(\delta^*)=0$. From Mathematica, (\ref{eqn:fig-b-third}) is nonnegative for all $c\in[b_2,\alpha_1]$ (see
\ref{app:c.1}(3)). Since $\int_{\frac{b_1}{3}}^{b_1}V(\delta)\,d\delta=\frac{2}{27}b_1^3\geq 0$, we have $\int_{\delta^*}^{b_1}V(\delta)\,d\delta\geq 0$. This verifies condition 4(c) of Theorem \ref{thm:Myerson}.

Observe that $V'(\delta)\leq 0$ only when $\delta\in[b_1-b_2,\delta_1]$. Also, $V(\delta^*)=0=V(\frac{b_1}{3})$. So $V(\delta)$ can be negative only when $\delta$ is in some subset of $[b_1-b_2,\frac{b_1}{3}]$, say in the interval $[l_1,\frac{b_1}{3}]$. Observe that the integral $\int_x^{b_1}V(\delta)\,d\delta$ thus attains its maximum either at $\delta^*$ or at $\frac{b_1}{3}$. But we just evaluated $\int_{\delta^*}^{\frac{b_1}{3}}V(\delta)\,d\delta\geq 0$, and so the maximum cannot be at $x=\frac{b_1}{3}$. Thus we have $\int_x^{b_1}V(\delta)\,d\delta\leq\int_{\delta^*}^{b_1}V(\delta)\,d\delta$ for all $x\in[\delta^*,b_1]$. Hence the result.\qed
\end{proof}

Observe that at $c=\alpha_1$, we have $\int_{\delta^*}^{\frac{b_1}{3}}V(\delta)\,d\delta=0$. When $c>\alpha_1$, the quantity turns negative, causing the condition in Theorem \ref{thm:Myerson}(4d) to fail. A transition occurs from the structure depicted in Figure \ref{fig:b-ini} to that depicted in Figure \ref{fig:c-ini}. We now proceed to prove the optimality of the structure in Figure \ref{fig:c-ini}, i.e., part 1(c) of Theorem \ref{thm:consolidate}.
\begin{theorem}\label{thm:menu-3}
Consider the case when $b_1\in[b_2,3b_2/2]$, and $c\in[\alpha_1,\alpha_2]$, where $\alpha_1$ and $\alpha_2$ are as defined as in Theorem \ref{thm:consolidate}. Then the optimal mechanism is as depicted in Figure \ref{fig:c-ini} (see also Figure \ref{fig:c-new}). The values of $h$ and $\delta^*$ are found by solving (\ref{eqn:fig-b-second}) and (\ref{eqn:fig-c-second}) simultaneously, and the values of $\delta_1$ and $\delta_2$ are given by
$$
(\delta_1,\delta_2)=\left(\delta^*+\frac{b_1b_2-2b_2\delta^*-b_2^2/2}{3h/2+c},\frac{b_1b_2-(3h/2+c)\delta_1}{3/2(h+\delta^*)+c}\right).
$$
The values of $a_1$ and $a_2$ are given by $(a_1,a_2)=\left(\frac{h}{\delta_1-\delta^*},\frac{h+\delta^*}{\delta_2+\delta^*}\right)$.
\end{theorem}
\begin{proof}
See \ref{app:b}. This too relies on Mathematica for verification of certain inequalities.\qed
\end{proof}

Consider $b_1\in[b_2,3b_2/2]$. The proof (in \ref{app:b}) indicates that at $c=\alpha_2$, we have $a_1+a_2=1$, and that when $c>\alpha_2$, we have $a_1+a_2<1$. This causes the monotonicity of $q_1$ to fail (recall that $q_1$ increasing is one of the constraints of Problem (\ref{eqn:optim-myerson})). Further, when $a_1+a_2=1$, the slope of the line segment joining $(c,c+\delta_2)$, $(c+h+\delta^*,c+h)$, and the slope of the line segment joining $(c+h+\delta^*,c+h)$, $(c+\delta_1,c)$, are equal, i.e., $-\frac{1-a_2}{a_2}=-\frac{a_1}{1-a_1}$. The two line segments thus turn into a single line segment that joins $(c,c+\delta_2)$, $(c+\delta_1,c)$. A transition thus occurs from the structure depicted in Figure \ref{fig:c-ini} to that in Figure \ref{fig:d-ini}, with $a_2=1-a_1=a$.

Consider $b_1\geq 3b_2/2$. At $c=\beta$, we have $h=0$. Thus a transition occurs from the structure depicted in Figure \ref{fig:b-ini} to that in Figure \ref{fig:d'-ini}.

We now proceed to prove the optimality of the structures depicted in Figures \ref{fig:d-ini}--\ref{fig:e'-ini}, i.e., parts 1(d)--(e) and 2(c)--2(d) of Theorem \ref{thm:consolidate}.

\begin{theorem}\label{thm:menu-4,5}
\begin{enumerate}
\item[(i)] Consider the case when $c\in[\beta,\frac{216b_1^2b_2}{108b_1^2-108b_1b_2-5b_2^2}]$, and $b_1\geq 3b_2/2$, where $\beta$ is as defined Theorem \ref{thm:consolidate}. Then the optimal mechanism is as depicted in Figure \ref{fig:d'-ini} (see also Figure \ref{fig:d'-new}). The values of $\delta_1$ and $\delta_2$ are computed by solving the following equations simultaneously.
\begin{align*}
&-\frac{3}{2}\delta_1\delta_2-c(\delta_1+\delta_2)+b_1b_2=0.\\&-\frac{2}{27}b_2^3+\frac{1}{2}\delta_1\delta_2(\delta_2-\delta_1)+\frac{c}{2}(\delta_2^2-\delta_1^2)+\frac{1}{16}b_2(2b_1-b_2)^2=0.
\end{align*}
The value of $a$ is given by $a=\frac{\delta_1}{\delta_1+\delta_2}$. If $c\geq\frac{216b_1^2b_2}{108b_1^2-108b_1b_2-5b_2^2}$, then the optimal mechanism is as depicted in Figure \ref{fig:e'-ini} (see also Figure \ref{fig:e'-new}).
\item[(ii)] Consider the case when $b_1\in[b_2,3b_2/2]$, and $c\in[\alpha_2,\frac{27b_1^2b_2^2}{4(b_1^3-b_2^3)}]$, where $c=\alpha_2$ is as defined in Theorem \ref{thm:consolidate}. Then the optimal mechanism is as depicted in Figure \ref{fig:d-ini} (see also Figure \ref{fig:d-new}). The values of $\delta_1$ and $\delta_2$ are computed by solving the following equations simultaneously.
\begin{align*}
&-\frac{3}{2}\delta_1\delta_2-c(\delta_1+\delta_2)+b_1b_2=0.\\&\frac{2}{27}(b_1^3-b_2^3)+\frac{1}{2}\delta_1\delta_2(\delta_2-\delta_1)+\frac{c}{2}(\delta_2^2-\delta_1^2)=0.
\end{align*}
The value of $a$ is given by $a=\frac{\delta_1}{\delta_1+\delta_2}$. If $c\geq\frac{27b_1^2b_2^2}{4(b_1^3-b_2^3)}$, then the optimal mechanism is as depicted in Figure \ref{fig:e-ini} (see also Figure \ref{fig:e-new}).
\end{enumerate}
\end{theorem}

\begin{proof}
See \ref{app:b}.\qed
\end{proof}

\section{On Extending to Uniform Distributions on General Rectangles}
We have computed the optimal mechanism in the two-item unit-demand setting when $z\sim\mbox{Unif}[c,c+b_1]\times[c,c+b_2]$ for every nonnegative $(c,b_1,b_2)$. Our computation used the method based on the virtual valuation function designed in \cite{Pav11}. We can now ask if there is a generalization of this method for more general distributions, specifically for uniform distributions on rectangles $[c_1,c_1+b_1]\times[c_2,c_2+b_2]$, when $c_1\ne c_2$. We conjecture that the optimal mechanisms would have structures similar to the five structures as in the case of $c_1=c_2$. We now report some promising preliminary results that support this conjecture.
\begin{theorem}\label{thm:extension}
Consider the case when $b_1\geq b_2$. Let $c_2\geq 0$, $c_1\geq c_2$, and $2c_1-c_2\leq b_2$. Then, the optimal mechanism is as depicted in Figure \ref{fig:a-ini} (see also Figure \ref{fig:a-new}). The values of $\delta_1$ and $\delta_2$ are computed by solving the following equations simultaneously.
\begin{align*}
  &-3\delta_1\delta_2-c_2\delta_1-c_1\delta_2+b_1b_2=0.\\
  &-\frac{3}{2}\delta_2^2+2b_2\delta_2-\frac{b_2^2}{2}-d(b_2-\delta_2)+(c_2-2b_2+3\delta_2)\delta_1=0.
\end{align*}
\end{theorem}
\begin{proof}
See \ref{app:b}. The proof traces the same steps as in the proof of Theorem \ref{thm:menu-1}.\qed
\end{proof}

\section{Conclusion and Future Work}

We solved the problem of computing the optimal mechanism for the two-item one-buyer unit-demand setting, when the buyer's valuation $z\sim\mbox{Unif}[c,c+b_1]\times[c,c+b_2]$ for arbitrary nonnegative values of $(c,b_1,b_2)$. Our results show that a wide range of structures arise out of different values of $c$. When the buyer guarantees that his valuations for the items are at least $c$, the seller offers different menus based on the guaranteed minimum $c$ and the upper bounds $c+b_i$, $i = 1,2$.

Taking a cue from the solution method in the unrestricted setting \cite{TRN16}, we initially attempted to solve the problem using the duality approach in \cite{DDT15}, but constructing a dual measure in the unit-demand setting turned out to be intricate. We then used the virtual valuation method used in \cite{Pav11} to compute the solution. We now characterize the pros and cons of these approaches.

The duality approach could not be pursued systematically because the construction of a shuffling measure that both convex-dominates $0$ and spans over more than one line segment appears to be difficult. Observe that in both Examples 2 and 3, there exists some constant allocation region that is a part of both the top boundary and the right boundary of $D$. So the shuffling measure had to be constructed so that it spans over two line segments connected at the top-right corner of $D$. To get around this issue, we had to construct (i) a shuffling measure on the line $z_1+z_2=2c+\delta_2$ in Example 2, and (ii) a shuffling measure that transfers mass horizontally in Example 3. The problem of constructing a ``generalized'' shuffling measure that both convex-dominates $0$ and also spans over two segments, thereby rendering the dual approach practical, is a possible direction for future work.

The virtual valuation method on the other hand, did not pose any issue when constant allocation regions span over the top-right corner. The approach provides a generalized procedure to verify if a menu at hand is optimal or not, under the (only) constraint that the distribution satisfies the negative power rate condition (stated in Theorem \ref{thm:pav-1}). So unlike the duality approach, we cannot use this approach to solve the problem for general distributions. But our results for $z\sim\mbox{Unif}[c,c+b_1]\times[c,c+b_2]$ and the extension to general rectangles suggest that this approach can be used to solve the problem of computing the optimal mechanism for all distributions satisfying the negative power rate condition. The key challenge in solving these problems is to find the exclusion region $Z$ for arbitrary distributions, so that we can use Theorem \ref{thm:Myerson} to verify if the menu is optimal or not. Coming up with a generalized procedure to compute $Z$ is a possible direction for future work.

Our proofs used Mathematica to verify certain algebraic inequalities that turn out to be complicated functions of $(c,b_1,b_2)$ involving fifth roots and eighth roots of some expressions. This leads us to the following questions. From a rather abstract perspective, Pavlov's sufficient conditions lead to the identification of a family of polynomial equalities and inequalities in the variables $(h,\delta^*,\delta_1,\delta_2)$ in Figures \ref{fig:a-new}--\ref{fig:e'-new}, indexed by the parameters $(c,b_1,b_2)$. In a nutshell, our work is a careful analysis of the solution space, denoted $L_{c,b_1,b_2}$, associated with the polynomial equalities and inequalities. We argued that $L_{c,b_1,b_2}$ is nonempty for every parameter $(c,b_1,b_2)$. We also captured the transitions of $L_{c,b_1,b_2}$ as the parameters vary. Can this view provide a more systematic procedure to solve the case of uniform distribution on any rectangle in the positive quadrant, or more generally, the case of any distribution of valuations on the positive quadrant? Alternatively, can the procedure of this paper (both existence of solutions and capture of transitions) be automated on Mathematica or other similar tool? These are some computation related problems that might be of interest to the computer scientists.

\section*{Acknowledgements}
This work was supported by the Defence Research and Development Organisation [Grant no. DRDO0667] under the DRDO-IISc Frontiers Research Programme. The first author thanks Prof. G. Pavlov for a very informative discussion.

\appendix
\section*{Appendix}
\section{Proofs from Section 2}\label{app:a}
{\bf Proof of Lemma \ref{lem:cvx}:} We first compute the quantities $\bar{\alpha}^{(1)}([1.26,1.26+2/3]\times\{2.26\})$ and $\int_{0}^{2/3}t\,d\bar{\alpha}^{(1)}(1.26+t,2.26)$:
\begin{align*}
  \bar{\alpha}^{(1)}([1.26,1.26+2/3]\times\{2.26\}) & =\int_{0}^{2/3}(3t-1)\,dt=(3/2)(2/3)^2-2/3=0,\\
  \int_{0}^{2/3}t\,d\bar{\alpha}^{(1)}(1.26+t,2.26) & =\int_{0}^{2/3}t(3t-1)\,dt=\frac{2^3}{3^3}-\frac{1}{2}\cdot\frac{2^2}{3^2}\geq 0.
\end{align*}
We compute the same quantities for $\beta^{(1)}$:
\begin{align*}
  &\bar{\beta}^{(1)}([1.26+2/3,2.26]\times\{2.26\})\\&=\int_{2/3}^{43/63}(3t-1)\,dt+\int_{43/63}^{1}(t(1.0155)+(1.9845)(43/63)-2.26)\,dt\\
  &=(3/2)((43/63)^2-(2/3)^2)-1/63+1.0155(1-(43/63)^2)/2\\&\hspace*{2in}+(20/63)(1.9845(43/63)-2.26)=0,
\end{align*}
and
\begin{align*}
  &\int_{2/3}^{1}t\,d\bar{\beta}^{(1)}(1.26+t,2.26)\\&=\int_{2/3}^{43/63}t(3t-1)\,dt+\int_{43/63}^{1}t(t(1.0155)+(1.9845)(43/63)-2.26)\,dt\\
  &=(43/63)^3-(2/3)^3-((43/63)^2-(2/3)^2)/2+(1.0155)(1-(43/63)^3)/3\\&\hspace*{1.5in}+(1-(43/63)^2)(1.9845(43/63)-2.26)/2=0.
\end{align*}

Now consider $h$ to be the affine shift of any increasing convex function $g$ (i.e., $h=\theta_1g+\theta_2, \theta_1>0, \theta_2\in\mathbb{R}$) such that $h(t)=t$ for $t=43/63$ and $t=\frac{2.26-1.9845*43/63}{1.0155}\approx 0.891679$. Observe that $\beta^{(1)}(1.26+t,2.26)\geq 0$ when $t\in[2/3,43/63]\cup[0.891679,1]$, and $\beta^{(1)}(1.26+t,2.26)<0$ when $t\in(43/63,0.891679)$. So we have $h(t)\leq t$ when $\beta^{(1)}<0$, and $h(t)\geq t$ when $\beta^{(1)}>0$. Now,
\begin{align*}
  &\int_{2/3}^{1}g(t)\,d\bar{\beta}^{(1)}(1.26+t,2.26)\\&=\frac{1}{\theta_1}\int_{2/3}^{1}h(t)\,d\bar{\beta}^{(1)}(1.26+t,2.26)\\
  &= \frac{1}{\theta_1}\left(\int_{2/3}^{1}(h(t)-t+t)\,d\bar{\beta}^{(1)}(1.26+t,2.26)\right)\\
  &=\frac{1}{\theta_1}\left(\int_{2/3}^{1}(h(t)-t)\,d\bar{\beta}^{(1)}(1.26+t,2.26)\right)\\
  &\geq 0
\end{align*}
where the first equality follows from $\bar{\beta}^{(1)}([1.26+2/3,2.26]\times\{2.26\})=0$, the third equality follows from $\int_{2/3}^{1}t\,d\bar{\beta}^{(1)}(1.26+t,2.26)=0$, and the last inequality follows because $\sgn{(h(t)-t)}=\sgn{\left(\beta^{(1)}(t)\right)}$ for every $t\in[2/3,1]$. The proof of $\bar{\beta}^{(1)}\succeq_{cvx}0$ is similar. Hence the result.\qed

{\bf Proof of Theorem \ref{thm:eg-2}:} We define $q$ as given in Figure \ref{fig:illust-2}, and construct $u$ such that $\nabla u=q$. We now construct the shuffling measure $\bar{\lambda}+\sum_i(\bar{\alpha}^{(i)}+\bar{\beta}^{(i)})$ as follows. We define $\alpha^{(i)}$ and $\beta^{(i)}$ same as in (\ref{eqn:alpha}) and (\ref{eqn:beta}) respectively, but with $\delta_1=\delta_2=\frac{((3+\sqrt{33})/8)-1}{(27-3\sqrt{33})/32}>\delta_2'$ and $a=(27-3\sqrt{33})/32$. We define $\lambda:D\rightarrow\mathbb{R}$, as in (\ref{eqn:shuffle-eg-2}).

We now construct $\gamma$ as follows. Let $\gamma_1=\gamma_1^Z+\gamma_1^{D\backslash Z}$, with $\gamma_1^Z=\bar{\mu}^Z$ and $\gamma_1^{D\backslash Z}=(\bar{\mu}^{D\backslash Z}+\sum_i(\bar{\alpha}^{(i)}+\bar{\beta}^{(i)}))_++\bar{\lambda}_+$. This is supported on $Z\cup([1.5,2.5]\times\{2.5\})\cup(\{2.5\}\times[1.5,2.5])\cup\{z:\lambda(z)\geq 0\}$. We define $\gamma_1^s$ as the Radon-Nikodym derivative of $\gamma_1$ w.r.t. the surface Lebesgue measure. It is easy to see that $\gamma_1^s(z)=\mu_s(z)+\sum_i(\alpha^{(i)}(z)+\beta^{(i)}(z))+\lambda(z)$ when $z\in(Z\cap D)([1.5,2.5]\times\{2.5\})\cup(\{2.5\}\times[1.5,2.5])\cup\{z:\lambda(z)\geq 0\}$, and zero otherwise. Now we specify $\gamma(\cdot~|~x)$ for every $x$ in the support of $\gamma_1$.
\begin{enumerate}
 \item[(a)] For $x\in Z$, we define $\gamma(y~|~x)=\delta_x(y)$. This is interpreted as no mass being transferred.
 \item[(b)] For $x\in([1.5,2.5]\times\{2.5\})\cup(\{2.5\}\times[1.5,2.5])$, we define $\gamma(y~|~x)=(\mu(y)+\mu_s(y)+\lambda(y))_-/\gamma_1^s(x)$ when $y\in\{y\in QRS\delta_2P_2P_1\delta_1Q:y_1-y_2=x_1-x_2\}$, and zero otherwise (see Figure \ref{fig:segmentation}). (By an abuse of notation, we denote the values of $\delta_1$ and $\delta_2$ as points marked in the Figure.) This is interpreted as transfer of $\gamma_1^s(x)$ from the boundary to the above line segment.
 \item[(c)] For $\{x:\lambda(x)>0\}$, we define $\gamma(y~|~x)=(\mu(y)+\mu_s(y))_-/\lambda(x)$ when $y\in\{y\in(\delta_1P_1\delta_1'\delta_1)\cup(\delta_2P_2\delta_2'\delta_2):y_1-y_2=x_1-x_2\}$, and zero otherwise (see Figure \ref{fig:segmentation}). This is interpreted as transfer of $\lambda(x)$ from the point $x$ on the line $x_1+x_2=2c+\delta_2$ to the above line segment.
\end{enumerate}
We then define $\gamma(F)=\int_{(x,y)\in F}\gamma_1(dx)\gamma(dy~|~x)$ for any measurable $F\in D\times D$. It is now easy to check that $\gamma_2^Z=\bar{\mu}^Z$, and $\gamma_2^{D\backslash Z}=(\bar{\mu}^{D\backslash Z}+\sum_i(\bar{\alpha}^{(i)}+\bar{\beta}^{(i)}))_-+\bar{\lambda}_-$. Thus we have $(\gamma_1-\gamma_2)^Z=0$, and $(\gamma_1-\gamma_2)^{D\backslash Z}=\bar{\mu}^{D\backslash Z}+\sum_i(\bar{\alpha}^{(i)}+\bar{\beta}^{(i)})+\bar{\lambda}$.

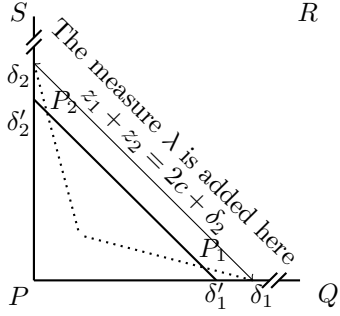
\begin{figure}[t]
\centering
\begin{tikzpicture}[scale=0.2,font=\small,axis/.style={very thick, -}]
\node at (-1,-1) {$P$};
\draw [axis,thick,-] (0,0)--(15.5,0);
\draw [axis,thick,-] (16,0.5)--(15,-0.5);
\draw [axis,thick,-] (16.6,0.5)--(15.6,-0.5);
\draw [axis,thick,-] (16.1,0)--(17.5,0);
\node [right] at (18,-1) {$Q$};
\draw [axis,thick,-] (0,0)--(0,15.5);
\draw [axis,thick,-] (0.5,16)--(-0.5,15);
\draw [axis,thick,-] (0.5,16.6)--(-0.5,15.6);
\draw [axis,thick,-] (0,16.1)--(0,17.5);
\node [above] at (18,16.5) {$R$};
\node [above] at (-1,16.5) {$S$};
\draw [axis,thick,-] (0,12)--(12,0);
\draw [thick,dotted] (0,14.4)--(3,3);
\draw [thick,dotted] (14.4,0)--(3,3);
\draw [thin,<->] (0,14.4)--(14.4,0);
\node [rotate=-45] at (8,8) {$z_1+z_2=2c+\delta_2$};
\node [rotate=-45] at (9.25,9.25) {The measure $\lambda$ is added here};
\node at (15,-1) {$\delta_1$};
\node at (12,-1) {$\delta_1'$};
\node at (11.8,2) {$P_1$};
\node [below] at (-1,15) {$\delta_2$};
\node [below] at (-1,12) {$\delta_2'$};
\node at (2,11.8) {$P_2$};
\end{tikzpicture}
\caption{The mechanism in Figure \ref{fig:illust-2} magnified near its left-bottom corner. The full support set $D$ is denoted by $PQRS$. The slope of the dotted lines equal $-(1-a)/a$ and $-a/(1-a)$, respectively, with $a=(27-3\sqrt{33})/32$. $P_1$, $P_2$ are points where the dotted lines intersect with the line denoted $\delta_1'\delta_2'$.}\label{fig:segmentation}
\end{figure}

The proof that $\gamma$ satisfies all the required conditions of Lemma \ref{lem:compslack} traces the same steps as in the proof of Theorem \ref{thm:eg-1}. The extra step here is to show that $\bar{\lambda}\succeq_{cvx}0$. We do this (i) by proving that both the measure $\bar{\lambda}$ and its mean vanish in its support set, and then (ii) by using the same arguments in Lemma \ref{lem:cvx}. We now compute
\begin{align*}
&\int_{1-\delta_2}^{1}\lambda(c+(t-1+\delta_2)/2,c+\delta_2-(t-1+\delta_2)/2)\,dt\\&=\int_{0}^{\delta_2-\sqrt{5/3}+1}(3at+c)\,dt\\&\hspace*{.25in}+\int_{1-(\sqrt{5/3}-1)}^{1}(3t(a-1/2)+3/2(1-(\sqrt{5/3}-1))-3a(1-\delta_2))\,dt\\&=3/2(\delta_2-\sqrt{5/3}+1)^2+c(\delta_2-\sqrt{5/3}+1)\\&\hspace*{1in}+3/2(a-1/2)(1-(2-\sqrt{5/3})^2)\\&\hspace*{1.5in}+(\sqrt{5/3}-1)(3/2(2-\sqrt{5/3})-3a(1-\delta_2))\\&=0
\end{align*}
where the last equality follows by putting in the values of $c$ and $\delta_2$. We also have
\begin{multline*}
  \int_{1-\delta_2}^{1}(t-1)\lambda(c+(t-1+\delta_2)/2,c+\delta_2-(t-1+\delta_2)/2)\,dt\\+\int_{1-\delta_1}^{1}(t-1)\lambda(c+\delta_1+(t-1+\delta_1)/2,c-(t-1+\delta_1)/2)\,dt=0
\end{multline*}
which follows because (i) $\lambda$ is symmetric about the line $t=1$, and (ii) $(t-1)$ is an odd function about the line $t=1$. The proof of $\bar{\lambda}\succeq_{cvx}0$ now traces the same steps as in Lemma \ref{lem:cvx}.\qed

{\bf Proof of Theorem \ref{thm:eg-3}:} We define $q$ as given in Figure \ref{fig:illust-3}, and construct $u$ such that $\nabla u=q$. Defining $\delta^*:=\delta_1-\delta_2$, we now construct the shuffling measure $\bar{\alpha}+\bar{\alpha}^{(o)}+\bar{\alpha}^{(h)}$, according to the terms defined in (\ref{eqn:shuffle-eg-3}), (\ref{eqn:shuffle-eg-3-o}), and (\ref{eqn:shuffle-eg-3-h}).

We now construct $\gamma$ as follows. Let  $\gamma_1=\gamma_1^Z+\gamma_1^{D\backslash Z}$, with $\gamma_1^Z=\bar{\mu}^Z$ and $\gamma_1^{D\backslash Z}=(\bar{\mu}^{D\backslash Z}+\bar{\alpha}+\bar{\alpha}^{(o)}+\bar{\alpha}^{(h)})_+$. This is supported on $Z\cup([0,1.2]\times\{1\})\cup(\{1.2\}\times[0,1])$. We define $\gamma_1^s$ as the Radon-Nikodym derivative of $\gamma_1$ w.r.t. the surface Lebesgue measure. It is easy to see that $\gamma_1^s(z)=\mu_s(z)+\alpha(z)+\alpha^{(o)}(z)+\alpha^{(h)}(z)$ when $z\in(Z\cap D)([0,1.2]\times\{1\})\cup(\{1.2\}\times[0,1])$, and zero otherwise. Now we specify $\gamma(\cdot~|~x)$ for every $x$ in the support of $\gamma_1$.
\begin{enumerate}
 \item[(a)] For $x\in Z$, we define $\gamma(y~|~x)=\delta_x(y)$. This is interpreted as no mass being transferred.
 \item[(b)] For $x\in([0,1+\delta^*]\times\{1\})\cup(\{1.2\}\times([0,\delta_1-0.2]\cup[2/3,1]))$, we define $\gamma(y~|~x)=(\mu(y)+\mu_s(y))_-/\gamma_1^s(x)$ when $y\in\{y\in D\backslash Z:y_1-y_2=x_1-x_2\}$, and zero otherwise. This is interpreted as transfer of $\gamma_1^s(x)$ from the boundary to the above line segment.
 \item[(c)] For $x\in([1+\delta^*,1.2]\times\{1\})$, we define $\gamma(y~|~x)=(\mu(y)+\mu_s(y))_-/\gamma_1^s(x)$ when $\{y_1-y_2=x_1-x_2,\,y_2\in[2/3,1]\}$, and zero otherwise. Again, this is interpreted as transfer of $\gamma_1^s(x)$ from the boundary to the above line segment.
 \item[(d)] For $x\in(\{1.2\}\times[\delta_1-0.2,2/3])$, we define $\gamma(y~|~x)=(\mu(y)+\mu_s(y))_-/\gamma_1^s(x)$, when $y\in\{y\in D\backslash Z:y_1-y_2=x_1-x_2\}\cup\{y_2=x_2,y_1-y_2\in[\delta^*,b_1-b_2]\}$, and zero otherwise. This is interpreted as a transfer of $\gamma_1^s(x)$ from the boundary to two line segments -- one is a $45^\circ$ line segment contained within $D\backslash Z$, and the other is a horizontal line contained within $\{y_1-y_2\in[\delta^*,b_1-b_2]\}$. The transfers occur respectively due to the shuffling measures $\alpha^{(o)}$, oblique transfer, and $\alpha^{(h)}$, horizontal transfer.
\end{enumerate}

We then define $\gamma(F)=\int_{(x,y)\in F}\gamma_1(dx)\gamma(dy~|~x)$ for any measurable $F\in D\times D$. It is now easy to check that $\gamma_2^Z=\bar{\mu}^Z$, and $\gamma_2^{D\backslash Z}=(\bar{\mu}^{D\backslash Z}+\bar{\alpha}+\bar{\alpha}^{(o)}+\bar{\alpha}^{(h)})_-$. Thus we have $(\gamma_1-\gamma_2)^Z=0$, and $(\gamma_1-\gamma_2)^{D\backslash Z}=(\bar{\mu}^{D\backslash Z}+\bar{\alpha}+\bar{\alpha}^{(o)}+\bar{\alpha}^{(h)})$.

The proof that $\gamma$ satisfies all the required conditions of Lemma \ref{lem:compslack} traces the same steps as in the proof of Theorem \ref{thm:eg-1}. The extra step here is to show that $\bar{\alpha}\succeq_{cvx}0$, $\bar{\alpha}^{(o)}+\bar{\alpha}^{(h)}\succeq_{cvx}0$. We show $\bar{\alpha}\succeq_{cvx}0$ by first proving that $\bar{\alpha}([0,1+\delta^*]\times\{1\})=0\leq\int_{0}^{1+\delta^*}t\,\bar{\alpha}(dt,1)$, and then by tracing the same steps as in the proof of Lemma \ref{lem:cvx}. The convex dominance for the other measure is also shown the same way. We now fill in the details. To obtain $\bar{\alpha}\succeq_{cvx}0$, we first verify that
\begin{align*}
  \bar{\alpha}([0,1+\delta^*]\times\{1\})&=\int_{0}^{1-\delta_2}(3t-1)\,dt+\int_{1-\delta_2}^{1+\delta^*}(2-3\delta_2)\,dt\\&=3/2(1-\delta_2)^2-(1-\delta_2)+\delta_1(2-3\delta_2)\\&=0,
\end{align*}
and then verify that
\begin{align*}
  &\int_{0}^{1+\delta^*}t\,\bar{\alpha}(dt,1)\\&=\int_{0}^{1-\delta_2}t(3t-1)\,dt+\int_{1-\delta_2}^{1+\delta^*}t(2-3\delta_2)\,dt\\&=(1-\delta_2)^3-(1-\delta_2)^2/2+(2-3\delta_2)((1+\delta^*)^2-(1-\delta_2)^2)/2\\&\approx0.103227\geq 0.
\end{align*}
We then complete the proof of $\bar{\alpha}\succeq_{cvx}0$ by tracing the same steps as in the proof of Lemma \ref{lem:cvx}. 

To prove that $\bar{\alpha}^{(o)}+\bar{\alpha}^{(h)}\succeq_{cvx}0$, we first verify that
\begin{align*}
  &(\bar{\alpha}^{(o)}+\bar{\alpha}^{(h)})(\{1.2\}\times[0,1])\\&=\int_{0}^{1.2-\delta_1}(3t-1.2)\,dt+\int_{1.2-\delta_1}^{1}(2.4-3\delta_1)\,dt\\&\hspace*{1.5in}+\int_{\delta_1-0.2}^{\delta_2}3(t-\delta_1+0.2)\,dt+\int_{\delta_2}^{2/3}3(0.2-\delta^*)\,dt\\&=3/2(1.2-\delta_1)^2-1.2(1.2-\delta_1)+(\delta_1-0.2)(2.4-3\delta_1)\\&\hspace*{1.5in}+3/2(0.2-\delta^*)^2+(2-3\delta_2)(0.2-\delta^*)\\&=0,
\end{align*}
and then verify that
\begin{align*}
  &\int_{0}^{1}t\,(\bar{\alpha}^{(o)}+\bar{\alpha}^{(h)})(1.2,dt)\\&=\int_{0}^{1.2-\delta_1}t(3t-1.2)\,dt+\int_{1.2-\delta_1}^{1}t(2.4-3\delta_1)\,dt\\&\hspace*{1.5in}+\int_{\delta_1-0.2}^{\delta_2}3t(t-\delta_1+0.2)\,dt+\int_{\delta_2}^{2/3}3t(0.2-\delta^*)\,dt\\&=(1.2-\delta_1)^3-0.6(1.2-\delta_1)^2+(1.2-3\delta_1/2)(1-(1.2-\delta_1)^2)\\&\hspace*{1in+(\delta_2^3-(\delta_1-0.2)^3)}-3/2(\delta_1-0.2)(\delta_2^2-(\delta_1-0.2)^2)\\&\hspace*{1.5in}+3/2(4/9-\delta_2^2)(0.2-\delta_1+\delta_2)\approx0.137171\geq 0.
\end{align*}
The proof of $\bar{\alpha}^{(o)}+\bar{\alpha}^{(h)}\succeq_{cvx}0$ is then completed by tracing the same steps of the proof of Lemma \ref{lem:cvx}.\qed

\section{Proofs from Section 3 and Section 4}\label{app:b}
{\bf Proof of Lemma \ref{lem:V-V'}:} Recall that the marginal profit function is defined as $\bar{V}(\delta)=\delta g(u_1(\delta),\delta)-\int_{\delta}^{b_1}g(u_1(\tilde{\delta}),\tilde{\delta})\,d\tilde{\delta}+\int_{\delta}^{b_1}(\tilde{\delta} q_1(\tilde{\delta})-u_1(\tilde{\delta}))\frac{\partial}{\partial u_1}g(u_1(\tilde{\delta}),\tilde{\delta})\,d\tilde{\delta}$. Consider the term $\int_{\delta}^{b_1}g(u_1(\tilde{\delta}),\tilde{\delta})\,d\tilde{\delta}$.
\begin{multline*}
  \int_{\delta}^{b_1}g(u_1(\tilde{\delta}),\tilde{\delta})\,d\tilde{\delta}=\int_{\delta}^{b_1}\int_{z:z\in D\backslash Z,z_1-z_2=\tilde{\delta}}f(z)\,dz\,d\tilde{\delta}\\=\int_{z:z\in D\backslash Z,z_1-z_2\geq\delta}f(z)\,dz.
\end{multline*}
We use integration by parts on $\int_X(z.\nabla h(z)-h(z))f(z)\,dz$, and we obtain $\int_X h(z)\nu(z)\,dz+\int_{\partial X}h(z)\nu_s(z)\,dz$. Here,
$$
   \nu(z):=-z\cdot\nabla f(z)-3f(z),\,z\in X;\,\,
 \nu_s(z):=(z\cdot n(z))f(z),\,z\in\partial X.
$$
We regard $\nu$ as the density of a measure that is absolutely continuous with two-dimensional Lebesgue measure, and $\nu_s$ as the density of a measure that is absolutely continuous with the surface Lebesgue measure. Defining the measure $\bar{\nu}(A):=\int_X\mathbf{1}_A(z)\nu(z)\,dz+\int_{\partial X}\mathbf{1}_A(z)\nu_s(z)\,dz$ for all measurable sets $A$ and substituting $h(z)=1\,\forall z\in X$, we get $\int_X(z.\nabla h(z)-h(z))f(z)\,dz=\int_X-f(z)\,dz=\bar{\nu}(X)$. So we have
$$
  \int_{\delta}^{b_1}g(u_1(\tilde{\delta}),\tilde{\delta})\,d\tilde{\delta}=\int_{z:z\in D\backslash Z,z_1-z_2\geq\delta}f(z)\,dz=-\bar{\nu}({z\in D\backslash Z: z_1-z_2\geq\delta}).
$$

We now compare the components of the measures $\bar{\mu}$ and $\bar{\nu}$ in some set $X\subseteq D$. The functions $\mu$ and $\nu$ are clearly equal. The function $\nu_s(z)$ is nonzero in every $z\in\partial X$, whereas $\mu_s(z)$ is nonzero only when $z\in(X\cap\partial D)$. In other words, $\nu_s(z)$ is also nonzero for every $z\in(\partial X\backslash\partial D)$, when compared with $\mu_s(z)$. We now show that the terms $\delta g(u_1(\delta),\delta)$ and $\int_{\delta}^{b_1}(\tilde{\delta} q_1(\tilde{\delta})-u_1(\tilde{\delta}))\frac{\partial}{\partial u_1}g(u_1(\tilde{\delta}),\tilde{\delta})\,d\tilde{\delta}$ cancel those ``extra'' nonzero values. In other words, we show that
\begin{multline}\label{eqn:app-1}
  \delta g(u_1(\delta),\delta)+\int_{\delta}^{b_1}(\tilde{\delta} q_1(\tilde{\delta})-u_1(\tilde{\delta}))\frac{\partial}{\partial u_1}g(u_1(\tilde{\delta}),\tilde{\delta})\,d\tilde{\delta}\\=(\bar{\mu}-\bar{\nu})({z:z\in D\backslash Z, z_1-z_2\geq\delta}),
\end{multline}
and this completes our proof. We show this for the mechanism depicted below in Figure \ref{fig:c-appendix}, with the exclusion region $Z$ being a convex, decreasing set. Observe that all the mechanisms depicted in Figures \ref{fig:a-ini}--\ref{fig:e'-ini} have this property.

\begin{figure}[t]
\centering
\begin{tikzpicture}[scale=0.3,font=\small,axis/.style={very thick, -}]
\draw [axis,thick,-] (0,0)--(15,0);
\draw [axis,thick,-] (0,0)--(0,12);
\draw [axis,thick,-] (0,12)--(15,12);
\draw [axis,thick,-] (15,0)--(15,12);
\draw [thin,-] (11.5,0)--(15,3.5);
\node [rotate=45] at (12.5,1.5) {\tiny case 1};
\node at (11.5,-1) {$\delta^{(1)}$};
\draw [axis,thick,-] (0,6) to[in=100,out=-10] (9,0);
\node at (-1.3,6) {$-\delta_2$};
\node at (-1.3,12) {$-b_2$};
\draw [thin,-] (0,8)--(4,12);
\node [rotate=45] at (1.25,10.25) {\tiny case 3};
\node at (-1.3,8) {$\delta^{(3)}$};
\draw [thin,-] (8.4,1.9)--(15,8.5);
\node [rotate=45] at (10.75,4.75) {\tiny case 2};
\draw [thin,dotted] (8.4,1.9)--(6.5,0);
\node at (6,-1) {$\delta^{(2)}$};
\node at (7.25,1.5) {$x$};
\draw [thick,dotted] (15,12)--(3,0);
\node at (1.5,-1) {$b_1-b_2$};
\node at (8.5,-1) {$\delta_1$};
\node at (15,-1) {$b_1$};
\node at (3.5,2) {\Large$Z$};
\end{tikzpicture}
\caption{The structure of a typical mechanism. The variables marked in the boundary denote the values of $\delta$.}\label{fig:c-appendix}
\end{figure}
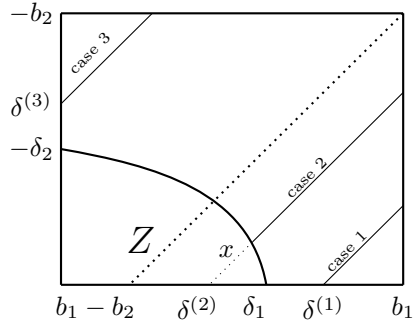

Define $z_2^*:[-\delta_2,\delta_1]\rightarrow[c,c+\delta_2]$ as
\begin{multline*}
  z_2^*(\delta)=\{z_2\in[c,c+b_2]:u(z_2+\delta,z_2)=0,\\u(z_2+\delta+h,z_2+h)>0\mbox{ for all }h>0\mbox{ sufficiently small}\}.
\end{multline*}
Observe that $z_2^*(\delta)$ is the value of $z_2$ in the curve $(-\delta_2\delta_1)$ that separates $Z$ and $D\backslash Z$, when $z_1-z_2=\delta$. So $z_2^*$ is decreasing with $\delta$, with $z_2^*(-\delta_2)=c+\delta_2$, and $z_2^*(\delta_1)=c$. Also, the curve $(-\delta_2\delta_1)$ can be represented by the points $\{(\delta+z_2^*(\delta),z_2^*(\delta)),\,\delta\in[-\delta_2,\delta_1]\}$. We now compute $u_1(\delta)$ for every $\delta\in[-b_2,b_1]$. We use the fact that $u(z)=0$ when $z\in Z$, and $u(z)=u_1(z_1-z_2)+z_2$.

$$
  u_1(\delta)=\begin{cases}-c-\delta_2-\int_{\delta}^{-\delta_2}q_1(\tilde{\delta})\,d\tilde{\delta}&\mbox{if }\delta\in[-b_2,-\delta_2]\\-z_2^*(\delta)&\mbox{if }\delta\in[-\delta_2,\delta_1]\\-c+\int_{\delta_1}^{\delta}q_1(\tilde{\delta})\,d\tilde{\delta}&\mbox{if }\delta\in[\delta_1,b_1].\end{cases}
$$

Using the values of $u_1(\delta)$, we now compute $g(u_1(\delta),\delta)$.
$$
  g(u_1(\delta),\delta)=\frac{1}{b_1b_2}\begin{cases}b_2+\delta&\mbox{if }\delta\in[-b_2,-\delta_2]\\b_2+c+u_1(\delta)&\mbox{if }\delta\in[-\delta_2,b_1-b_2]\\b_1-\delta+c+u_1(\delta)&\mbox{if }\delta\in[b_1-b_2,\delta_1]\\b_1-\delta&\mbox{if }\delta\in[\delta_1,b_1].\end{cases}
$$
For ease of notation, we drop the factor $1/(b_1b_2)$ in the rest of the proof.

To show (\ref{eqn:app-1}), we consider the following three cases: (i) $\delta\in[\delta_1,b_1]$, (ii) $\delta\in[-\delta_2,\delta_1]$, and (iii) $\delta\in[-b_2,-\delta_2]$ (see Figure \ref{fig:c-appendix}). In case (i), consider $\delta=\delta^{(1)}$. The measures $\bar{\mu}$ and $\bar{\nu}$ differ only in that $\bar{\nu}$ has an extra nonzero line measure on the line segment $\{z\in D\backslash Z:z_1-z_2=\delta^{(1)}\}$. We thus have
\begin{multline*}
  (\bar{\mu}-\bar{\nu})({z\in D\backslash Z:z_1-z_2\geq\delta^{(1)}})=-\int_{\substack{z\in D\backslash Z:\\z_1-z_2=\delta^{(1)}}}\nu_s(z)\,dz\\=\int_{\substack{z\in D\backslash Z:\\z_1-z_2=\delta^{(1)}}}(z_1-z_2)f(z)\,dz=\delta^{(1)}g(u_1(\delta^{(1)}),\delta^{(1)}).
\end{multline*}
Now observe that $\frac{\partial}{\partial u_1}g(u_1(\delta),\delta)=0$ when $\delta\in[\delta_1,b_1]$. So (\ref{eqn:app-1}) holds.

In case (ii), consider $\delta=\delta^{(2)}$. Then, $\bar{\nu}$ has an extra nonzero line measure on (i) the line segment $\{z\in D\backslash Z:z_1-z_2=\delta^{(2)}\}$, and (ii) the curve $\delta_1x$. Now we have
\begin{align*}
  &(\bar{\mu}-\bar{\nu})({z\in D\backslash Z:z_1-z_2\geq\delta^{(2)}})\\&=-\int_{\substack{z\in D\backslash Z:\\z_1-z_2=\delta^{(2)}}}\nu_s(z)\,dz-\int_{z\in\mbox{curve }\delta_1x}\nu_s(z)\,dz\\&=\int_{\substack{z\in D\backslash Z:\\z_1-z_2=\delta^{(2)}}}(z_1-z_2)f(z)\,dz-\int_{z\in\mbox{curve }\delta_1x}(z\cdot n(z))f(z)\,dz\\&=\delta^{(2)}g(u_1(\delta^{(2)}),\delta^{(2)})+\int_{\delta^{(2)}}^{\delta_1}((z_2^*(\delta)+\delta)q_1(\delta)+z_2^*(\delta)(1-q_1(\delta)))\,d\delta\\&=\delta^{(2)}g(u_1(\delta^{(2)}),\delta^{(2)})+\int_{\delta^{(2)}}^{\delta_1}(z_2^*(\delta)+\delta q_1(\delta))\,d\delta.
\end{align*}
where the third equality follows because (i) $\nabla u=q$, and (ii) $q_1(\delta)+q_2(\delta)=1$ for $z\in D\backslash Z$. Now observe that we have $\frac{\partial}{\partial u_1}g(u_1(\delta),\delta)=1$ when $\delta\in[-\delta_2,\delta_1]$. Therefore,
\begin{multline*}
  \delta^{(2)}g(u_1(\delta^{(2)}),\delta^{(2)})+\int_{\delta^{(2)}}^{b_1}(\delta q_1(\delta)-u_1(\delta))\frac{\partial}{\partial u_1}g(u_1(\delta),\delta)\,d\delta\\=\delta^{(2)}g(u_1(\delta^{(2)}),\delta^{(2)})+\int_{\delta^{(2)}}^{\delta_1}(z_2^*(\delta)+\delta q_1(\delta))\,d\delta.
\end{multline*}
Eq. (\ref{eqn:app-1}) thus holds for case 2.

In case (iii), consider $\delta=\delta^{(3)}$. Then, $\bar{\nu}$ has an extra nonzero line measure on (i) the line segment $\{z\in D\backslash Z:z_1-z_2=\delta^{(3)}\}$, and (ii) the curve $(-\delta_2\delta_1)$. Now by an analysis similar to case 2, it follows that
\begin{align*}
  &(\bar{\mu}-\bar{\nu})({z\in D\backslash Z:z_1-z_2\geq\delta^{(3)}})\\&=\delta^{(3)}g(u_1(\delta^{(3)}),\delta^{(3)})+\int_{-\delta_2}^{\delta_1}(z_2^*(\delta)+\delta q_1(\delta))\,d\delta\\&=\delta^{(3)}g(u_1(\delta^{(3)}),\delta^{(3)})+\int_{\delta^{(3)}}^{b_1}(\delta q_1(\delta)-u_1(\delta))\frac{\partial}{\partial u_1}g(u_1(\delta),\delta)\,d\delta.
\end{align*}
\qed

{\bf Proof of Theorem \ref{thm:menu-3}:}

{\bf Step 1:} We compute the virtual valuation function of the mechanism depicted in Figure \ref{fig:c-new}.
$$
  V(\delta)=\begin{cases}V(-\delta_2)-(c-2b_2+3\delta_2)(\delta+\delta_2)+\frac{3}{2}\frac{\delta_2-h}{\delta_2+\delta^*}(\delta+\delta_2)^2&\delta\in[-\delta_2,\delta^*]\\V(\delta^*)-(c-2b_2)(\delta-\delta^*)+\frac{3h}{2}\frac{(\delta_1-\delta)^2}{\delta_1-\delta^*}-\frac{3h}{2}(\delta_1-\delta^*)&\delta\in[\delta^*,b']\end{cases}
$$
where $b_1-b_2$ is indicated as $b'$. The expression for $V(\delta)$ when $\delta\in[-b_2,-\delta_2]\cup[\delta_1,b_1]$ remains the same as in (\ref{eqn:V-fig-a}), and the expression when $\delta\in[b',\delta_1]$ is given by
$$
  V(\delta)=V(b')-(c-2b_1)(\delta-b')-\frac{3}{2}(\delta^2-(b')^2-\frac{h}{\delta_1-\delta^*}((\delta_1-\delta)^2-(\delta_1-b')^2)).
$$

{\bf Step 2:} The mechanism has six unknowns: $h$, $\delta^*$, $\delta_1$, $\delta_2$, $a_1$, and $a_2$. Since $q=\nabla u$, a conservative field, we must have the slope of the line separating $(0,0)$ and $(1-a_2,a_2)$ allocation regions satisfying $-\frac{1-a_2}{a_2}=\frac{h-\delta_2}{h+\delta^*}$, which yields $a_2=\frac{h+\delta^*}{\delta_2+\delta^*}$. Similarly, the slope of the line separating $(0,0)$ and $(a_1,1-a_1)$ allocation regions must satisfy $-\frac{a_1}{1-a_1}=\frac{h}{\delta_1-\delta^*-h}$, which yields $a_1=\frac{h}{\delta_1-\delta^*}$.

We compute the other four unknowns by equating $\bar{\mu}(Z)=0$, $V(\delta^*)=0$, $\int_{-\frac{b_2}{3}}^{\delta^*}V(\delta)\,d\delta=0$, and $\int_{\delta^*}^{\frac{b_1}{3}}V(\delta)\,d\delta=0$. The latter three conditions follow from Theorem \ref{thm:Myerson} 3(b) and 3(c) because $q_1(\delta)=1-a_2$ for $\delta\in[-\frac{b_2}{3},\delta^*]$, and $q_1(\delta)=a_1$ for $\delta\in[\delta^*,\frac{b_1}{3}]$. We then have the following implications.

\begin{equation}
 \bar{\mu}(Z)=0\Rightarrow-(3h/2+c)(\delta_1+\delta_2)-3\delta_2\delta^*/2+b_1b_2=0.\label{eqn:fig-c-delta2}
\end{equation}
From (\ref{eqn:menu-1-V'}), we see that $V(\delta^*)$ is the negative of $\bar{\mu}$ measure of the nonconvex pentagon bound by $(c,c)$, $(c,c+b_2)$, $(c+b_2+\delta^*,c+b_2)$, $(c+h+\delta^*,c+h)$, and $(c+\delta_1,c)$. Thus
\begin{equation}
 V(\delta^*)=0\Rightarrow-(3h/2+c)(\delta_1-\delta^*)-2b_2\delta^*-b_2^2/2+b_1b_2=0.\label{eqn:fig-c-delta1}
\end{equation}
The expression for $\int_{-\frac{b_2}{3}}^{\delta^*}V(\delta)\,d\delta$ remains the same as in (\ref{eqn:fig-b-first-replica}).
\begin{equation}
 \int_{-\frac{b_2}{3}}^{\delta^*}V(\delta)\,d\delta=0\Rightarrow\frac{1}{54}(4b_2+3\delta^*)(b_2+3\delta^*)^2-\frac{c+h+\delta^*}{2}(\delta^*+\delta_2)^2=0.\label{eqn:fig-c-3}
\end{equation}
Next
\begin{align}
 &\int_{\delta^*}^{\frac{b_1}{3}}V(\delta)\,d\delta=0\nonumber\\&\Rightarrow\int_{\delta^*}^{b_1-b_2}V(\delta)\,d\delta+\int_{b_1-b_2}^{\delta_1}V(\delta)\,d\delta+\int_{\delta_1}^{\frac{b_1}{3}}V(\delta)\,d\delta=0\nonumber\\&\Rightarrow-\frac{2}{27}b_1^3+b_1b_2\delta_1+\frac{1}{2}(-b_2^2\delta_1+2b_2\delta^*(-2\delta_1+\delta^*)-(\delta_1-\delta^*)^2 (c+2h))=0\nonumber\\&\Rightarrow\frac{c+h}{2}(\delta_1-\delta^*)^2-b_2(\delta^*)^2+\frac{b_2\delta^*}{2}(2b_1-b_2)-\frac{2}{27}b_1^3+(\delta_1-\delta^*)V(\delta^*)=0.\label{eqn:fig-c-fourth}
\end{align}

The values of $h$, $\delta^*$, $\delta_1$ and $\delta_2$ can be obtained by solving these four equations simultaneously. We now proceed to prove that $(h,\delta^*)$ can be computed by solving (\ref{eqn:fig-b-second}) and (\ref{eqn:fig-c-second}) simultaneously.

We first find an expression for $\delta_2+\delta^*$. Rearranging (\ref{eqn:fig-c-delta2}), we have
\begin{equation}\label{eqn:fig-c-delta2-clear}
  \delta_2=\frac{b_1b_2-(3h/2+c)\delta_1}{3/2(h+\delta^*)+c}
\end{equation}
Similarly, rearranging (\ref{eqn:fig-c-delta1}), we have $\delta_1=\delta^*+\frac{b_1b_2-2b_2\delta^*-b_2^2/2}{3h/2+c}$. Substituting $\delta_1$ in (\ref{eqn:fig-c-delta2-clear}), we get
$$
  \delta_2+\delta^*=\frac{b_1b_2-(3h/2+c)(\delta_1-\delta^*)+\frac{3}{2}(\delta^*)^2}{3/2(h+\delta^*)+c}=\frac{(b_2+3\delta^*)(b_2+\delta^*)/2}{3/2(h+\delta^*)+c}.
$$
Plugging this into (\ref{eqn:fig-c-3}), we eliminate $\delta_2$. Similarly, plugging $(\delta_1-\delta^*)=\frac{b_1b_2-2b_2\delta^*-b_2^2/2}{3h/2+c}$ (obtained by rearranging (\ref{eqn:fig-c-delta1})) in (\ref{eqn:fig-c-fourth}), we eliminate $\delta_1$. We thus solve the following equations:
\begin{align}
 &27(c+h+\delta^*)(b_2+\delta^*)^2-4(4b_2+3\delta^*)(3/2(h+\delta^*)+c)^2=0\nonumber\\
 &(2b_1^3/27+b_2(\delta^*)^2-b_2\delta^*(b_1-b_2/2))(3h/2+c)^2\nonumber\\&\hspace*{1.5in}-\frac{(c+h)}{2}(2b_2\delta^*+b_2^2/2-b_1b_2)^2=0\nonumber
\end{align}
which are (\ref{eqn:fig-b-second}) and (\ref{eqn:fig-c-second}), respectively.

{\bf Step 3:} We now proceed to evaluate the bounds of the variables, in order to the prove the existence of a meaningful solution that solves (\ref{eqn:fig-b-second}) and (\ref{eqn:fig-c-second}). In Step 3a, we first prove that the condition $q_1$ increasing in Problem (\ref{eqn:optim-myerson}) is satisfied only when the left-hand side of (\ref{eqn:fig-c-third}) is nonnegative. In Steps 3b--3d, we prove the bounds on $(h,\delta^*)$, $\delta_1$ and $\delta_2$, respectively.

{\bf Step 3a:} We compute the values of $c$ where monotonicity of $q_1$ holds. Observe that monotonicity of $q_1$ holds when $1-a_2\leq a_1$, and that of $q_2$ holds when $1-a_1\leq a_2$. We thus verify if $a_1+a_2\geq 1$. On substituting the expressions for $a_1$ and $a_2$, we obtain
\begin{multline*}
  \frac{(h+\delta^*)(3/2(h+\delta^*)+c)}{(3/2(\delta^*)^2+2b_2\delta^*+b_2^2/2)}+\frac{h(3h/2+c)}{(b_1b_2-2b_2\delta^*-b_2^2/2)}\geq 1\\\Rightarrow (b_2^2+4b_2\delta^*-3\delta^*h)(b_2^2+4b_2\delta^*-2c\delta^*-3\delta^*h)\\-2b_1b_2(b_2^2+4b_2\delta^*-2c(\delta^*+h)-3h(2\delta^*+h))\geq 0.
\end{multline*}
The monotonicity condition thus amounts to verifying if the left-hand side of (\ref{eqn:fig-c-third}) is nonnegative. We verify via Mathematica that the expression is nonnegative for $c\in[b_2,\alpha_2]$, and that $\alpha_2\leq 2(t-1.4)(b_1-b_2)+1.4b_2$ (see \ref{app:c.2}(6--7)). We thus compute the bounds of $h$, $\delta^*$, $\delta_1$, and $\delta_2$ when $c\in[b_2,2(t-1.4)(b_1-b_2)+1.4b_2]$.

{\bf Step 3b:} We now evaluate the bounds on $\delta_1$ and $\delta_2$ in order to prove the existence of a meaningful solution that solves (\ref{eqn:fig-b-second}) and (\ref{eqn:fig-c-second}) simultaneously. Specifically, we now prove that there exists $(h,\delta^*)\in[0,\frac{2b_2-c}{3}]\times[0,b_1-b_2]$ that simultaneously solves (\ref{eqn:fig-b-second}) and (\ref{eqn:fig-c-second}). We show this using the same techniques as in Step 3 of proof of Theorem \ref{thm:menu-1}.

We first show that $\delta^*|_h$ is continuous in $h$, and decreases as $h$ increases. We rewrite (\ref{eqn:fig-c-second}) as follows.
\begin{multline*}
  ((3h/2+c)^2-2b_2(c+h))(b_2(\delta^*)^2-b_2(b_1-b_2/2)\delta^*)\\+\frac{2b_1^3}{27}(3h/2+c)^2-\frac{b_2^2}{2}(b_1-b_2/2)^2(c+h)=0
\end{multline*}
Solving this equation for $\delta^*$, we obtain
$$
  \delta^*|_h=\frac{b_1}{2}-\frac{b_2}{4}-\frac{3b_2-2b_1}{4}\frac{3h/2+c}{3}\sqrt{\frac{(8b_1-3b_2)/3}{2b_2^2(c+h)-b_2(3h/2+c)^2}}
$$
To prove the continuity of $\delta^*|_h$ in $h$, it suffices to show that the term $2b_2^2(c+h)-b_2(3h/2+c)^2$ is strictly positive for the desired values of $h$, since the expression is quadratic with negative coefficient on $h^2$. At $h=0$, the expression equals $b_2c(2b_2-c)>0$ for all $c\leq 2b_2$, and at $h=\frac{2b_2-c}{3}$, the expression equals $b_2/12(2b_2-c)(2b_2+3c)>0$. Thus $\delta^*|_h$ is continuous in $h$ when $h\in[0,\frac{2b_2-c}{3}]$.

To prove that $\delta^*|_h$ decreases in $h$, it suffices to prove that $\frac{c+h}{(3h/2+c)^2}$ decreases with $h$. We prove it by differentiating the term w.r.t. $h$, and proving that the numerator is nonpositive. The numerator of the derivative is $(3h/2+c)(-3h/2-2c)\leq 0$. We have thus shown that $\delta^*|_h$ decreases as $h$ increases.

We now show that $\delta^*|_{h=\frac{2b_2-c}{3}}\leq b_1-b_2$. From (\ref{eqn:fig-c-second}), we obtain
$$
  \delta^*|_{h=\frac{2b_2-c}{3}}=\frac{b_1}{2}-\frac{b_2}{4}-\frac{(3b_2-2b_1)}{4}\frac{(2b_2+c)}{3b_2}\sqrt{\frac{b_2(8b_1-3b_2)}{(2b_2-c)(2b_2+3c)}}.
$$
Observe that (a) $2b_2+c\geq 3b_2$ since $c\geq b_2$, and (b) $\frac{b_2(8b_1-3b_2)}{(2b_2-c)(2b_2+3c)}\geq 1$ since $\min_{b_1\in[b_2,3b_2/2]}(b_2(8b_1-3b_2))=5b_2^2$ and $\max_{c\geq b_2}((2b_2-c)(2b_2+3c))=5b_2^2$. We thus have
$$
\delta^*|_{h=\frac{2b_2-c}{3}}\leq\frac{b_1}{2}-\frac{b_2}{4}-\frac{3b_2-2b_1}{4}=b_1-b_2.
$$

We now show that $\delta^*|_{h=0}\geq 0$. From (\ref{eqn:fig-c-second}), we obtain
$$
  \delta^*|_{h=0}=\frac{b_1}{2}-\frac{b_2}{4}-\frac{(3b_2-2b_1)}{36}\sqrt{\frac{3c(8b_1-3b_2)}{b_2(2b_2-c)}}.
$$
When $b_1\in[b_2,3b_2/2]$, $\delta^*|_{h=0}$ decreases when $c$ increases from $b_2$ to $2b_2$. We thus obtain a lower bound on $\delta^*$ by substituting an upper bound on $c$. We use $c\leq 1.4b_2+2(1.75-1.4)(b_1-b_2)$ instead of $1.4b_2+2(t-1.4)(b_1-b_2)$ to simplify the calculation.
$$
  \delta^*|_{h=0}\geq\frac{2b_1-b_2}{4}-\frac{3b_2-2b_1}{36}\sqrt{\frac{2.1(b_1+b_2)(8b_1-3b_2)}{b_2(1.3b_2-0.7b_1)}}
$$
To prove that this expression is nonnegative, it suffices to prove that $(9(2b_1-b_2))^2(b_2(1.3b_2-0.7b_1))\geq(3b_2-2b_1)^2(2.1(b_1+b_2)(8b_1-3b_2))$. Simplifying this expression, we obtain
$$
  -67.2b_1^4-67.2b_1^3b_2+648b_1^2b_2^2-648b_1b_2^3+162b_2^4\geq 0
$$
which is true for $b_1\in[b_2,3b_2/2]$. We have thus shown that $\delta^*|_{h=0}\geq 0$.

We now proceed to prove that (a) substituting the entry points $(h|_{\delta^*=0},0)$ and $(\frac{2b_2-c}{3},\delta^*|_{h=\frac{2b_2-c}{3}})$ makes the expression nonpositive, and (b) substituting the exit points $(0,\delta^*|_{h=0})$ and $(h|_{\delta^*=b_1-b_2},b_1-b_2)$ on left-hand side of (\ref{eqn:fig-b-second}) makes the expression nonnegative.

We now consider the entry point $(h|_{\delta^*=0},0)$. Substituting $\delta^*=0$ on (\ref{eqn:fig-c-second}), we obtain $2b_1^3(3h/2+c)^2/27-(c+h)b_2^2(b_1-b_2/2)^2/2=0$. Let $h=h|_{\delta^*=0}$ solve this equation. Substituting $(h,\delta^*)=(h|_{\delta^*=0},0)$ in (\ref{eqn:fig-b-second}), we get $27(c+h|_{\delta^*=0})b_2^2-16b_2(3h|_{\delta^*=0}/2+c)^2$. We now prove that this expression is nonpositive.
\begin{align*}
  &27(c+h|_{\delta^*=0})b_2-16(3h|_{\delta^*=0}/2+c)^2\\
  &=27(c+h|_{\delta^*=0})b_2-\frac{108(c+h|_{\delta^*=0})b_2^2(b_1-b_2/2)^2}{b_1^3}\\&=\frac{27(c+h|_{\delta^*=0})b_2}{b_1^3}(b_1^3-4b_2(b_1-b_2/2)^2)\\&=\frac{27(c+h|_{\delta^*=0})b_2}{b_1^3}(b_1-b_2)(b_1^2-3b_1b_2+b_2^2)\leq 0
\end{align*}
when $b_1\in[b_2,3b_2/2]$. The first equality occurs since $h|_{\delta^*=0}$ solves $2b_1^3(3h/2+c)^2/27-(c+h)b_2^2(b_1-b_2/2)^2/2=0$.

We now consider the entry point $(\frac{2b_2-c}{3},\delta^*|_{h=\frac{2b_2-c}{3}})$. We first prove that $\delta^*|_{h=\frac{2b_2-c}{3}}\geq-\frac{2b_2}{3}$.
\begin{align*}
  &\delta^*|_{h=\frac{2b_2-c}{3}}\geq-\frac{2b_2}{3}\\&\Rightarrow\frac{b_1}{2}-\frac{b_2}{4}-\frac{(3b_2-2b_1)}{4}\frac{(2b_2+c)}{3b_2}\sqrt{\frac{b_2(8b_1-3b_2)}{(2b_2-c)(2b_2+3c)}}\geq-\frac{2b_2}{3}\\&\Rightarrow (6b_1+5b_2)^2b_2(2b_2-c)(2b_2+3c)\geq (3b_1-2b_2)^2(2b_2+c)^2(8b_1-3b_2)
\end{align*}
Consider $b_2\leq c\leq tb_2$, and $b_1\in[b_2,3b_2/2]$ (recall that $t=3(37+3\sqrt{465})/176$). The inequality then clearly holds, since we have (a) $6b_1+5b_2\geq 8b_1-3b_2$, (b) $(6b_1+5b_2)(2b_2-c)\geq (6b_1+5b_2)b_2/4\geq b_2^2\geq (3b_2-2b_1)^2$, (c) $b_2\geq 3b_2-2b_1$, and (d) $2b_2+3c\geq 2b_2+c$. We have thus shown that $\delta^*|_{h=\frac{2b_2-c}{3}}\geq-\frac{2b_2}{3}$.

Substituting $(h,\delta^*)=(\frac{2b_2-c}{3},\delta^*|_{h=\frac{2b_2-c}{3}})$ on the left-hand side of (\ref{eqn:fig-b-second}), we get
$$
  (b_2-c)(2b_2^2+4b_2c+3(b_2+c)\delta^*|_{h=\frac{2b_2-c}{3}})\leq 0
$$
for $c\leq b_2$, since $\delta^*|_{h=\frac{2b_2-c}{3}}\geq-\frac{2b_2}{3}$.

The expression is thus nonpositive at the entry points. We now proceed to prove that the expression is nonnegative at the exit points.

We now consider the exit point $(0,\delta^*|_{h=0})$. Substituting $h=0$ on the left-hand side of (\ref{eqn:fig-c-second}), we obtain
$$
  \delta^*|_{h=0}=\frac{b_1}{2}-\frac{b_2}{4}-\frac{(3b_2-2b_1)}{36}\sqrt{\frac{3c(8b_1-3b_2)}{b_2(2b_2-c)}}
$$

Substituting $(h,\delta^*)=(0,\delta^*|_{h=0})$ on the left-hand side of (\ref{eqn:fig-b-second}), we get
$$
  27b_2^2c-16b_2c^2+(27b_2^2+6b_2c-12c^2)\delta^*|_{h=0}+9(2b_2-c)(\delta^*|_{h=0})^2.
$$
From Mathematica, this is nonnegative when $c\in[b_2,2(t-1.4)(b_1-b_2)+1.4b_2]$ (see \ref{app:c.2}(1)).

We now consider the exit point $(h|_{\delta^*=b_1-b_2},b_1-b_2)$. Substituting $\delta^*=b_1-b_2$ in (\ref{eqn:fig-c-second}), we obtain
$$
  h|_{\delta^*=b_1-b_2}=\frac{9b_2^2-4c(b_1+3b_2)+3b_2\sqrt{9b_2^2+4c(b_1+3b_2)}}{6(b_1+3b_2)}.
$$

Substituting $(h,\delta^*)=(h|_{\delta^*=b_1-b_2},b_1-b_2)$ in (\ref{eqn:fig-b-second}), we get
$$
  27b_1^2(b_1-b_2+c+h|_{\delta^*=b_1-b_2})-(3b_1+b_2)(3b_1-3b_2+2c+3h|_{\delta^*=b_1-b_2})^2.
$$
From Mathematica, this is nonnegative when $c\in[b_2,2(t-1.4)(b_1-b_2)+1.4b_2]$ (\ref{app:c.2}(2)). The expression is thus nonnegative at the exit points.

We have thus shown that there exists a $(h,\delta^*)\in[0,\frac{2b_2-c}{3}]\times[0,b_1-b_2]$ that simultaneously solves (\ref{eqn:fig-b-second}) and (\ref{eqn:fig-c-second}), for all values of $(c,b_1,b_2)$ in the statement of the theorem.

{\bf Step 3c:} We now prove that $\delta_1\in[h+\delta^*,\frac{b_1}{3}]$. To prove $\delta_1\geq h+\delta^*$, we first assume the contrapositive, and do the following. (a) Solving (\ref{eqn:fig-b-first}) and (\ref{eqn:fig-b-second}) simultaneously, we obtain $(h_{II},\delta^*_{II})$ in the mechanism depicted in Figure \ref{fig:b-new}. We prove that $((h_{II},\delta^*_{II})<(h_{III},\delta^*_{III}))$; (b) We then show $\int_{\delta^*_{II}}^{\frac{b_1}{3}}V_{II}(\delta)\,d\delta\geq\int_{\delta^*_{III}}^{\frac{b_1}{3}}V_{III}(\delta)\,d\delta=0$. But $\int_{\delta^*_{II}}^{\frac{b_1}{3}}V_{II}(\delta)\,d\delta$ is negative for $c\in[\alpha_1,tb_2]$ (from \ref{app:c.1}(3)), which is a contradiction. We now proceed to prove our claim.

Observe that when $\delta_1<h+\delta^*$, we have $a_1>1$, and the mechanism appears as depicted (in solid lines) in Figures \ref{fig:c-proof-1} and \ref{fig:c-proof-2}. We now solve the problem for the parameters $(h_{II},\delta^*_{II})$ in the mechanism depicted in Figure \ref{fig:b-ini}. We first prove that if $(h,\delta^*)$ satisfies (\ref{eqn:fig-b-second}), then $h$ increases with increase in $\delta^*$. Solving (\ref{eqn:fig-b-second}) for $h$, we get
\begin{multline*}
  h=\frac{9b_2^2-16b_2c-6\delta^*(b_2+2c)-9(\delta^*)^2}{6(4b_2+3\delta^*)}\\+\frac{3(b_2+\delta^*)\sqrt{9b_2^2+16b_2c+6\delta^*(3b_2+2c)+9(\delta^*)^2}}{6(4b_2+3\delta^*)}.
\end{multline*}
Denoting $X:=9b_2^2+16b_2c+6\delta^*(3b_2+2c)+9(\delta^*)^2$, and differentiating with respect to $\delta^*$, we get
\begin{align*}
  \frac{\partial h}{\partial\delta^*}&=\frac{(4b_2+3\delta^*)\left(-6(b_2+2c+3\delta^*)+3\sqrt{X}+\frac{3(b_2+\delta^*)(9b_2+6c+9\delta^*)}{\sqrt{X}}\right)}{6(4b_2+3\delta^*)^2}\\&\hspace*{0.5in}-\frac{3\left(9b_2^2-16b_2c-6\delta^*(b_2+2c)-9(\delta^*)^2+3(b_2+\delta^*)\sqrt{X}\right)}{6(4b_2+3\delta^*)^2}\\&=\frac{-51b_2^2-72b_2\delta^*-27(\delta^*)^2+3b_2\sqrt{X}+\frac{9(4b_2+3\delta^*)(b_2+\delta^*)(3b_2+2c+3\delta^*)}{\sqrt{X}}}{6(4b_2+3\delta^*)^2}\\&\geq 0
\end{align*}
if $(9(\delta^*)^2+24b_2\delta^*+17b_2^2)\sqrt{X}\leq b_2\cdot X+3(4b_2+3\delta^*)(b_2+\delta^*)(3b_2+2c+3\delta^*)$. Squaring this expression on both sides and simplifying, we have
$$
  4(4b_2+3\delta^*)^2(b_2-c)(b_2^2(9b_2+25c)+6b_2\delta^*(3b_2+5c)+9(\delta^*)^2(b_2+c))\leq 0
$$
which clearly is true for $c\geq b_2$. This proves that if $(h,\delta^*)$ satisfies (\ref{eqn:fig-b-second}), then $h$ increases monotonically in $\delta^*$.

We now claim that $(h_{II},\delta^*_{II})<(h_{III},\delta^*_{III})$. This is because if not, then (i) $(h_{II},\delta^*_{II})>(h_{III},\delta^*_{III})$ must hold, since $(h,\delta^*)$ in both the mechanisms satisfies (\ref{eqn:fig-b-second}), and $h$ monotonically increases in $\delta^*$; (ii) If $(h_{II},\delta^*_{II})>(h_{III},\delta^*_{III})$, then $\bar{\mu}(Z)=0$ cannot be true for both the mechanisms simultaneously (see Figure \ref{fig:c-proof-1}). We have proved our claim.

We now evaluate $\int_{\delta^*_{II}}^{\frac{b_1}{3}}V_{II}(\delta)\,d\delta-\int_{\delta^*_{III}}^{\frac{b_1}{3}}V_{III}(\delta)\,d\delta$. From (\ref{eqn:V-alternate}), we have $V'(\delta)=-\bar{\mu}(z:z\in D\backslash Z,z_1-z_2=\delta\}$. Observe from Figure \ref{fig:c-proof-2} that $V'_{II}(\delta)<V'_{III}(\delta)$ when $\delta\in(l,\delta_1^{II})$ for some $l\in[\delta^*_{III},\delta_1^{III}]$, $V'_{II}(\delta)>V'_{III}(\delta)$ when $\delta\in[\delta^*_{II},l)$, and $V'_{II}(\delta)=V'_{III}(\delta)$ when $\delta\in[\delta_1^{II},\frac{b_1}{3}]\cup\{l\}$. Since we have $V_{II}(\delta^*_{II})=V_{III}(\delta^*_{III})=V_{II}(\frac{b_1}{3})=V_{III}(\frac{b_1}{3})=0$, we conclude that $V_{II}(\delta)\geq V_{III}(\delta)$ when $\delta\in[\delta^*_{II},\frac{b_1}{3}]$. Further, from $V'(\delta)=-(c-2b_2)-3(\delta_1^{II}-\delta)=-(c-2b_2+3\delta_2^{II})+3\delta$ when $\delta\in[\delta^*_{II},b_1-b_2]$, we have $V'_{II}(\delta)\geq 0$ in that interval, and thus $V_{II}(\delta)=V_{II}(\delta^*_{II})+\int_{\delta^*_{II}}^{\delta}V'_{II}(\tilde{\delta})\,d\tilde{\delta}\geq 0$. Therefore,
\begin{multline*}
  \int_{\delta^*_{II}}^{\frac{b_1}{3}}V_{II}(\delta)\,d\delta-\int_{\delta^*_{III}}^{\frac{b_1}{3}}V_{III}(\delta)\,d\delta\\=\int_{\delta^*_{III}}^{\frac{b_1}{3}}(V_{II}(\delta)-V_{III}(\delta))\,d\delta+\int_{\delta^*_{II}}^{\delta^*_{III}}V_{II}(\delta)\,d\delta\,d\delta\geq 0.
\end{multline*}

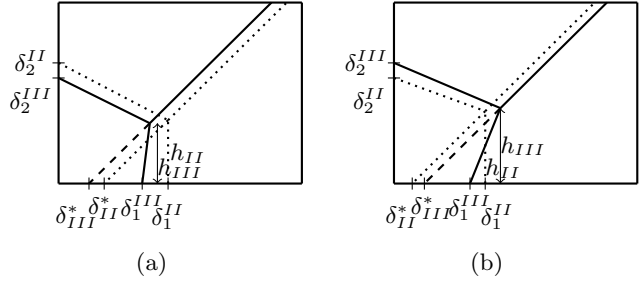
\begin{figure}[H]
\centering
\begin{tabular}{cc}
\subfloat[]{\label{fig:c-proof-1}\begin{tikzpicture}[scale=0.2,font=\footnotesize,axis/.style={very thick, -}]
\draw [axis,thick,-] (0,0)--(16,0);
\draw [axis,thick,-] (0,0)--(0,12);
\draw [axis,thick,-] (0,12)--(16,12);
\draw [axis,thick,-] (16,0)--(16,12);
\draw [thick,dotted] (0,8)--(7.2,4.2);
\draw [thin,-] (-0.4,8)--(0.4,8);
\node at (-1.8,8) {$\delta_2^{II}$};
\draw [axis,thick,-] (0,7)--(6,4);
\draw [thin,-] (-0.4,7)--(0.4,7);
\node [below] at (-1.8,7) {$\delta_2^{III}$};
\draw [thin,<->] (6.5,0)--(6.5,4);
\node at (8,1) {$h_{III}$};
\draw [thick,dotted] (7.2,0)--(7.2,4.2);
\node at (8.5,2) {$h_{II}$};
\draw [axis,thick,-] (6,4)--(5.5,0);
\draw [thin,-] (5.5,-0.4)--(5.5,0.4);
\node at (5.5,-1.6) {$\delta_1^{III}$};
\draw [thin,-] (7.2,-0.4)--(7.2,0.4);
\node at (7.2,-2.2) {$\delta_1^{II}$};
\draw [thick,dashed] (6,4)--(2,0);
\draw [thick] (6,4)--(14,12);
\draw [thin,-] (2,-0.4)--(2,0.4);
\node at (1.2,-2.2) {$\delta^*_{III}$};
\draw [thick,dotted] (7.2,4.2)--(3,0);
\draw [thick,dotted] (7.2,4.2)--(15,12);
\draw [thin,-] (3,-0.4)--(3,0.4);
\node at (3,-1.6) {$\delta^*_{II}$};
\end{tikzpicture}}&
\subfloat[]{\label{fig:c-proof-2}\begin{tikzpicture}[scale=0.2,font=\footnotesize,axis/.style={very thick, -}]
\draw [axis,thick,-] (0,0)--(16,0);
\draw [axis,thick,-] (0,0)--(0,12);
\draw [axis,thick,-] (0,12)--(16,12);
\draw [axis,thick,-] (16,0)--(16,12);
\draw [thick,dotted] (0,7)--(6,4.8);
\draw [thin,-] (-0.4,7)--(0.4,7);
\node [below] at (-1.6,7) {$\delta_2^{II}$};
\draw [axis,thick,-] (0,8)--(7,5);
\draw [thin,-] (-0.4,8)--(0.4,8);
\node at (-1.8,8) {$\delta_2^{III}$};
\draw [thin,<->] (7,0)--(7,5);
\node at (8.5,2.5) {$h_{III}$};
\draw [thick,dotted] (6,0)--(6,4.8);
\node at (7.2,1) {$h_{II}$};
\draw [axis,thick,-] (7,5)--(5,0);
\draw [thin,-] (5,-0.4)--(5,0.4);
\node at (5,-1.6) {$\delta_1^{III}$};
\draw [thin,-] (6,-0.4)--(6,0.4);
\node at (6.8,-2.2) {$\delta_1^{II}$};
\draw [thick,dashed] (7,5)--(2,0);
\draw [thick] (7,5)--(14,12);
\draw [thin,-] (2,-0.4)--(2,0.4);
\node at (2.5,-1.6) {$\delta^*_{III}$};
\draw [thick,dotted] (6,4.8)--(1.2,0);
\draw [thick,dotted] (6,4.8)--(13.2,12);
\draw [thin,-] (1.2,-0.4)--(1.2,0.4);
\node at (0.5,-2.2) {$\delta^*_{II}$};
\end{tikzpicture}}
\end{tabular}
\caption{Mechanism in Figure \ref{fig:b-ini} (in dotted lines) superimposed on mechanism in Figure \ref{fig:c-ini} (in solid lines). Figure indicates the case when (a) $(h_{II},\delta^*_{II})<(h_{III},\delta_{III})$; (b) $(h_{II},\delta^*_{II})<(h_{III},\delta_{III})$.}\label{fig:c-proof}
\end{figure}

\vspace*{10pt}
This proves that $\delta_1\geq h+\delta^*$, and also that $a_1\leq 1$. We then verify the upper bound $\delta_1\leq\frac{b_1}{3}$ via Mathematica (see \ref{app:c.2}(5)).

{\bf Step 3d:} We now prove that $\delta_2\in[\frac{2b_2-c}{3},\frac{b_2}{3}]$. Suppose that $\delta_2<\frac{2b_2-c}{3}$. Then, from $V(\delta)=-(c-2b_2+3\delta_2)+3\frac{\delta_2-h}{\delta_2+\delta^*}(\delta+\delta_2)$ when $\delta\in[-\delta_2,\delta^*]$, we have $V'(-\delta_2)=-(c-2b_2+3\delta_2)>0$. Also, $V'(\delta^*)\geq 0$ holds since $h\leq\frac{2b_2-c}{3}$. So we have $V'(\delta)\geq 0$ when $\delta\in[-\frac{b_2}{3},\delta^*]$. This implies that $V(\delta^*)>0$, a contradiction. So we have $\delta_2\geq\frac{2b_2-c}{3}\geq h$, and thus $a_2\leq 1$. The upper bound $\delta_2\leq\frac{b_2}{3}$ is verified via Mathematica (see \ref{app:c.2}(4)).

{\bf Step 4:} We now proceed to prove that the conditions in Theorem \ref{thm:Myerson} (2)--(4) are satisfied. Observe from the expressions of $V(\delta)$ it is nonpositive when $\delta\in[-b_2,\frac{-b_2}{3}]$ (i.e., in the interval where $q_1=0$), and nonnegative when $\delta\in[\frac{b_1}{3},b_1]$ (i.e., in the interval where $q_1=1$). This proves the conditions of Theorem \ref{thm:Myerson} (2) and \ref{thm:Myerson} (4).

We now prove the conditions in Theorem \ref{thm:Myerson} (3). The values of $V'(\delta)$ can be computed as
$$
  V'(\delta)=\begin{cases}-(c-2b_2+3\delta_2)+3\frac{\delta_2-h}{\delta_2+\delta^*}(\delta+\delta_2)&\delta\in(-\delta_2,\delta^*]\\-(c-2b_2)-3\frac{h}{\delta_1-\delta^*}(\delta_1-\delta)&\delta\in[\delta^*,b_1-b_2)\\-(c-2b_1+3\delta)-3\frac{h}{\delta_1-\delta^*}(\delta_1-\delta)&\delta\in(b_1-b_2,\delta_1)\end{cases}
$$
The values of $V'(\delta)$ when $\delta\in[-b_2,\delta_2)\cup(\delta_1,b_1]$ is the same as (\ref{eqn:menu-1-V'}).

The proof of $\int_{-\frac{b_2}{3}}^{x}V(\delta)\,d\delta\geq 0$ for every $x\in[-\frac{b_2}{3},\delta^*]$, is the same as that in the proof of Theorem \ref{thm:menu-2}. So we proceed to prove $\int_{\delta^*}^{x}V(\delta)\,d\delta\geq 0$ for every $x\in[\delta^*,\frac{b_1}{3}]$. Observe that $V'(\delta)$ is positive when $\delta\in[\delta^*,b_1-b_2)\cup(\delta_1,\frac{b_1}{3}]$. This is because (i) $V(\delta^*)\geq 0$ since $h\leq\frac{2b_2-c}{3}$, (ii) $V'(\delta)$ increasing in the interval $[\delta^*,b_1-b_2)$, and (iii) $\delta_1\leq\frac{b_1}{3}$. We now claim that $V'(\delta)\leq 0$ in some continuous subset of $[b_1-b_2,\delta_1]$. This is because (i) $V'(\delta)$ decreases in the interval $(b_1-b_2,\delta_1)$, and so when $V'(\delta)=0$ at some $l_1\in[b_1-b_2,\delta_1]$, then $V'(\delta)\leq 0$ for every $\delta\in[l_1,\delta_1)$; (ii) if $V'(\delta)>0$ for every $\delta\in(b_1-b_2,\delta_1)$, then $V(\delta)\geq 0$ throughout the interval, and thus $\int_{\delta^*}^{\frac{b_1}{3}}V(\delta)\,d\delta=0$ cannot be true. We have proved the claim.

Combining the fact that $V(\delta^*)=V(\frac{b_1}{3})=\int_{\delta^*}^{\frac{b_1}{3}}V(\delta)=\,d\delta=0$, with $V'(\delta)$ being nonnegative everywhere other than some continuous subset of $\delta\in(b_1-b_2,\delta_1)$, it is now easy to see that $\int_{\delta^*}^xV(\delta)\,d\delta\geq 0$ for all $x\in[\delta^*,\frac{b_1}{3}]$.
\qed

{\bf Proof of Theorem \ref{thm:menu-4,5} (i):}

{\bf Step 1:} We compute the virtual valuation function for the mechanism depicted in Figure \ref{fig:d'-new}.
$$
 V(\delta)=\begin{cases}\bar{\mu}(Z)+\frac{3}{2}\delta^2+2b_2\delta+\frac{b_2^2}{2}&\delta\in[-b_2,-\delta_2]\\V(-\delta_2)-(c-2b_2+3\delta_2)(\delta+\delta_2)+\frac{3}{2}\frac{\delta_2}{\delta_1+\delta_2}(\delta+\delta_2)^2&\delta\in[-\delta_2,\delta_1]\\V(\delta_1)+2b_2(\delta-\delta_1)&\delta\in[\delta_1,b']\\V(b')-b_1(\delta-b_1+b_2)-\frac{3}{2}((b_1-\delta)^2-b_2^2)&\delta\in[b',b_1]\end{cases}
$$
where $b_1-b_2$ is denoted as $b'$.

{\bf Step 2:} The mechanism has three unknowns -- $\delta_1$, $\delta_2$ and $a_1$. Since $q=\nabla u$, a conservative field, we must have the slope of the line separating $(0,0)$ and $(1-a,a)$ allocation regions satisfying $-\frac{1-a}{a}=-\frac{\delta_2}{\delta_1}$. This yields $a=\frac{\delta_1}{\delta_1+\delta_2}$.

We now compute the other two parameters by equating $\bar{\mu}(Z)=0$ and $\int_{-\frac{b_2}{3}}^{\frac{b_1}{2}-\frac{b_2}{4}}V(\delta)\,d\delta=0$. The latter condition follows from Theorem \ref{thm:Myerson} 3(c) because $q_1(\delta)=1-a\in(0,1)$ for $\delta\in[-\frac{b_2}{3},\frac{b_1}{2}-\frac{b_2}{4}]$.
\begin{equation}
\bar{\mu}(Z)=0\Rightarrow-\frac{3}{2}\delta_1\delta_2-c(\delta_1+\delta_2)+b_1b_2=0.\label{eqn:fig-d'-first}
\end{equation}
\begin{align}
&\int_{-\frac{b_2}{3}}^{\frac{b_1}{2}-\frac{b_2}{4}}V(\delta)\,d\delta=0\nonumber\\&\Rightarrow\int_{-\frac{b_2}{3}}^{-\delta_2}V(\delta)\,d\delta+\int_{-\delta_2}^{\delta_1}V(\delta)\,d\delta+\int_{\delta_1}^{\frac{b_1}{2}-\frac{b_2}{4}}V(\delta)\,d\delta=0\nonumber\\&\Rightarrow b_2(\delta_2^2-b_2^2/9)+\frac{1}{2}(b_2^3/27-\delta_2^3)+\frac{b_2^2}{2}(b_2/3-\delta_2)\nonumber\\&\hspace*{.1in}-(3\delta_1\delta_2/2+c(\delta_1+\delta_2)-b_2^2/2)(2b_1-b_2-4\delta_1)/4+b_2(2b_1-b_2)/4-b_2\delta_1^2\nonumber\\&\hspace*{.5in}-(2b_2\delta_2-3\delta_2^2/2-b_2^2/2)(\delta_1+\delta_2)-(c-2b_2+2\delta_2)\frac{(\delta_1+\delta_2)^2}{2}=0\nonumber\\&\Rightarrow\frac{2b_2^3}{27}+\frac{\delta_1-\delta_2}{2}(\delta_1\delta_2+c(\delta_1+\delta_2))-\frac{b_2}{16}(2b_1-b_2)^2-\frac{(2b_1-b_2)\bar{\mu}(Z)}{4}=0.\label{eqn:fig-d'-second}
\end{align}
The values of $\delta_1$ and $\delta_2$  can be obtained by solving (\ref{eqn:fig-d'-first}) and (\ref{eqn:fig-d'-second}) simultaneously.

{\bf Step 3:} We now evaluate the bounds on the variables in order to prove the existence of a meaningful solution that solves (\ref{eqn:fig-d'-first}) and (\ref{eqn:fig-d'-second}). Specifically, we show that there exists a $(\delta_1,\delta_2)\in[0,\frac{b_1}{2}-\frac{b_2}{4}]\times[0,\frac{b_2}{3}]$, as a simultaneous solution to (\ref{eqn:fig-d'-first}) and (\ref{eqn:fig-d'-second}). We show as in Step 3 of proof of Theorem \ref{thm:menu-1}.

We have $\delta_1|_{\delta_2}=\frac{b_1b_2-c\delta_2}{c+3\delta_2/2}$ and $\delta_2|_{\delta_1}=\frac{b_1b_2-c\delta_1}{c+3\delta_1/2}$ from (\ref{eqn:fig-d'-first}). It is clear that $\delta_2|_{\delta_1=x}$ is continuous in $x$, and also monotonically decreases in $x$. That $\delta_2|_{\delta_1=0}=b_1b_2/c\geq 0$ is also clear. We now verify if $\delta_2|_{\delta_1=\frac{b_1}{2}-\frac{b_2}{4}}\leq\frac{b_2}{3}$. Observe that $c\geq 3b_2/2$ from the statement of the theorem, since $c=3b_2/2$ makes the left-hand side of (\ref{eqn:menu-2-bound}) positive. We use $c\geq 3b_2/2$ crucially in the verification process.
$$
  \frac{b_1b_2-c(b_1/2-b_2/4)}{c+3b_1/4-3b_2/8}\leq\frac{b_1b_2-3b_2(b_1/2-b_2/4)/2}{3b_2/2+3b_1/4-3b_2/8}=\frac{b_1b_2/4+3b_2^2/8}{9b_2/8+3b_1/4}=\frac{b_2}{3},
$$
where the inequality occurs because $c\geq 3b_2/2$ from the statement of the theorem.

We now substitute $(\frac{b_1}{2}-\frac{b_2}{4},\delta_2|_{\delta_1=\frac{b_1}{2}-\frac{b_2}{4}})$ on the left-hand side of (\ref{eqn:fig-d'-second}), to obtain
\begin{multline*}
  -\left(\frac{(6b_1+b_2)^2}{864(6b_1-3b_2+8c)^2}\right)\\(72b_1^2b_2+144b_1b_2^2-90b_2^3+(-36b_1^2+84b_1b_2+399b_2^2)c-(96b_1+208b_2)c^2)
\end{multline*}
which is nonnegative for all $c\geq\beta$. We now substitute $(\delta_1|_{\delta_2=0},0)$ to obtain
$$
  \frac{1}{432}\frac{b_2}{c}(216b_1^2b_2-108b_1^2c+108b_1b_2c+5b_2^2c)\geq 0
$$
for all $c\leq\frac{216b_1^2b_2}{108b_1^2-108b_1b_2-5b_2^2}$. This shows that the left-hand side of (\ref{eqn:fig-d'-second}) is nonnegative at the entry points of the curve $(\delta_1|_{\delta_2},\delta_2)$ in the desired rectangle.

We now substitute $(0,\delta_2|_{\delta_1=0})$, and obtain
$$
  \frac{1}{432}\frac{b_2}{c}(108b_1b_2c+5b_2^2c-216b_1^2b_2-108b_1^2c)\leq 0
$$
because (i) $108b_1c(b_2-b_1)\leq 0$ for $b_1\geq b_2$, and (ii) $5b_2^2c\leq 216b_1^2b_2$ for $b_1\geq 3b_2/2$ and $c\leq 243b_2/38$. We now substitute $(\delta_1|_{\delta_2=\frac{b_2}{3}},\frac{b_2}{3})$, and obtain
$$
  \frac{b_2(6b_1+b_2)^2(5b_2^2+12b_2c-12c^2)}{432(b_2+2c)^2}\leq 0
$$
for $c\geq 3b_2/2$. Recall that $c\geq 3b_2/2$ holds from the theorem statement. This shows that the left-hand side of (\ref{eqn:fig-d'-second}) is nonpositive at the exit points of the curve $(\delta_1|_{\delta_2},\delta_2)$ in the desired rectangle.

We have thus shown that there exists $(\delta_1,\delta_2)\in[0,\frac{b_1}{2}-\frac{b_2}{4}]\times[0,\frac{b_2}{3}]$ that solves (\ref{eqn:fig-d'-first}) and (\ref{eqn:fig-d'-second}) simultaneously for all values of $(c,b_1,b_2)$ on the statement of the theorem.

{\bf Step 4:} We now proceed to prove that the conditions in Theorem \ref{thm:Myerson} (2)--(4) are satisfied. Observe that $V'(\delta)$ changes its sign from negative to positive only at $\delta=-\frac{2b_2}{3}$ in the interval where $q_1=0$, and from positive to negative only at $\delta=\max(\frac{2b_1}{3},b_1-b_2)$ in the interval where $q_1=1$. The proof of parts (2) and (4) now traces the same steps as in Theorem \ref{thm:menu-1}. Similarly, the proof that $\int_{-\frac{b_2}{3}}^{x}V(\delta)\,d\delta\geq 0$ holds for every $x\in[-\frac{b_2}{3},\frac{b_1}{2}-\frac{b_2}{4}]$, is the same as in Theorem \ref{thm:menu-2}. This completes the proof of optimality of the mechanism in Figure \ref{fig:d'-ini}.

At $c=\frac{216b_1^2b_2}{108b_1^2-108b_1b_2-5b_2^2}$, we obtain $\delta_2=0$, when we solve (\ref{eqn:fig-d'-first}) and (\ref{eqn:fig-d'-second}) simultaneously. A transition thus occurs from the structure depicted in Figure \ref{fig:d'-ini} to that in Figure \ref{fig:e'-ini}. We now show that the optimal mechanism as depicted in Figure \ref{fig:e'-new}, when $c\geq\frac{216b_1^2b_2}{108b_1^2-108b_1b_2-5b_2^2}$.

{\bf Step 1:} We first consider the zero allocation region to be $Z=([c,c+\frac{b_1b_2}{c}],c)$, and compute the virtual valuation function as follows.
$$
  V(\delta)=\begin{cases}\bar{\mu}(Z)+\frac{3}{2}\delta^2+2b_2\delta+\frac{b_2^2}{2}&\delta\in[-b_2,0]\\V(0)-(c-2b_2)\delta&\delta\in[0,\frac{b_1b_2}{c}]\\V(\frac{b_1b_2}{c})+2b_2(\delta-\frac{b_1b_2}{c})&\delta\in[\frac{b_1b_2}{c},b']\\V(b')-b_1(\delta-b_1+b_2)-\frac{3}{2}((b_1-\delta)^2-b_2^2)&\delta\in[b',b_2]\end{cases}
$$
where $b_1-b_2$ is denoted by $b'$. Observe that the zero allocation region only consists of a portion of the bottom boundary. So when we modify $q(z)$ to be $(0,1)$ instead of $(0,0)$ for every $z\in Z$, the utility $u(z)$ remains the same for every $z\in D$. The expected revenue, $\mathbb{E}_{z\sim f}[z\cdot q(z)-u(z)]$, also remains unchanged, because $\int_Zf(z)\,dz=0$. So we continue our analysis of the mechanism in Figure \ref{fig:e'-ini}, assuming the zero allocation region $Z$ to be non-empty.

The mechanism does not have any unknowns to compute. So steps 2 and 3 are not necessary. We move straightaway to step 4.

{\bf Step 4:} We now prove that the conditions in Theorem \ref{thm:Myerson} (2) and \ref{thm:Myerson} (4) are satisfied. Since $V'(\delta)$ changes sign from positive to negative only at $\delta=\max(\frac{2b_1}{3},b_1-b_2)$ in the interval where $q_1=1$, the proof for Theorem \ref{thm:Myerson} (4) traces the same steps as in Theorem \ref{thm:menu-1}. But $V'(\delta)$ changes sign at three values of $\delta$ in the interval where $q_1=0$. So proving the other condition needs more work.

We have $V(-b_2)=V(-\frac{b_2}{3})=0$, and $V(\delta)\leq 0$ when $\delta\in[-b_2,-\frac{b_2}{3}]$. We now evaluate $\int_{-\frac{b_2}{3}}^{\frac{b_1}{2}-\frac{b_2}{4}}V(\delta)\,d\delta$.
\begin{equation}\label{eqn:only-menu-v'}
  \int_{-\frac{b_2}{3}}^{\frac{b_1}{2}-\frac{b_2}{4}}V(\delta)\,d\delta=\frac{2}{27}b_2^3+\frac{b_1^2b_2^2}{2c}-\frac{1}{16}b_2(2b_1-b_2)^2\leq 0
\end{equation}
when $c\geq\frac{216b_1^2b_2}{108b_1^2-108b_1b_2-5b_2^2}$. So we have $\int_{-b_2}^{\frac{b_1}{2}-\frac{b_2}{4}}V(\delta)\,d\delta\leq 0$.

Observe that $V(\delta)$ is negative when $\delta\in[-b_2,-\frac{b_2}{3}]$, positive when $\delta\in[-\frac{b_2}{3},\frac{b_2^2}{2(c-2b_2)}]$, and negative again when $\delta\in[\frac{b_2^2}{2(c-2b_2)},\frac{b_1}{2}-\frac{b_2}{4}]$. So the integral $\int_{-b_2}^{x}V(\delta)\,d\delta$ thus attains its minimum either at $-\frac{b_2}{3}$ or at $\frac{b_1}{2}-\frac{b_2}{4}$. But from (\ref{eqn:only-menu-v'}), we have $\int_{-\frac{b_2}{3}}^{\frac{b_1}{2}-\frac{b_2}{4}}V(\delta)\,d\delta\leq 0$, and so the minimum cannot occur at $-\frac{b_2}{3}$. Therefore, $\int_{-b_2}^{x}V(\delta)\,d\delta\geq\int_{-\frac{b_2}{3}}^{\frac{b_1}{2}-\frac{b_2}{4}}V(\delta)\,d\delta$ holds for all $x\in[-b_2,\frac{b_1}{2}-\frac{b_2}{4}]$. Hence the result.\qed

\vspace*{10pt}
{\bf Proof of Theorem \ref{thm:menu-4,5}(ii):} Consider case (ii), where $b_1\in[b_2,3b_2/2]$. The values of $V(\delta)$ and the expression for $\bar{\mu}(Z)=0$ are the same as in the proof of Theorem \ref{thm:menu-4,5}(i). We thus skip step 1.

{\bf Step 2:} We now compute the expression for $-\int_{-\frac{b_2}{3}}^{\frac{b_1}{3}}V(\delta)\,d\delta=0$.
\begin{align}
  &-\int_{-\frac{b_2}{3}}^{\frac{b_1}{3}}V(\delta)\,d\delta=0\nonumber\\&\Rightarrow-\int_{-\frac{b_2}{3}}^{-\delta_2}V(\delta)\,d\delta-\int_{-\delta_2}^{\delta_1}V(\delta)\,d\delta-\int_{\delta_1}^{b_1-b_2}V(\delta)\,d\delta-\int_{b_1-b_2}^{\frac{b_1}{3}}V(\delta)\,d\delta=0\nonumber\\&\Rightarrow\frac{2}{27}(b_1^3-b_2^3)-\frac{\delta_1-\delta_2}{2}\left(b_1b_2-\frac{\delta_1\delta_2}{2}\right)-\left(b_1-b_2-\frac{\delta_1-\delta_2}{2}\right)\bar{\mu}(Z)=0.\label{eqn:fig-d-first}
\end{align}
The values of $(\delta_1,\delta_2)$ can be obtained by solving (\ref{eqn:fig-d'-first}) and (\ref{eqn:fig-d-first}) simultaneously.

{\bf Step 3:} We evaluate the bounds on the variables in order to prove the existence of a meaningful solution that solves (\ref{eqn:fig-d'-first}) and (\ref{eqn:fig-d-first}). Specifically, we show that there exists a $(\delta_1,\delta_2)\in[0,\frac{b_1}{3}]\times[0,\frac{b_2}{3}]$, as a simultaneous solution to (\ref{eqn:fig-d'-first}) and (\ref{eqn:fig-d-first}). Again, we show this as in Step 3 of proof of Theorem \ref{thm:menu-1}.

We have $\delta_1|_{\delta_2}=\frac{b_1b_2-c\delta_2}{c+3\delta_2/2}$ and $\delta_2|_{\delta_1}=\frac{b_1b_2-c\delta_1}{c+3\delta_1/2}$ from (\ref{eqn:fig-d'-first}).  It is clear that $\delta_2|_{\delta_1=x}$ is continuous in $x$, and also monotonically decreases in $x$. That $\delta_2|_{\delta_1=0}=b_1b_2/c\geq 0$ is also clear. We now verify if $\delta_2|_{\delta_1=\frac{b_1}{3}}\leq\frac{b_2}{3}$ for all $c\geq\alpha_2$.

We first observe from Mathematica that $\alpha_2\geq 1.36b_2+2(t-1.36)(b_1-b_2)$, with $t=3(37+3\sqrt{465})/176$ (see \ref{app:c.2}(7)). We now show that $1.36b_2+2(t-1.36)(b_1-b_2)\geq\frac{5}{2}\frac{b_1b_2}{b_1+b_2}$. This is same as showing that $4(t-1.36)b_1^2-2.28b_1b_2+(2.72-4(t-1.36))b_2^2\geq 0$, which is clearly true since (i) the roots of this quadratic expression are imaginary, and (ii) the coefficients of $b_1^2$ term and $b_2^2$ term are positive. We thus verify if $\delta_2|_{\delta_1=\frac{b_1}{3}}\leq\frac{b_2}{3}$ for all $c\geq\frac{5}{2}\frac{b_1b_2}{b_1+b_2}$.
$$
  \frac{b_1b_2-cb_1/3}{c+b_1/2}\leq\frac{b_1b_2-5b_1(b_1b_2)/(6(b_1+b_2))}{5b_1b_2/(2(b_1+b_2))+b_1/2}=\frac{b_1^2b_2/6+b_1b_2^2}{3b_1b_2+b_1^2/2}=\frac{b_2}{3},
$$
where the inequality occurs because of $c\geq\frac{5}{2}\frac{b_1b_2}{b_1+b_2}$.

We substitute $(\frac{b_1}{3},\delta_2|_{\delta_1=\frac{b_1}{3}})$ on the left-hand side of (\ref{eqn:fig-d-first}), to obtain
\begin{multline*}
  \frac{4b_1^5-16b_1b_2^3c-16b_2^3c^2+b_1^4(-6b_2+15c)}{54(b_1+2c)^2}\\+\frac{12b_1^3(3b_2^2-3b_2c+c^2)-4b_1^2b_2(b_2^2-27b_2c+18c^2)}{54(b_1+2c)^2}
\end{multline*}
We now verify if this expression is nonpositive for every $c\geq\alpha_2$, $b_1\in[b_2,3b_2/2]$. (Recall that $\alpha_2\geq 1.36b_2+2(t-1.36)(b_1-b_2)$). We now prove that the expression is nonpositive when $c=1.36b_2+2(t-1.36)(b_1-b_2)$, $b_1\in[b_2,3b_2/2]$, and that it is decreasing in $c$. Substituting $c=1.36b_2+2(t-1.36)(b_1-b_2)$, we have
\begin{multline*}
  (21.8931)b_1^5-(52.8447)b_1^4b_2+(33.1421)b_1^3b_2^2\\+(14.2829)b_1^2b_2^3-(24.4661)b_1b_2^4-(6.01705)b_2^5\leq 0
\end{multline*}
for $b_1\in[b_2,3b_2/2]$. Differentiating the numerator with respect to $c$, we have
$$
  -16b_1b_2^3-32b_2^3c+15b_1^4-36b_1^3b_2+108b_1^2b_2(b_2-c)+24b_1^2c(b_1-3b_2/2)\leq 0
$$
for every $c\geq b_2$, $b_1\in[b_2,3b_2/2]$. We now substitute $(\delta_1|_{\delta_2=0},0)$ in (\ref{eqn:fig-d-first}) to obtain $2/27(b_1^3-b_2^3)-b_1^2b_2^2/(2c)\leq 0$ when $c\leq\frac{27b_1^2b_2^2}{4(b_1^3-b_2^3)}$. This shows that the left-hand side of (\ref{eqn:fig-d-first}) is nonpositive at the entry points of the curve $(\delta_1|_{\delta_2},\delta_2)$ in the desired rectangle.

When $\delta_1=0$, we have $\delta_2=b_1b_2/c$, and thus substituting $(0,\delta_2|_{\delta_1=0})$ on the left-hand side of (\ref{eqn:fig-d-first}), we get $2(b_1^3-b_2^3)/27+(b_1b_2)^2/(2c)\geq 0$. We now substitute $(\delta_1|_{\delta_2=\frac{b_2}{3}},\frac{b_2}{3})$, to obtain
\begin{multline*}
  \frac{4b_1^3(b_2+2c)^2-36b_1^2b_2^2(b_2+3c)+6b_1b_2^2(b_2^2+6b_2c+12c^2)}{54(b-2+2c)^2}\\-\frac{b_2^3(4b_2^2+15b_2c+12c^2)}{54(b_2+2c)^2}
\end{multline*}
We now verify if this expression is nonnegative for every $c\geq\alpha_2$, $b_1\in[b_2,3b_2/2]$. We verify this for $c\geq 1.36b_2+2(t-1.36)(b_1-b_2)$, and prove that it is increasing in $c$. Substituting $c=1.36b_2+2(t-1.36)(b_1-b_2)$, we have
\begin{multline*}
  (8.92237)b_1^5+(26.6023)b_1^4b_2-(20.6703)b_1^3b_2^2\\-(16.0947)b_1^2b_2^3+(32.9614)b_1b_2^4-(17.7114)b_2^5\geq 0
\end{multline*}
for every $b_1\geq b_2$. Differentiating the numerator with respect to $c$, we get
$$
  16b_1^3(b_2+2c)+36b_1b_2^3-15b_2^4+108b_1b_2^2(c-b_1)+24b_2^2c(3b_1/2-b_2)\geq 0
$$
when $c\geq b_1\geq b_2$. This shows that the left-hand side of (\ref{eqn:fig-d-first}) is nonnegative at the exit points of the curve $(\delta_1|_{\delta_2},\delta_2)$ in the desired rectangle.

We have thus shown that there exists $(\delta_1,\delta_2)\in[0,\frac{b_1}{3}]\times[0,\frac{b_2}{3}]$ that solves (\ref{eqn:fig-d'-first}) and (\ref{eqn:fig-d-first}) simultaneously for all values of $(c,b_1,b_2)$ in the statement of the theorem.

{\bf Step 4:} We now prove the conditions of Theorem \ref{thm:Myerson}(b)--(d) are satisfied.  Observe that $V'(\delta)$ changes its sign from negative to positive only at $\delta=-\frac{2b_2}{3}$ in the interval where $q_1=0$, and from positive to negative only at $\delta=\frac{2b_1}{3}$ in the interval where $q_1=1$. The proof of parts (2) and (4) now traces the same steps as in Theorem \ref{thm:menu-1}.

It only remains to prove that $\int_{-\frac{b_2}{3}}^{x}V(\delta)\,d\delta\geq 0$ holds for every $x\in[-\frac{b_2}{3},\frac{b_1}{3}]$. We consider two cases: (a) $c\geq 2b_2$, (b) $c\in[\alpha_2,2b_2]$. Consider case (a). We have $V'(\delta)\geq 0$ when $\delta\in[-\frac{b_2}{3},-\delta_2]$, $V'(\delta)\leq 0$ when $\delta\in[-\delta_2,\delta_1]$ (since $c\geq 2b_2$), and $V'(\delta)\geq 0$ when $\delta\in[\delta_1,\frac{b_1}{3}]$. We also have $V(-\frac{b_2}{3})=V(\frac{b_1}{3})=0$. It follows that $V(\delta)\geq 0$ when $\delta\in[-\frac{b_2}{3},\frac{b_1}{3}]$, and $\int_{-\frac{b_2}{3}}^{x}V(\delta)\,d\delta\geq\int_{-\frac{b_2}{3}}^{\frac{b_1}{3}}V(\delta)\,d\delta\geq 0$ for every $x\in[-\frac{b_2}{3},\frac{b_1}{3}]$.

In case (b), we prove $\int_{-\frac{b_2}{3}}^{x}V(\delta)\,d\delta\geq 0$ holds for every $x\in[-\frac{b_2}{3},\frac{b_1}{3}]$, by comparing the mechanism in Figure \ref{fig:d-ini} with that in Figure \ref{fig:c-ini}. We first prove that $(\delta_1^{IV},\delta_2^{IV})$ obtained by solving (\ref{eqn:fig-d'-first}) and (\ref{eqn:fig-d-first}) in the mechanism depicted in Figure \ref{fig:d-ini}, is at most the value of $(\delta_1^{III},\delta_2^{III})$ values obtained in Figure \ref{fig:c-ini}. We then argue that $\int_{-\frac{b_2}{3}}^{x}(V_{IV}(\delta)-V_{III}(\delta))\geq 0$ for every $x\in[\frac{b_2}{3},\frac{b_1}{3}]$. Since we know that condition 3d in Theorem \ref{thm:Myerson} holds for the mechanism in Figure \ref{fig:c-ini}, the proof is complete.

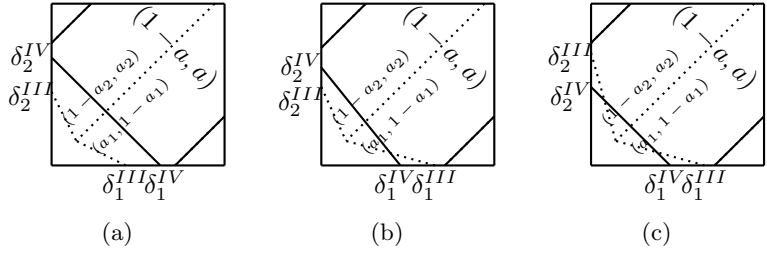
\begin{figure}[H]
\centering
\begin{tabular}{ccc}
\subfloat[]{\label{fig:menu IV - 1}\begin{tikzpicture}[scale=0.13,font=\small,axis/.style={very thick, -}]
\draw [axis,thick,-] (0,0)--(17.5,0);
\draw [axis,thick,-] (0,0)--(0,16.5);
\draw [axis,thick,-] (0,16.5)--(17.5,16.5);
\draw [axis,thick,-] (17.5,0)--(17.5,16.5);
\draw [axis,thick,-] (0,11)--(11,0);
\draw [thick,dotted] (0,12.5)--(4,16.5);
\draw [thick,dotted] (12.5,0)--(17.5,5);
\draw [thick,dotted] (0,7.5)--(2.5,2.5);
\draw [thick,dotted] (7.5,0)--(2.5,2.5);
\draw [axis,thick,-] (12.5,0)--(17.5,5);
\draw [axis,thick,-] (0,12.5)--(4,16.5);
\draw [thick,dotted] (16.5,16.5)--(2.5,2.5);
\node [rotate=45] at (5,8) {\tiny$(1-a_2,a_2)$};
\node [rotate=45] at (8,5) {\tiny$(a_1,1-a_1)$};
\node [rotate=-45] at (12.5,12) {$(1-a,a)$};
\node at (7.5,-2) {$\delta_1^{III}$};
\node at (11.5,-2) {$\delta_1^{IV}$};
\node at (-2,11) {$\delta_2^{IV}$};
\node at (-2,7) {$\delta_2^{III}$};
\end{tikzpicture}}&
\subfloat[]{\label{fig:menu IV - 2}\begin{tikzpicture}[scale=0.13,font=\small,axis/.style={very thick, -}]
\draw [axis,thick,-] (0,0)--(17.5,0);
\draw [axis,thick,-] (0,0)--(0,16.5);
\draw [axis,thick,-] (0,16.5)--(17.5,16.5);
\draw [axis,thick,-] (17.5,0)--(17.5,16.5);
\draw [axis,thick,-] (0,10)--(8,0);
\draw [axis,thick,-] (0,12)--(4.5,16.5);
\draw [axis,thick,-] (12.5,0)--(17.5,5);
\draw [thick,dotted] (0,8)--(2.5,2.5);
\draw [thick,dotted] (11.5,0)--(2.5,2.5);
\draw [thick,dotted] (12.5,0)--(17.5,5);
\draw [thick,dotted] (0,12)--(4.5,16.5);
\draw [thick,dotted] (16.5,16.5)--(2.5,2.5);
\node [rotate=45] at (5,8) {\tiny$(1-a_2,a_2)$};
\node [rotate=45] at (8,5) {\tiny$(a_1,1-a_1)$};
\node [rotate=-45] at (12.5,12) {$(1-a,a)$};
\node at (7.5,-2) {$\delta_1^{IV}$};
\node at (11.5,-2) {$\delta_1^{III}$};
\node at (-2,10) {$\delta_2^{IV}$};
\node at (-2,6.5) {$\delta_2^{III}$};
\end{tikzpicture}}&
\subfloat[]{\label{fig:menu IV - 3}\begin{tikzpicture}[scale=0.13,font=\small,axis/.style={very thick, -}]
\draw [axis,thick,-] (0,0)--(17.5,0);
\draw [axis,thick,-] (0,0)--(0,16.5);
\draw [axis,thick,-] (0,16.5)--(17.5,16.5);
\draw [axis,thick,-] (17.5,0)--(17.5,16.5);
\draw [axis,thick,-] (0,8)--(8,0);
\draw [axis,thick,-] (0,12.5)--(4,16.5);
\draw [axis,thick,-] (12.5,0)--(17.5,5);
\draw [thick,dotted] (0,11.5)--(2.5,2.5);
\draw [thick,dotted] (11.5,0)--(2.5,2.5);
\draw [thick,dotted] (12.5,0)--(17.5,5);
\draw [thick,dotted] (0,12.5)--(4,16.5);
\draw [thick,dotted] (16.5,16.5)--(2.5,2.5);
\node [rotate=45] at (5,8) {\tiny$(1-a_2,a_2)$};
\node [rotate=45] at (8,5) {\tiny$(a_1,1-a_1)$};
\node [rotate=-45] at (12.5,12) {$(1-a,a)$};
\node at (7.5,-2) {$\delta_1^{IV}$};
\node at (11.5,-2) {$\delta_1^{III}$};
\node at (-2,11) {$\delta_2^{III}$};
\node at (-2,7) {$\delta_2^{IV}$};
\end{tikzpicture}}
\end{tabular}
\caption{The mechanisms in Figures \ref{fig:c-ini} and \ref{fig:d-ini} superimposed on each other, when (a) $(\delta_1^{III},\delta_2^{III})<(\delta_1^{IV},\delta_2^{IV})$; (b) $\delta_1^{III}>\delta_1^{IV}$, and $\delta_2^{III}<\delta_2^{IV}$; (c) $(\delta_1^{III},\delta_2^{III})>(\delta_1^{IV},\delta_2^{IV})$. The mechanism in Figure \ref{fig:c-ini} is denoted by dotted lines, and that in Figure \ref{fig:d-ini} by solid lines.}\label{fig:menu IV}
\end{figure}

\vspace*{10pt}
We now prove that $(\delta_1^{IV},\delta_2^{IV})<(\delta_1^{III},\delta_2^{III})$. Suppose not. We have two cases: (i) $(\delta_1^{IV},\delta_2^{IV})>(\delta_1^{III},\delta_2^{III})$, (ii) One of $(\delta_1^{IV},\delta_2^{IV})$, say $\delta_2^{IV}$, is greater than $\delta_2^{III}$. From Mathematica, we have $a_1+a_2<1$ when $c\in[\alpha_2,2b_2)$ (see \ref{app:c.2}(6)). Thus the mechanisms depicted in Figures \ref{fig:b-ini} and \ref{fig:d-ini} appear as in Figure \ref{fig:menu IV - 1} for case (i), and as in Figure \ref{fig:menu IV - 2} for case (ii).

Consider case (i). We have $\bar{\mu}_{III}(Z)=\bar{\mu}_{IV}(Z)-\mbox{a negative number}>\bar{\mu}_{III}(Z)$. So $\bar{\mu}(Z)=0$ cannot hold simultaneously for both the mechanisms, a contradiction. Consider case (ii). We then have $V'_{III}(\delta)>V'_{IV}(\delta)$ for $\delta\in(-\delta_2^{IV},l_1)$ for some $l_1\in[-\delta_2^{III},\delta_1^{IV}]$, $V'_{III}(\delta)<V'_{IV}(\delta)$ for $\delta\in(l_1,\delta_1^{III})$, and $V'_{III}(\delta)=V'_{III}(\delta)$ for $\delta\in[-\frac{b_2}{3},-\delta_2^{IV}]\cup\{l_1\}\cup[\delta_1^{III},\frac{b_1}{3}]$. We also have $V_{III}(-\frac{b_2}{3})=V_{III}(\frac{b_1}{3})=V_{IV}(-\frac{b_2}{3})=V_{IV}(\frac{b_1}{3})=0$. So $V_{III}(\delta)-V_{IV}(\delta)=V_{III}(-\frac{b_2}{3})-V_{IV}(-\frac{b_2}{3})+\int_{-\frac{b_2}{3}}^{\delta}(V'_{III}(\tilde{\delta})-V'_{IV}(\tilde{\delta}))\,d\tilde{\delta}>0$ for all $\delta\in(-\frac{b_2}{3},\frac{b_1}{3})$. Thus $\int_{-\frac{b_2}{3}}^{\frac{b_1}{3}}(V_{III}(\delta)-V_{IV}(\delta))\,d\delta>0$, a contradiction. We have proved our claim.

We thus have $(\delta_1^{III},\delta_2^{III})>(\delta_1^{IV},\delta_2^{IV})$, and the mechanisms appear as in Figure \ref{fig:menu IV - 3}. We have $V'_{III}(\delta)<V'_{IV}(\delta)$ for $\delta\in(-\delta_2^{III},l_1)\cup(l_2,\delta_1^{III})$ for some $l_1\in[-\delta_2^{IV},\delta_1^{IV}]$ and $l_2\in[l_1,\delta_1^{IV}]$, $V'_{III}(\delta)>V'_{IV}(\delta)$ for $\delta\in(l_1,l_2)$, and $V'_{III}(\delta)=V'_{III}(\delta)$ for $\delta\in[-\frac{b_2}{3},-\delta_2^{III}]\cup\{l_1,l_2\}\cup[\delta_1^{III},\frac{b_1}{3}]$. We also have $V_{III}(-\frac{b_2}{3})=V_{III}(\frac{b_1}{3})=V_{IV}(-\frac{b_2}{3})=V_{IV}(\frac{b_1}{3})=\int_{-\frac{b_2}{3}}^{\frac{b_1}{3}}V_{III}(\delta)\,d\delta=\int_{-\frac{b_2}{3}}^{\frac{b_1}{3}}V_{IV}(\delta)\,d\delta=0$. We now have a series of observations.
\begin{itemize}
\item[$\bullet$] $V_{IV}(\delta)-V_{III}(\delta)=V_{IV}(-\frac{b_2}{3})-V_{III}(-\frac{b_2}{3})+\int_{-\frac{b_2}{3}}^{\delta}(V'_{IV}(\tilde{\delta})-V'_{III}(\tilde{\delta}))\,d\tilde{\delta}\geq 0$ when $\delta\in[-\frac{b_2}{3},l_1]$.
\item[$\bullet$] $V_{IV}(\delta)-V_{III}(\delta)=V_{IV}(\frac{b_1}{3})-V_{III}(\frac{b_1}{3})-\int_{\delta}^{\frac{b_1}{3}}(V'_{IV}(\tilde{\delta})-V'_{III}(\tilde{\delta}))\,d\tilde{\delta}\leq 0$ when $\delta\in[l_2,\frac{b_1}{3}]$.
\item[$\bullet$] By a similar argument, it is easy to see that $V_{IV}(\delta)-V_{III}(\delta)$ is nonnegative when $\delta\in[l_1,m]$ for some $m\in[l_1,l_2]$, and is nonpositive when $\delta\in[m,l_2]$. We thus have $V_{IV}(\delta)\geq V_{III}(\delta)$ when $\delta\in[-\frac{b_2}{3},m]$, and $V_{IV}(\delta)\leq V_{III}(\delta)$ when $\delta\in[m,\frac{b_1}{3}]$.
\item[$\bullet$] $\int_{-\frac{b_2}{3}}^{x}(V_{IV}(\delta)-V_{III}(\delta))\,d\delta\geq 0$ when $x\in[-\frac{b_2}{3},m]$.
\item[$\bullet$] $\int_{-\frac{b_2}{3}}^{x}(V_{IV}(\delta)-V_{III}(\delta))\,d\delta=-\int_x^{\frac{b_1}{3}}(V_{IV}(\delta)-V_{III}(\delta))\,d\delta\geq 0$ for any $x\in[m,\frac{b_1}{3}]$.
\item[$\bullet$] Notice that $\int_{-\frac{b_2}{3}}^{x}V_{III}(\delta)\,d\delta\geq 0$ holds for any $x\in[-\frac{b_2}{3},\frac{b_1}{3}]$, and thus $\int_{-\frac{b_2}{3}}^{x}V_{IV}(\delta)\,d\delta\geq 0$ now follows.
\end{itemize}
This completes the proof of optimality of the mechanism depicted in \ref{fig:d-ini}.

For the proof of optimality of the mechanism depicted in \ref{fig:e-ini} when $c\geq\frac{27b_1^2b_2^2}{4(b_1^3-b_2^3)}$, we note that the proof is exactly the same as in the proof of Theorem \ref{thm:menu-4,5}(i), except for the term $\int_{-\frac{b_2}{3}}^{\frac{b_1}{3}}V(\delta)\,d\delta=\frac{2}{27}(b_2^3-b_1^3)+\frac{b_1^2b_2^2}{2c}$. The expression clearly is negative when $c\geq\frac{27b_1^2b_2^2}{4(b_1^3-b_2^3)}$.\qed

\vspace*{10pt}
{\bf Proof of Theorem \ref{thm:extension}:} We fix $c_1-c_2=d$. Observe that the domain of the functions $(q_1,t)$ is the interval $[d-b_2,d-b_1]$. But it can be verified that all the results hold even for a shifted version of the domain. So we redefine $\delta=z_1-z_2-d$, and retain the domain to be $[-b_2,b_1]$.

{\bf Step 1:} We compute the virtual valuation function for the mechanism depicted in Figure \ref{fig:a-new}.
\begin{equation*}
V(\delta)=\begin{cases}\bar{\mu}(Z)+\frac{3}{2}\delta^2+2b_2\delta+\frac{b_2^2}{2}+d(\delta+b_2)&\delta\in[-b_2,-\delta_2]\\V(\delta_2)-(c_2-2b_2+3\delta_2)(\delta+\delta_2)&\delta\in[-\delta_2,\delta^*]\\V(\delta^*)-(c_2-2b_2)(\delta-\delta^*)+\frac{3}{2}((\delta_1-\delta)^2-\delta_2^2)&\delta\in[\delta^*,b_1-b_2]\\V(b_1-b_2)-(c_1-2b_1+3\delta_1)(\delta-b_1+b_2)&\delta\in[b_1-b_2,\delta_1]\\-\frac{3}{2}\delta^2+2b_1\delta-\frac{b_1^2}{2}-d(\delta-b_1)&\delta\in[\delta_1,b_1]\end{cases}
\end{equation*}

{\bf Step 2:} The mechanism has three unknowns -- $\delta^*$, $\delta_1$, and $\delta_2$. Observe that the line between the points $(c_1+b_2+\delta^*,c_2+b_2)$ and $(c_1+\delta^*,c_2)$ passes through $(c_1+\delta_1,c_2+\delta_2)$. So we have $\delta^*=\delta_1-\delta_2$.

We now proceed to compute $\delta_1$ and $\delta_2$. We do so by equating $\bar{\mu}(Z)=0$ and $V(\delta^*)=0$. The latter follows from Theorem \ref{thm:Myerson} because $q_1=0$ for $\delta\in[-b_2,\delta^*]$.
\begin{align}
  \bar{\mu}(Z)=0&\Rightarrow-3\delta_1\delta_2-c_2\delta_1-c_1\delta_2+b_1b_2=0.\label{eqn:fig-app-first}\\
  V(\delta^*)=0&\Rightarrow-\frac{3}{2}\delta_2^2+2b_2\delta_2-\frac{b_2^2}{2}-d(b_2-\delta_2)+(c_2-2b_2+3\delta_2)\delta_1=0.\label{eqn:fig-app-second}
\end{align}

The values of $\delta_1$ and $\delta_2$ can be computed by solving (\ref{eqn:fig-app-first}) and (\ref{eqn:fig-app-second}) simultaneously.

{\bf Step 3:} We now evaluate the bounds on $\delta_1$ and $\delta_2$ in order to prove the existence of a meaningful solution that solves (\ref{eqn:fig-app-first}) and (\ref{eqn:fig-app-second}) simultaneously. Specifically, we show that there exists a $(\delta_1,\delta_2)\in[\frac{b_1}{2}-\frac{c_1}{3}+\frac{c_1c_2}{6b_2},\frac{2b_1-c_1}{3}]\times[\frac{b_2+2d}{3},\frac{2b_2-c_2}{3}]$ as a simultaneous solution to (\ref{eqn:fig-app-first}) and (\ref{eqn:fig-app-second}). To show this, we do the following.
\begin{itemize}
\item We first show that there exists a $(\delta_1,\delta_2)\in[\frac{b_1}{2}-\frac{c_1}{3}+\frac{c_1c_2}{6b_2},\frac{2b_1-c_1}{3}]\times[\frac{b_2+2d}{3},\frac{2b_2-c_2}{3}]$ satisfying (\ref{eqn:fig-app-first}). We do this by showing that (a) $\delta_1|_{\delta_2=x}$ is continuous in $x$, and (b) $\delta_1|_{\delta_2=\frac{2b_2-c_2}{3}}=\frac{b_1}{2}-\frac{c_1}{3}-\frac{c_1c_2}{6b_2}$. We further show that in addition to continuity, $\delta_1|_{\delta_2=x}$ is also monotone; it decreases as $x$ increases.
\item It now suffices to show that the entry and the exit points of the curve $(\delta_1|_{\delta_2=x},x)$ in the rectangle $[\frac{b_1}{2}-\frac{c_1}{3}+\frac{c_1c_2}{6b_2},\frac{2b_1-c_1}{3}]\times[\frac{b_2+2d}{3},\frac{2b_2-c_2}{3}]$ changes sign when substituted on the left-hand side of (\ref{eqn:fig-a-second}). The entry point clearly is $(\frac{b_1}{2}-\frac{c_1}{3}+\frac{c_1c_2}{6b_2},\frac{2b_2-c_2}{3})$. The exit point is either $(\frac{2b_1-c_1}{3},\delta_2|_{\delta_1=\frac{2b_1-c_1}{3}})$ or $(\delta_1|_{\delta_2=\frac{b_2+2d}{3}},\frac{b_2+2d}{3})$. So we show that (a) substituting $(\frac{b_1}{2}-\frac{c_1}{3}+\frac{c_1c_2}{6b_2},\frac{2b_2-c_2}{3})$ on left-hand side of (\ref{eqn:fig-app-second}), makes the expression nonnegative, and (b) substituting $(\frac{2b_1-c_1}{3},\delta_2|_{\delta_1=\frac{2b_1-c_1}{3}})$ or $(\delta_1|_{\delta_2=\frac{b_2+2d}{3}},\frac{b_2+2d}{3})$ on left-hand side of (\ref{eqn:fig-app-second}), makes the expression nonpositive.
\end{itemize}

We now fill in the details. We have $\delta_1|_{\delta_2}=\frac{b_1b_2-c_1\delta_2}{3\delta_2+c_2}$ and $\delta_2|_{\delta_1}=\frac{b_1b_2-c_2\delta_1}{3\delta_1+c_1}$ from (\ref{eqn:fig-a-first}). It is clear that $\delta_1|_{\delta_2=x}$ is continuous in $x$, and also monotonically decreases in $x$. That $\delta_1|_{\delta_2=\frac{2b_2-c_2}{3}}=\frac{b_1}{2}-\frac{c_1}{3}-\frac{c_1c_2}{6b_2}$ can easily be computed by substituting $\delta_2=\frac{2b_2-c_2}{3}$.

We now consider the points $(\delta_1|_{\delta_2=\frac{b_2+2d}{3}},\frac{b_2+2d}{3})$ and $(\delta_1|_{\delta_2=\frac{2b_2-c_2}{3}},\frac{2b_2-c_2}{3})$. Substituting $\delta_1=\frac{b_1b_2-c_1\delta_2}{3\delta_2+c_2}$ from (\ref{eqn:fig-app-first}) in (\ref{eqn:fig-app-second}), we obtain
\begin{multline}\label{eqn:fig-app-delta_2}
  -9\delta_2^3-\delta_2^2(9c_2-12b_2)-\delta_2(2c_2^2-10b_2c_2+2b_2c_1+3b_2^2-6b_1b_2)\\-b_2^2c_2+2b_1b_2c_2-4b_1b_2^2-2b_2c_2(c_1-c_2)=0.
\end{multline}
When $\delta_2=\frac{2b_2-c_2}{3}$, the left-hand side of (\ref{eqn:fig-app-delta_2}) equals $\frac{2}{3}b_2(b_2+c_2)(b_2-2c_1+c_2)\geq 0$ for $2c_1-c_2\leq b_2$. Thus the expression is nonnegative at the entry point.

When $\delta_2=\frac{b_2+2d}{3}$, the left-hand side of (\ref{eqn:fig-app-delta_2}) equals $-\frac{2}{3}(3b_1b_2-c_1(b_2+2d))(b_2-2c_1+c_2)\leq 0$ for $c_1\leq b_2$, $d\leq b_2/2$, and $2c_1-c_2\leq b_2$. Thus the expression is nonpositive at the exit point $(\delta_1|_{\delta_2=\frac{b_2+2d}{3}},\frac{2b_2+2d}{3})$.

We now consider the point $(\frac{2b_1-c_1}{3},\delta_2|_{\delta_1=\frac{2b_1-c_1}{3}})$. Substituting $\delta_2=\frac{b_1b_2-c_2\delta_1}{3\delta_1+c_1}$ from (\ref{eqn:fig-app-first}) in (\ref{eqn:fig-app-second}), and obtain
\begin{multline}\label{eqn:fig-app-delta_1}
  -36b_2\delta_1^3+(9b_2(2b_1-b_2)+6b_2(c_2-7c_1)+3c_2^2)\delta_1^2\\+(12b_1b_2(b_2+c_1)-2b_2c_1(3b_2+8c_1)+8b_2c_1c_2+2c_1c_2^2)\delta_1\\-3b_1^2b_2^2+4b_1b_2^2c_1+2b_1b_2c_1^2-b_2^2c_1^2-2b_2c_1^3-2b_1b_2c_1c_2+2b_2c_1^2c_2=0.
\end{multline}
When $\delta_1=\frac{2b_1-c_1}{3}$, the left-hand side of (\ref{eqn:fig-app-delta_1}) equals $\frac{1}{3}(-8b_1^3b_2+2b_1b_2c_1c_2-c_1^2c_2^2+b_1^2(3b_2^2-8b_2c_1+8b_2c_2+4c_2^2))$. We now prove that this expression is nonpositive for all $c_1,c_2$ under consideration.
\begin{multline*}
 -8b_1^3b_2+2b_1b_2c_1c_2-c_1^2c_2^2+b_1^2(3b_2^2-8b_2(c_1-c_2)+4c_2^2)\\\leq -8b_1^3b_2+2b_1b_2c_1c_2-c_1^2c_2^2+b_1^2(3b_2^2+4c_2^2)\\\leq -8b_1^3b_2+4b_1^2b_2^2+4b_1^2c_2^2\leq -8b_1^2b_2(b_1-b_2)\leq 0
\end{multline*}
where the first inequality follows from $c_1\geq c_2$; the second inequality occurs because the expression is maximized when $c_1=\frac{b_1b_2}{c_2}$; the third inequality follows because when $c_2\in[0,b_2]$, the expression is maximized at $c_2=b_2$; and the final inequality occurs since $b_1\geq b_2$. Thus the expression is nonpositive at the exit point $(\frac{2b_1-c_1}{3},\delta_2|_{\delta_1=\frac{2b_1-c_1}{3}})$.

We have thus shown that there exists a $(\delta_1,\delta_2)\in[\frac{b_1}{2}-\frac{c_1}{3}+\frac{c_1c_2}{6b_2},\frac{2b_1-c_1}{3}]\times[\frac{b_2+2d}{3},\frac{2b_2-c_2}{3}]$ as a simultaneous solution to (\ref{eqn:fig-app-first}) and (\ref{eqn:fig-app-second}), when the values of $(c_1,c_2,b_1,b_2)$ satisfy the conditions in the statement of the theorem.

{\bf Step 4:} We now proceed to prove parts (c) and (d) in Theorem \ref{thm:Myerson} (2) and \ref{thm:Myerson} (4). We first compute $V'(\delta)$ for almost every $\delta\in[-b_2,b_1]$.

$$
  V'(\delta)=\begin{cases}3\delta+2b_2+d&\delta\in(-b_2,-\delta_2)\\-(c_2-2b_2+3\delta_2)&\delta\in(-\delta_2,\delta^*]\\-(c_2-2b_2)-3(\delta_1-\delta)&\delta\in[\delta^*,b_1-b_2)\\-(c_1-2b_1+3\delta_1)&\delta\in(b_1-b_2,\delta_1)\\-3\delta+2b_1-d&\delta\in(\delta_1,b_1)\end{cases}
$$

Observe that $V'(\delta)$ is negative when $\delta\in[-b_2,-\frac{2b_2+d}{3}]$, and positive when $\delta\in[-\frac{2b_2+d}{3},-\delta_2]$ (follows because $\delta_2\leq\frac{2b_2-c_2}{3})$. We also have $V(-b_2)=V(\delta^*)=0$. So $V(\delta)=V(-b_2)+\int_{-b_2}^{\delta}V'(\tilde{\delta})\,d\tilde{\delta}\leq 0$ for all $\delta\in[-b_2,\delta^*]$, and hence $\int_{-b_2}^{\delta^*}V(\delta)\,d\delta\leq 0$, and $\int_{-b_2}^xV(\delta)\,d\delta\geq\int_{-b_2}^{\delta^*}V(\delta)\,d\delta$ for all $x\in[-b_2,\delta^*]$.

We now claim that $V'(\delta)$ is positive when $\delta\in[\delta^*,\frac{2b_1-d}{3}]$, and negative when $\delta\in[\frac{2b_1-d}{3},b_1]$. Observe that $V'(\delta)$ is continuous at $\delta=\delta^*$, and that it increases in the interval $[\delta^*,b_1-b_2]$. So $V'(\delta)\geq 0$ when $\delta\in[\delta^*,b_1-b_2]$. Also, $V'(\delta)\geq 0$ when $\delta\in[b_1-b_2,\delta_1]$ because $\delta_1\leq\frac{2b_1-c_1}{3}$. That $V'(\delta)$ is positive when $\delta\in[\delta_1,\frac{2b_1-d}{3}]$, and negative when $\delta\in[\frac{2b_1-d}{3},b_1]$ is obvious. We have proved our claim.

Since we also have $V(b_1)=0=V(\delta^*)$, it follows that $V(\delta)=V(\delta^*)+\int_{\delta^*}^{\delta}V'(\tilde{\delta})\,d\tilde{\delta}\geq 0$ for all $\delta\in[\delta^*,b_1]$. So we have $\int_{\delta^*}^{b_1}V(\delta)\,d\delta\geq 0$ and $\int_{\delta^*}^xV(\delta)\,d\delta\leq\int_{\delta^*}^{b_1}V(\delta)\,d\delta$ for all $x\in[\delta^*,b_1]$.\qed

\section{The Weak Duality Result}\label{app:new}
In this section, we show the weak duality relationship between (\ref{eqn:dual}) and (\ref{eqn:primal}). Take the primal problem
$$
\max_{\substack{{u(z)-u(z')\leq\|z-z'\|_\infty}\\{u\mbox{ cont, conv, inc}}}}\int_Du(z)\,d\bar{\mu}(z)
$$
and rewrite it as
$$
\max_{(u\mbox{ cont, conv, inc})}\min_{\gamma\geq 0}\int_Du(z)\,d\bar{\mu}(z)+\int_{D\times D}(\|z-z'\|_\infty-u(z)+u(z'))\,d\gamma(z,z').
$$
We can do this because if $u(z)-u(z')>\|z-z'\|_\infty$, then the minimizer can choose an adverse $\gamma$ to make the second integral approach $-\infty$. The maximizer would not want this to happen and would hence choose $u$ to ensure that $u(z)-u(z')\leq\|z-z'\|_\infty$ for all pairs $z,z'$ in $D$. The quantity $d\gamma(z,z')$ is then a price measure for violating the constraint $u(z)-u(z')\leq\|z-z'\|_\infty$.

Let us now write the dual:
$$
\min_{\gamma\geq 0}\max_{(u\mbox{ cont, conv, inc})}\int_Du(z)\,d\bar{\mu}(z)+\int_{D\times D}(\|z-z'\|_\infty-u(z)+u(z'))\,d\gamma(z,z').
$$
Define $\gamma_1(z)=\int_D\gamma(z,dz')$ and $\gamma_2(z')=\int_D\gamma(dz,z')$. Now rewrite the dual as
$$
\min_{\gamma\geq 0}\max_{(u\mbox{ cont, conv, inc})}\int_Du(z)\,d(\bar{\mu}(z)-(\gamma_1(z)-\gamma_2(z)))+\int_{D\times D}\|z-z'\|_\infty\,d\gamma(z,z')
$$
and recognize it to be
$$
\min_{\gamma:\gamma_1-\gamma_2\succeq_{cvx}\bar{\mu}}\int_{D\times D}\|z-z'\|_\infty\,d\gamma(z,z').
$$
This is because if $\gamma$ satisfied $\gamma_1-\gamma_2\nsucceq_{cvx}\bar{\mu}$, then the maximizer can choose an adversarial $u$ with $\int_Du(z)\,d(\bar{\mu}(z)-(\gamma_1(z)-\gamma_2(z)))>0$, and drive the first integral to $\infty$.

This establishes weak duality and provides us with some understanding of how the dual arises and why $\gamma$ may be interpreted as prices for violating the primal constraint.

\section{Proofs using Mathematica}\label{app:c}
\subsection{Expressions used in Theorem \ref{thm:menu-2}}\label{app:c.1}
\begin{enumerate}
\item We first find the expression for $h$ that solves (\ref{eqn:fig-b-first}) and (\ref{eqn:fig-b-second}) simultaneously.\\

{\em Mathematica Input:}

$\bm{\delta^*=(b_1b_2-b_2^2/2-3h^2/2-ch)/(2b_2);}$\\
$\bm{\mbox{Solve}[27(c+h+\delta^*)(b_2+\delta^*)^2-4(4b_2+3\delta^*)(3/2(h+\delta^*)+c)^2==0,\,h]}$\\
     
{\em Mathematica Output:}

$\{\{h\rightarrow\mbox{ Root}[-72b_1^2b_2^3-144b_1b_2^4+90b_2^5+36b_1^2b_2^2c-84b_1b_2^3c-399b_2^4c+96b_1b_2^2c^2+208b_2^3c^2+(108b_1^2b_2^2+36b_1b_2^3-477b_2^4+432b_1b_2^2c+768b_2^3c-72b_1b_2c^2+84b_2^2c^2-96b_2c^3) \# 1+(432b_1b_2^2+684b_2^3-324b_1b_2c+90b_2^2c-504b_2c^2+36c^3) \# 1^2+(-324b_1b_2-54b_2^2-864b_2c+216c^2) \# 1^3+(-486b_2+405c) \# 1^4+243 \# 1^5\,\&,\,1]\}\mbox{ (all five roots)}\}$\\

In this subsection, we verify (i) $\delta_2\leq b_2/3$ when $b_1\geq b_2$, $c\in[b_2,2b_2]$, (ii) the left-hand side of (\ref{eqn:fig-b-third}) is nonnegative when $b_1\in[b_2,3b_2/2]$, $c\in[b_2,\alpha_1]$, and (iii) $2(t-1)(b_1-b_2)+b_2\geq\alpha_1$, where $t=3(37+3\sqrt{465})/176$. We will use bullet (1) above.

\item We now proceed to verify if $\delta_2\leq b_2/3$. From (\ref{eqn:fig-b-delta2-clear}), we have $\delta_2=\frac{b_1b_2-(3h/2+c)(h+\delta^*)}{3/2(h+\delta^*)+c}$. Observe that this is in terms of $(h,\delta^*)$ that are obtained by solving (\ref{eqn:fig-b-first}) and (\ref{eqn:fig-b-second}). We thus initialize the values of $h$ and $\delta^*$ using expressions from bullet (1) above, and then find the values of $(c,b_1,b_2)$ for which $\delta_2\leq\frac{b_2}{3}$ holds.\\

{\em Mathematica Input:}

$\bm{h=\mbox{Root}[-72b_1^2b_2^3-144b_1b_2^4+90b_2^5+36b_1^2b_2^2c-84b_1b_2^3c-399b_2^4c+96b_1b_2^2c^2+208b_2^3c^2+(108b_1^2b_2^2+36b_1b_2^3-477b_2^4+432b_1b_2^2c+768b_2^3c-72b_1b_2c^2+84b_2^2c^2-96b_2c^3)\# 1+(432b_1b_2^2+684b_2^3-324b_1b_2c+90b_2^2c-504b_2c^2+36c^3)\# 1^2+(-324b_1b_2-54b_2^2-864b_2c+216c^2)\# 1^3+(-486b_2+405c)\# 1^4+243\# 1^5\,\&,\,3];}$

$\bm{\delta^*=(b_1b_2-3/2h^2-ch-b_2^2/2)/(2b_2);\delta_2=b_1b_2(3h/2+c)(h+\delta^*)/(3(h+\delta^*)/2+c);}$

$\bm{\mbox{Reduce}[\delta_2\leq b_2/3\,\&\&\,0\leq b_2 \leq b_1\,\&\&\,b_2\leq c\leq 2b_2,\,\{b_2,b_1,c\}]}$\\

{\em Mathematica Output:}

$b_2>0\,\&\&\,b_1\geq b_2\,\&\&\,b_2\leq c\leq 2b_2$\\

The output indicates that $\delta_2\leq b_2/3$ holds for every $b_1\geq b_2$, $c\in[b_2,2b_2]$.

\item We then find the values of $c$ for which the left-hand side of (\ref{eqn:fig-b-third}) is nonnegative.\\

{\em Mathematica Input:}

$\bm{\mbox{Reduce}[-2b_1^3/27-b_2(\delta^*)^2+b_2\delta^*(b_1-b_2/2)+(c+h)h^2/2\geq 0\,\&\&\,0\leq b_2\leq b_1\leq 3b_2/2\,\&\&\,b_2\leq c\leq tb_2,\,\{b_2,b_1,c\}]}$\\

{\em Mathematica Output:}

$b_2>0\,\&\&\,b_2\leq b_1\leq 1.5b_2\,\&\&\,b_2\leq c\leq\mbox{ Root}[f_{c-II}(c)\,\&,\,2]$\\

Here, $f_{c-II}(c)$ is a polynomial of degree $12$. We have not written it here since it is too long. Let $\alpha_1=\mbox{ Root}[f_{c-II}(c)\,\&,\,2]$. Then this proves that the left-hand side of (\ref{eqn:fig-b-third}) is nonnegative for every $c\in[b_2,\alpha_1]$.

\item To prove that $\alpha_1\leq 2(t-1)(b_1-b_2)+b_2$, with $t=3(37+3\sqrt{465})/176$, we again find the values of $c$ for which the left-hand side of (\ref{eqn:fig-b-third}) is nonnegative, but with $c$ restricted to $c\in[b_2,2(t-1)(b_1-b_2)+b_2]$.\\

{\em Mathematica Input:}

$\bm{t=\frac{3}{176}(37+3\sqrt{465});\mbox{Reduce}[-2b_1^3/27-b_2(\delta^*)^2+b_2\delta^*(b_1-b_2/2)+(c+h)h^2/2\geq 0\,\&\&\,0\leq b_2\leq b_1\leq 3b_2/2\,\&\&\,b_2\leq c\leq 2(t-1)(b_1-b_2)+b_2,\,\{b_2,b_1,c\}]}$\\

{\em Mathematica Output:}

$b_2>0\,\&\&\,b_2\leq b_1\leq 1.5b_2\,\&\&\,b_2\leq c\leq\mbox{ Root}[f_{c-II}(c)\,\&,\,2]$\\

This proves that $\alpha\leq 2(t-1)(b_1-b_2)+b_2$. This completes all the proofs in Theorem \ref{thm:menu-2}.
\end{enumerate}

\subsection{Expressions used in Theorem \ref{thm:menu-3}:}\label{app:c.2}
\begin{enumerate}
\item We first prove that the left-hand side of (\ref{eqn:fig-b-second}) is nonnegative when $(h,\delta^*)=(0,\delta^*|_{h=0})$.\\

{\em Mathematica Input:}

$\bm{\delta^*=\frac{b_1}{2}-\frac{b_2}{4}-\frac{(3b_2-2b_1)\sqrt{3b_2c(8b_1-3b_2)(2b_2-c)}}{(36b_2(2b_2-c))};}$

$\bm{t=3(37+3\sqrt{465})/176;\mbox{Reduce}[27b_2^2c-16b_2c^2+(27b_2^2+6b_2c-12c^2)\delta^*+(18b_2-9c)(\delta^*)^2\geq 0\,\&\&\,0\leq b_2\leq b_1 \leq 1.5b_2\,\&\&\,b_2\leq c\leq 2(t-1.4b_2)(b_1-b_2)+1.4b_2,\,\{b_2,b_1,c\}]}$\\

{\em Mathematica Output:}

$b_2>0\,\&\&\,b_2\leq b_1\leq 1.5b_2\,\&\&\,b_2\leq c\leq 0.66676b_1+0.73324b_2$\\

\item We now prove that the left-hand side of (\ref{eqn:fig-b-second}) is nonnegative when $(h,\delta^*)=(h|_{\delta^*=b_1-b_2},b_1-b_2)$.\\

{\em Mathematica Input:}

$\bm{h=\frac{9b_2^2-4c(b_1+3b_2)+3b_2\sqrt{9b_2^2+4c(b_1+3b_2)}}{6(b_1+3b_2)};}$\\
$\bm{\mbox{Reduce}[(27b_1^2(b_1-b_2+c+h)-(3b_1+b_2)(3b_1-3b_2+2c+3h)^2)\geq 0\,\&\&\,0\leq b_2\leq b_1\leq 1.5b_2\,\&\&\,b_2\leq c\leq 2(t-1.4)(b_1-b_2)+1.4b_2,\{b_2,b_1,c\}]}$\\

{\em Mathematica Output:}

$b_2>0\,\&\&\,b_2\leq b_1\leq 1.5b_2\,\&\&\,b_2\leq c\leq 0.66676b_1+0.73324b_2$\\

\item We now find the expression for $\delta^*$ that solves (\ref{eqn:fig-b-second}) and (\ref{eqn:fig-c-second}) simultaneously.

{\em Mathematica Input:}

$\bm{h=\frac{9b_2^2-16b_2c-6\delta^*(b_2+2c)-9(\delta^*)^2}{6(4b_2+3\delta^*)}}\\+\bm{\frac{3(b_2+\delta^*)\sqrt{9b_2^2+16b_2c+6\delta^*(3b_2+2c)+9(\delta^*)^2}}{6(4b_2+3\delta^*)};}$

$\bm{\mbox{Solve}\left[\frac{b_1}{2}-\frac{b_2}{4}-\frac{(3b_2-2b_1)(2c+3h)}{12}\sqrt{\frac{(8b_1-3b_2)}{(3b_2(8b_2(c+h)-(2c+3h)^2))}}\right.}\\\bm{\left.==\delta^*,\delta^*\right]}$\\

{\em Mathematica output:}

$\{\{\delta^*\rightarrow-4b_2/3\}\mbox{(twice)},\{\delta^*\rightarrow\mbox{Root}[-16(b_1-b_2)^2b_2^4 (b_1^2-3b_1b_2+b2^2)^2c^2+(-4(b_1-b_2)b_2^3(b_1^2-3b_1b_2+b_2^2)(108b_1^2b_2^3-108b_1b_2^4+27b_2^5+24b_1^3b_2c+96b_1^2b_2^2c-96b_1b_2^3c+24b_2^4c+16b_1^3c^2-24b_1^2b_2c^2+44b_1b_2^2c^2-16b_2^3c^2))\# 1+(-144b_1^6b_2^4+4176b_1^4b_2^6-13752b_1^3b_2^7+15660b_1^2b_2^8-7218b_1b_2^9+1152b_2^10-256b_1^6b_2^3c+512b_1^5b_2^4c+4384b_1^4b_2^5c-18064b_1^3b_2^6c+21888b_1^2b_2^7c-10368b_1b_2^8c+1680b_2^9c-96b_1^6b_2^2c^2+320b_1^5b_2^3c^2-144b_1^4b_2^4c^2-1392b_1^3b_2^5c^2+1884b_1^2b_2^6c^2-768b_1b_2^7c^2+96b_2^8c^2)\# 1^2+(-192b_1^6b_2^3+624b_1^5b_2^4+4128b_1^4b_2^5-24828b_1^3b_2^6+42984b_1^2b_2^7-27432b_1b_2^8+5580b_2^9-224b_1^6b_2^2c+768b_1^5b_2^3c+4960b_1^4b_2^4c-31904b_1^3b_2^5c+58680b_1^2b_2^6c-38808b_1b_2^7c+8064b_2^8c-64b_1^6b_2c^2+96b_1^5b_2^2c^2+1216b_1^4b_2^3c^2-4120b_1^3b_2^4c^2+7632b_1^2b_2^5c^2-5056b_1b_2^6c^2+1016b_2^7c^2)\\\# 1^3+(-64b_1^6b_2^2+288b_1^5b_2^3+2400b_1^4b_2^4-20496b_1^3b_2^5+55512b_1^2b_2^6-52650b_1b_2^7+14364b_2^8-64b_1^6b_2c+288b_1^5b_2^2c+2688b_1^4b_2^3c-23560b_1^3b_2^4c+69120b_1^2b_2^5c-70320b_1b_2^6c+20040b_2^7c-16b_1^6c^2+1344b_1^4b_2^2c^2-3280b_1^3b_2^3c^2+2592b_1^2b_2^4c^2-4020b_1b_2^5c^2+1424b_2^6c^2)\# 1^4+(576b_1^4b_2^3-6864b_1^3b_2^4+33696b_1^2b_2^5-56628b_1b_2^6+24696b_2^7+576b_1^4b_2^2c-7152b_1^3b_2^3c+37584b_1^2b_2^4c-70380b_1b_2^5c+33360b_2^6c+432b_1^4b_2c^2-1560b_1^3b_2^2c^2-5184b_1^2b_2^3c^2+10224b_1b_2^4c^2-2616b_2^5c^2)\# 1^5+(-576b_1^3b_2^3+7776b_1^2b_2^4-32832b_1b_2^5+27396b_2^6-576b_1^3b_2^2c+7776b_1^2b_2^3c-36720b_1b_2^4c+34272b_2^5c-432b_1^3b_2c^2-2916b_1^2b_2^2c^2+15876b_1b_2^3c^2-9729b_2^4c^2)\\\# 1^6+(-108b_2^2(2b_2+3c)(36b_1b_2-76b_2^2-18b_1c+29b_2c))\# 1^7+(972b_2^2(2b_2-c)(2b_2+3c))\# 1^8\,\&,\,1]\}\mbox{ (all eight roots)}\};$\\

We verify (i) $\delta_2\leq b_2/3$ when $b_1\in[b_2,3b_2/2]$, $c\in[b_2,2b_2]$, (ii) $\delta_1\leq b_1/3$ when $b_1\in[b_2,3b_2/2]$, $c\in[b_1,2b_2]$, (iii) the left-hand side of (\ref{eqn:fig-c-third}) is nonnegative when $b_1\in[b_2,3b_2/2]$, $c\in[b_2,\alpha_2]$, and (iv) $2(t-1.36)(b_1-b_2)+1.36b_2\leq\alpha_2\leq 2(t-1.4)(b_1-b_2)+1.4b_2$, where $t=3(37+3\sqrt{465})/176$. We will use bullet (3) above.

\item We now verify if $\delta_2\leq\frac{b_2}{3}$. From the statement of Theorem \ref{thm:menu-3}, we have $\delta_2=\frac{b_2^2/2+(2b_2-c-3h/2)\delta^*}{3(h+\delta^*)/2+c}$. We now initialize $(h,\delta^*)$ as in bullet (3), and find the values of $(c,b_1,b_2)$ for which $\delta_2\leq b_2/3$.

{\em Mathematica Input:}

$\bm{\delta^*=\mbox{Root}[-16(b_1-b_2)^2b_2^4 (b_1^2-3b_1b_2+b2^2)^2c^2+(-4(b_1-b_2)b_2^3(b_1^2-3b_1b_2+b_2^2)(108b_1^2b_2^3-108b_1b_2^4+27b_2^5+24b_1^3b_2c+96b_1^2b_2^2c-96b_1b_2^3c+24b_2^4c+16b_1^3c^2-24b_1^2b_2c^2+44b_1b_2^2c^2-16b_2^3c^2))\# 1+(-144b_1^6b_2^4+4176b_1^4b_2^6-13752b_1^3b_2^7+15660b_1^2b_2^8-7218b_1b_2^9+1152b_2^10-256b_1^6b_2^3c+512b_1^5b_2^4c+4384b_1^4b_2^5c}\\\bm{-18064b_1^3b_2^6c+21888b_1^2b_2^7c-10368b_1b_2^8c+1680b_2^9c-96b_1^6b_2^2c^2+320b_1^5b_2^3c^2-144b_1^4b_2^4c^2-1392b_1^3b_2^5c^2+1884b_1^2b_2^6c^2-768b_1b_2^7c^2+96b_2^8c^2)\# 1^2+(-192b_1^6b_2^3+624b_1^5b_2^4+4128b_1^4b_2^5-24828b_1^3b_2^6+42984b_1^2b_2^7-27432b_1b_2^8+5580b_2^9-224b_1^6b_2^2c+768b_1^5b_2^3c}\\\bm{+4960b_1^4b_2^4c-31904b_1^3b_2^5c+58680b_1^2b_2^6c-38808b_1b_2^7c+8064b_2^8c-64b_1^6b_2c^2+96b_1^5b_2^2c^2+1216b_1^4b_2^3c^2-4120b_1^3b_2^4c^2+7632b_1^2b_2^5c^2-5056b_1b_2^6c^2+1016b_2^7c^2)\# 1^3+(-64b_1^6b_2^2+288b_1^5b_2^3+2400b_1^4b_2^4-20496b_1^3b_2^5+55512b_1^2b_2^6-52650b_1b_2^7+14364b_2^8-64b_1^6b_2c+}\\\bm{288b_1^5b_2^2c+2688b_1^4b_2^3c-23560b_1^3b_2^4c+69120b_1^2b_2^5c-70320b_1b_2^6c+20040b_2^7c-16b_1^6c^2+1344b_1^4b_2^2c^2-3280b_1^3b_2^3c^2+2592b_1^2b_2^4c^2-4020b_1b_2^5c^2+1424b_2^6c^2)\# 1^4+(576b_1^4b_2^3-6864b_1^3b_2^4+33696b_1^2b_2^5-56628b_1b_2^6+24696b_2^7+576b_1^4b_2^2c-7152b_1^3b_2^3c+37584b_1^2b_2^4c-70380b_1b_2^5c+33360b_2^6c+432b_1^4b_2c^2-1560b_1^3b_2^2c^2-5184b_1^2b_2^3c^2+10224b_1b_2^4c^2-2616b_2^5c^2)\# 1^5+(-576b_1^3b_2^3+7776b_1^2b_2^4-}\\\bm{32832b_1b_2^5+27396b_2^6-576b_1^3b_2^2c+7776b_1^2b_2^3c-36720b_1b_2^4c+34272b_2^5c-432b_1^3b_2c^2-2916b_1^2b_2^2c^2+15876b_1b_2^3c^2-9729b_2^4c^2)}\\\bm{\# 1^6+(-108b_2^2(2b_2+3c)(36b_1b_2-76b_2^2-18b_1c+29b_2c))\# 1^7+}\\\bm{(972b_2^2(2b_2-c)(2b_2+3c))\# 1^8\,\&,\,5];}$

$\bm{h=\frac{9b_2^2-16b_2c-6\delta^*(b_2+2c)-9(\delta^*)^2}{6(4b_2+3\delta^*)}}\\\bm{+\frac{3(b_2+\delta^*)\sqrt{9b_2^2+16b_2c+6\delta^*(3b_2+2c)+9(\delta^*)^2}}{6(4b_2+3\delta^*)};}$

$\bm{\delta_2=\frac{b_2^2/2+(2b_2-c-3h/2)\delta^*}{(3(h+\delta^*)/2+c)};}$

$\bm{\mbox{Reduce}[\delta_2-b_2/3\leq 0\,\&\&\,0\leq b_2<b_1<1.5b_2\,\&\&\,b_2\leq c<2b_2,\,\{b_2,b_1,c\}]}$

{\em Mathematica Output:}

$b_2>0\,\&\&\,b_2<b_1<1.5b_2\,\&\&\,b_2\leq c<2b_2$\\

\item We now verify if $\delta_1\leq b_1/3$. We use $\delta_1=\delta^*+\frac{b_1b_2-2b_2\delta^*-b_2^2/2}{3h/2+c}$ from the statement of Theorem \ref{thm:menu-3}.

{\em Mathematica Input:}

$\bm{\delta_1=\delta^*+\frac{b_1b_2-2b_2\delta^*-b_2^2/2}{3h/2+c};\mbox{Reduce}[\delta_1-b_1/3\leq 0\,\&\&\,0\leq b_2<b_1<1.5b_2\,\&\&\,b_1\leq c<2b_2,\,\{b_2,b_1,c\}]}$

{\em Mathematica Output:}

$b_2>0\,\&\&\,b_2<b_1<1.5b_2\,\&\&\,b_1\leq c<2b_2$\\

\item We now verify the monotonicity of $q$, i.e., verify if the left-hand side of (\ref{eqn:fig-c-third}) is nonnegative when $c\in[b_2,\alpha_2]$.\\

{\em Mathematica Input:}

$\bm{\mbox{Reduce}[(b_2^2+4b_2\delta^*-3\delta^*h)(b_2^2+4b_2\delta^*-2c\delta^*-3\delta^*h)-2b_1b_2(b_2^2+4b_2\delta^*-2c(\delta^*+h)-3h(2\delta^*+h))\geq 0\,\&\&\,0\leq b_2\leq b_1\leq 1.5b_2\,\&\&\,b_2\leq c<2b_2,\{b_2,b_1,c\}]}$\\

{\em Mathematica Output:}

$b_2>0\,\&\&\,b_2\leq b_1\leq 1.5b_2\,\&\&\,b_2\leq c\leq\mbox{ Root}[f_{c-III}(c)\,\&,\,3]$\\

Here, $f_{c-III}(c)$ is a humongous polynomial running for several pages. Define $\alpha_2:=\mbox{ Root}[f_{c-III}(c)\,\&,\,3]$. Then this proves that the left-hand side of (\ref{eqn:fig-c-third}) is nonnegative for every $c\in[b_2,\alpha_2]$.

\item We finally verify the bounds on $\alpha_2$. We again find the values of $c$ for which the left-hand side of (\ref{eqn:fig-c-third}) is nonnegative, but with $c$ restricted to $[2(t-1.36)(b_1-b_2)+1.36b_2,2(t-1.4)(b_1-b_2)+1.4b_2]$.

{\em Mathematica Input:}

$\bm{t=\frac{3}{176}\left(37+3\sqrt{465}\right);\mbox{Reduce}[(b_2^2+4b_2\delta^*-3\delta^*h)(b_2^2+4b_2\delta^*-2c\delta^*-3\delta^*h)-2b_1b_2(b_2^2+4b_2\delta^*-2c(\delta^*+h)-3h(2\delta^*+h))\geq 0\,\&\&\,0\leq b_2\leq b_1\leq 1.5b_2\,\&\&\,2(t-1.36)(b_1-b_2)+1.36b_2\leq c\leq 2(t-1.4)(b_1-b2)+1.4b_2,\{b_2,b_1,c\}]}$\\

{\em Mathematica Output:}

$b_2>0\,\&\&\,b_2\leq b_1\leq 1.5b_2\\\&\&\,0.746758b_1+0.613242b_2\leq c\leq\mbox{ Root}[f_{c-III}(c)\,\&,\,3]$\\

This completes all the proofs in Theorem \ref{thm:menu-3}.

\end{enumerate}

\section*{References}
\bibliographystyle{elsarticle-harv}
\bibliography{unit-demand}

\end{document}